\documentclass[11pt,a4paper]{article}
\usepackage{geometry} 
\usepackage{amsmath,amsfonts,amsthm,amssymb} 
\usepackage[dvipsnames]{xcolor}  
\usepackage{enumitem} 
\usepackage{graphicx, float, subfig, overpic, comment} 
\usepackage{stmaryrd} 
\usepackage{faktor} 
\usepackage{mathrsfs,mathtools} 
\usepackage{titlesec} 
\usepackage[none]{hyphenat} 
\usepackage[font=small,labelfont=bf,tableposition=top]{caption} 
\usepackage{authblk} 
\usepackage{todonotes,color}
\usepackage{comment}

\usepackage{hyperref} 
\hypersetup{
	colorlinks=true,
	linkcolor=blue,
	citecolor=green,
}

\numberwithin{equation}{section}

\theoremstyle{plain}
\newtheorem{lemma}{Lemma}[section]
\newtheorem{corollary}[lemma]{Corollary}
\newtheorem{proposition}[lemma]{Proposition}
\newtheorem{remark}[lemma]{Remark}

\newtheorem{theorem}[lemma]{Theorem}

	\newcommand{\be}{\begin{equation}}
	\newcommand{\ee}{\end{equation}}
	\newcommand{\bes}{\begin{equation*}}
	\newcommand{\ees}{\end{equation*}}
	\newcommand{\bpm}{\begin{pmatrix}}
	\newcommand{\epm}{\end{pmatrix}}
	\newcommand{\qtext}[1]{\quad \mbox{ #1 } \quad}
	
	\newcommand{\dron}[2]  {\frac{\partial#1   }{\partial#2}}
	
	\newcommand{\modu}[1]{\left\vert#1 \right\vert}
	
	\newcommand\eps   {{\varepsilon}}
	
	\makeatletter
	\newcommand*\bigcdot{\mathpalette\bigcdot@{.6}}
	\newcommand*\bigcdot@[2]{\mathbin{\vcenter{\hbox{\scalebox{#2}{$\m@th#1\bullet$}}}}}
	\makeatother

	\DeclareTextSymbol{\degre}{OT1}{23}
	
		\newcommand\cA{{\mathcal A}}

		\newcommand\cB{{\mathcal B}}
		
		\newcommand\cBt{{\tilde{\cB}}}
			
		\newcommand\Dt  {{\tilde D }}

		\newcommand\rd {{\mathrm d }}

		\newcommand\rH  {{\mathrm H }}
		\newcommand\rHt  {\tilde{\rH }}

		\newcommand\cH  {{\mathcal H }}

		\newcommand\rK  {{\mathrm K }}

		\newcommand\rM  {{\mathrm M }}

		\newcommand\cO{{\mathcal O} }
		
		\newcommand\Qt  {\tilde{ Q }}
		
		\newcommand{\RR}{{\mathbb{R}}}
		
		\newcommand\rR  {{\mathrm R }}

		\newcommand\ry{{\mathrm y}}

		\newcommand\Gam   {\Gamma  }

		\newcommand\lam   {\lambda }

		\newcommand\psit{\tilde{\psi}}

\def\UDelaunay{\Upsilon^{\mathrm{Del}}}
\def\UPoincare{\Upsilon^{\mathrm{Poi}}}
\def\OmegaM{\Omega_{\mathrm{M}}}
\def\iM{i_{\mathrm{M}}}

\def\al{\alpha}

\def\eps{\varepsilon}

\def\phi{\varphi}

\def\RR{\mathbb R}

\def\DDD{\mathcal D}

\def\GG{\mathcal G}

\def\RRR{\mathcal R}

\def\TTT{\mathcal T}

\def\OO{\mathcal O}

\def\cH{\mathcal H}

\def\HH{\mathcal H}
\def\~{\tilde}

\def\de{\delta}

\def\eps{\varepsilon}
\def\pa{\partial}

\def\sl{m}
\def\gg{s}
\def\CP{\mathrm{CP}}
\def\AV{\mathrm{AV}}
 
\title{On the role of the fast oscillations in the secular dynamics of the lunar coplanar perturbation on Galileo satellites} 
\author[1]{Elisa Maria Alessi}
\author[2,5]{Inmaculada Baldom\'a}
\author[3]{Mar Giralt}
\author[4,5]{Marcel Guardia}
\author[1]{Alexandre Pousse}
\affil[1]{IMATI-CNR, Istituto di Matematica Applicata e Tecnologie informatiche ``E. Magenes'', Consiglio Nazionale delle Ricerche, Via Alfonso Corti 12, 20133 Milano, Italy}
\affil[2]{Departament de Matem\`atiques \& IMTECH, Universitat Polit\`ecnica de Catalunya, Diagonal 647, 08028 Barcelona, Spain}
\affil[3]{IMCCE, CNRS, Observatoire de Paris, Universit\'e PSL, Sorbonne Universit\'e, 77
Avenue Denfert-Rochereau, 75014 Paris, France}
\affil[4]{Departament de Matem\`atiques i Inform\`atica, Universitat de Barcelona, Gran Via, 585, 08007 Barcelona, Spain}
\affil[5]{Centre de Recerca Matem\`atica, Campus de Bellaterra, Edifici C, 08193 Barcelona, Spain}

\date{\today}

\begin{document}

\maketitle 

 \renewcommand{\thefootnote}{\roman{footnote}}
\footnotetext[0]{\emph{E-mail adresses:} 
	\href{mailto:em.alessi@mi.imati.cnr.it}{em.alessi@mi.imati.cnr.it} (E.M. Alessi),
	\href{mailto:immaculada.baldoma@upc.edu}{immaculada.baldoma@upc.edu} (I. Baldom\'a), 
	\href{mailto:mar.giralt@obspm.fr}{mar.giralt@obspm.fr} (M. Giralt),
	\href{mailto:guardia@ub.edu}{guardia@ub.edu} (M. Guardia).}
\renewcommand{\thefootnote}{\arabic{footnote}}

\tableofcontents

\begin{abstract}

Motivated by the practical interest in the third-body perturbation as a natural cleaning mechanism for high-altitude Earth orbits, we investigate the dynamics stemming from the secular Hamiltonian associated with the lunar perturbation, assuming that the Moon lies on the ecliptic plane. The secular Hamiltonian defined in that way is characterized by two timescales. We compare the location and stability of the fixed points associated with the secular Hamiltonian averaged with respect to the fast variable with the corresponding periodic orbits of the full system. Focusing on the orbit of the Galileo satellites, it turns out that the two dynamics cannot be confused, as the relative difference depends on the ratio between the semi-major axis of Galileo and the one of the Moon, that is not negligible. 
The result is relevant to construct rigorously the Arnold diffusion mechanism that can drive a natural growth in eccentricity that allows a satellite initially on a circular orbit in Medium Earth Orbit to reenter into the Earth's atmosphere.
\end{abstract}

\section{Introduction}

The third-body gravitational perturbation on a bounded orbit at the Earth is known to yield a long-term variation in eccentricity (e.g., \cite{Chao2005}). In the past years, a specific interest on this phenomenon has arisen, as a possible mechanism to facilitate the disposal of artificial satellites at the end-of-life and dilute the probability of collision (e.g., \cite{Jenkin2005}). Special emphasis in this regard is given to high-altitude orbits, in particular to Medium Earth Orbits (MEO), where the Global Navigation Satellite Systems (i.e., GPS, Galileo, Beidou and GLONASS) are located.  

Considering an equatorial geocentric reference system, the eccentricity growth has been investigated mostly associated with resonances that involve the argument of pericenter and the longitude of ascending node of the orbit of the satellite and the longitude of the ascending node of the orbit of the third body, namely, Moon or Sun. Pillar works on the identification of these resonances, the corresponding phase space and their practical role on high-altitude orbits are \cite{Musen1961,gC62,Hughes2000,sB01,Rossi2008, Daquin2016}. Several numerical studies (e.g., \cite{Alessi2014, Radtke2015, Pellegrino2021}) have highlighted that indeed there exist regions in the phase space such that a circular orbit can become very eccentric in the limit to reenter to the Earth's atmosphere without the need of any propulsion system. However, from the theoretical point of view, a firm conclusion on how to explain analytically the numerical results is still missing. 

Motivated by this need, in \cite{AlessiBGGP23} we show how an Arnold diffusion mechanism can be built to explain the phenomenon. That work lies in the same context of a series of works that are based on the analysis of the
Kaula's expansion of the third-body perturbation~\cite{Kaula1962}, that is averaged over the mean anomaly of the orbit of the spacecraft and the one of the orbit of the third body in order to bring out the secular dynamics. We can say that these studies follow two main different, but related, concepts: 1.) the eccentricity growth is due to a chaotic behavior because more than one resonance is important in a given region \cite{Rosengren2015, Daquin2016}; 2.) the growth is due to a single resonance, whose phase space is modified by other terms of the Kaula's expansion \cite{Gkolias2019, Daquin2022, Legnaro2023}. 
 
 In \cite{AlessiBGGP23} we follow the latter, by considering the full quadrupolar secular expansion of the lunar perturbation. The aim is to show how to construct an Arnold diffusion mechanism, starting from the analysis of an autonomous reduced model. That is, we analyze first what we call the  ``coplanar Hamiltonian'' (namely, assuming that the Moon lies in the ecliptic plane) and use it as a first order for the full model. Then, we analyze the secular Hamiltonian perturbatively taking  the inclination of the Moon as a small parameter. Note that Arnold diffusion can only take place for Hamiltonians systems of at least 3 degrees of freedom and therefore the coplanar Hamiltonian has too low dimension. However, it is a good first order to detect the hyperbolic structures that enhance the drift of eccentricity for the full secular Hamiltonian.
 
 
  The main difference between our work and previous ones is that the first order that we consider does not neglect the fast oscillations of the coplanar dynamics. Let us explain what we mean by that.

  As already mentioned, the mechanism considered in \cite{AlessiBGGP23} to achieve eccentricity growth is to drift along a resonance. More concretely we consider the $2:1$ resonance between the argument of the pericenter and the longitude of ascending node of the orbit of the satellite. Then, the model has two timescales: the slow one of the resonant angle and the fast one which is given by the oscillations of the argument of the pericenter or of the longitude of the ascending node.
  
  
  One can consider as an effective model the coplanar Hamiltonian 
  averaged over the fast oscillations along the resonance, which is integrable. This averaged model has been widely analyzed in literature (see \cite{Daquin2022, Legnaro2023, Legnaro2023b}). It already presents hyperbolic orbits and invariant manifolds which are the ``highways'' to achieve eccentricity growth.

  The purpose of the present paper is to analyze ``how good'' this averaged  model is. That is how the dynamics, and in particular, the hyperbolic invariant objects in phase space of the coplanar Hamiltonian deviate from those of the averaged coplanar model. In other words, how the fast oscillations in the coplanar Hamiltonian affect its global dynamics.

More concretely:
\begin{enumerate}
\item We analyze rigorously the dynamics of  the averaged Hamiltonian. It turns out that one can describe analytically in a rather precise way the location and character of the critical points for a rather large range of semi-major axis (which includes that of Galileo). Of major importance for \cite{AlessiBGGP23} are the critical points which lie at circular motion. 
\item We explain the need of overcoming the averaged approximation, that, as far as we understood, lay the foundations of the works \cite{Daquin2022, Legnaro2023, Legnaro2023b}. 
The main result that we will present  in this direction is to show that, at a fixed energy level, the distance between hyperbolic periodic circular orbits of the averaged coplanar system and that of the full coplanar Hamiltonian is proportional to $a^3/a_\rM^3$ with $a$ the semi-major axis of the orbit and $a_M$ the semi-major axis of the Moon. This is done through a perturbative analysis, that is, assuming $a^3/a_\rM^3$ small enough. Note, however, that for the semimajor axis of Galileo $a^3/a_\rM^3$ is not so small. That is, for this values one should expect that the dynamics of the averaged Hamiltonian is rather ``far'' from that of the full model even at short timescales. In particular, the inclination of the  periodic orbits presents a significant difference.


\item We analyze the eigenvalues of the circular periodic  orbits of the averaged Hamiltonian and show that they are hyperbolic in a small interval, parabolic at the endpoints of this interval and elliptic otherwise. We perform the same analysis for the coplanar model by carrying out  a perturbative analysis in $a^3/a_\rM^3$ of the eigenvalues of the periodic orbits. 
\end{enumerate}

The development presented is focused on the orbit of the Galileo constellation, but the results are general and can be applied to other cases.

The paper is structured as follows. In Section \ref{sec:secularHam} we introduce the models that we consider and a suitable set of coordinates to analyze it. We follow closely the exposition of the companion paper~\cite{AlessiBGGP23}. In Section~\ref{sec:mainresults}, we state the main results of the paper. In Section~\ref{sec:mainresults:critpoints} those concerning the existence and character of the critical points of the averaged system. Then, in Section~\ref{sec:mainresults:PO}, we present the results concerning the circular periodic orbits of the coplanar Hamiltonian emanating from the  critical points of the averaged model. Sections~\ref{app:AveragedHam} and~\ref{sec:eccentric_case} are devoted to prove, respectively, the results concerning to the circular and non-circular (eccentric) equilibrium points of the averaged system. Finally, in Section~\ref{appendix_monodromy} we prove the results related with the periodic orbits for the coplanar Hamiltonian system.

\section*{Acknowledgements}

E.M. Alessi and A. Pousse acknowledge support received by the project entitled ``{co-orbital motion and
three-body regimes in the solar system}",
funded by Fondazione Cariplo through the program ``{Promozione dell'attrattivit\`a e competitivit\`a
dei ricercatori su strumenti
dell'European Research Council -- Sottomisura rafforzamento}".

I. Baldom\'a has been supported by the grant PID-2021-
122954NB-100 funded by the Spanish State Research Agency through the programs
MCIN/AEI/10.13039/501100011033 and “ERDF A way of making Europe”.

M. Giralt has been supported by the European Union’s Horizon 2020 research and innovation programme under the Marie Sk\l odowska-Curie grant agreement No 101034255.
M. Giralt has also been supported by the research project PRIN 2020XB3EFL ``Hamiltonian and dispersive PDEs".

M. Guardia has been supported by the European Research Council (ERC) under the European Union's Horizon 2020 research and innovation programme (grant agreement No. 757802). M. Guardia is also supported by the Catalan Institution for Research and Advanced Studies via an ICREA Academia Prize 2019. This work is part of the grant PID-2021-122954NB-100 funded by MCIN/AEI/10.13039/501100011033 and ``ERDF A way of making Europe''. 

This work is also supported by the Spanish State Research Agency, through the Severo Ochoa and María de Maeztu Program for Centers and Units of Excellence in R\&D (CEX2020-001084-M).

\section{From the Secular to the Averaged Hamiltonians}\label{sec:secularHam}


Let us consider a spacecraft that moves under the gravitational attraction of the Earth, the perturbation due to the Earth's oblateness and the lunar gravitational perturbation. Assuming a geocentric equatorial reference system,  the orbital elements, namely, semi-major axis $a$, eccentricity $e$, inclination $i$, longitude of the ascending node $\Omega$ and argument of pericenter $\omega$, that define the ellipse where the motion occurs, change in time because of the perturbations. The orbit of the Moon is defined  in the geocentric ecliptic reference system by the corresponding orbital elements $(a_\rM, e_\rM, i_\rM, \Omega_\rM, \omega_\rM)$, where 
\[
a_\rM = 384400 \,\mbox{km}, \qquad  e_\rM =0.0549006, \qquad i_\rM		=	5.15^{\circ},
\]
while the longitude of the ascending node of the Moon with respect to the ecliptic plane varies approximately linearly with time in a period $T_{\Omega_\rM}$ of 1 Saros (about 6585.321347 days \cite{aR73, eP91}) due to the solar gravitational perturbation, namely,
\begin{equation}\label{eq:OmegaM}
    \Omega_\rM(t)=	\Omega_{\rM,0} + n_{\Omega_\rM}t, \quad \quad \quad n_{\Omega_\rM}=2\pi/T_{\Omega_\rM}.
\end{equation}

In Delaunay action-angle variables $(L,G,H,\ell,g,h)$, see, for instance, \cite{Szebehely},
 considering as small parameter of the problem 
 \[\alpha	=	\frac{a}{a_\rM},\]
		which characterizes the distance of the Moon with respect to the orbit of the satellite,
  the secular dynamics, see \cite{Daquin2016}, is described by the non-autonomous (recall~\eqref{eq:OmegaM}) Hamiltonian 
\begin{equation}\label{def:fullHamiltonian}
\mathtt{H}(L,G,H, g, h, \OmegaM;\iM)=    \rH_\rK(L)	
		+	\rHt_{0}(L, G, H) +	\alpha^3 \rHt_1(L,G,H, g, h, \OmegaM;\iM),	
\end{equation}
where each term is defined as follows.

The first term, 
\begin{equation*} 
\rH_\rK(L)=	-\frac{1}{2}\frac{\mu^2}{L^2},
\end{equation*}
 is the constant term associated with the Earth's monopole, being 
 \begin{equation}\label{def:mu}
 \mu=398600.44 \,\mbox{km}^3/\mbox{s}^2
 \end{equation}
 the mass parameter of the Earth.
 
 The second term, 
 	\be \label{def:initialHamiltonians_J2}
			\rHt_0(L, G, H)		=		\frac{1}{4}\frac{\rho_0}{L^3} \frac{G^2	-3H^2}{G^5},
	\ee
  is the perturbative term associated with the Earth's oblateness,  averaged over the orbital period of the spacecraft, being $\rho_0		=	\mu^4J_2 R^2$, with $J_2 = 1.08 \times 10^{-3}$ the coefficient of the second zonal harmonic in the geopotential and
		$R = 6378.14\,\mbox{km}$ the mean equatorial radius of the Earth. 
  
  Finally, the (truncated) secular perturbative term due to the Moon (see \cite{Kaula66, Hughes1980}) corresponds to  
 \begin{equation}\label{def:H1tilde}	
 \begin{split}
		\rHt_1	&(L,G,H, g, h, \OmegaM; \iM)\\
				=&	-\frac{\rho_1}{L^2} 
					\sum_{m=0}^2	
						\sum_{p=0}^2 
							\Dt_{m,p}(L,G,H) \sum_{s=0}^2
									c_{m,s}F_{2,s,1}(i_\rM) 
							\\
			&\times
							\left[
								U_2^{m,-s}(\epsilon)\cos\left(\psit_{m, p, s}(g,h,\OmegaM)\right)
							+	U_2^{m,s}(\epsilon)\cos\left(\psit_{m, p, -s}(g,h,\OmegaM)\right)
							\right].
		\end{split}
	\end{equation}
 In the expression above, 
 \begin{itemize}
     \item 	the function $U_2^{m, \mp s}(\epsilon)$ corresponds to the Giacaglia function (see Tab.~\ref{tab:U} in Appendix~\ref{appendix:HamiltonianDelaunay}) being $\epsilon =	23.44^{\circ}$ the obliquity of the ecliptic with respect to the equatorial plane.
\item One has that
	\begin{equation}\label{def:Dtilde}
	\Dt_{m,p}(L, G,H)	=	\tilde{F}_{m,p}(G,H)\tilde{X}_{p}(L,G),
	\end{equation}
 	with
  \[
  \begin{split}
  \tilde{F}_{m,p}(G,H)= F_{2,m,p}\left (\arccos\frac{H}{G} \right ),\quad
  \tilde{X}_p(L,G) = X_0^{2,2-2p}\left(\sqrt{1-\frac{G^2}{L^2}} \right),
  \end{split}
  \]
  where $F_{2, m, p}(i)$ are the Kaula's inclination functions (see~\cite{Kaula66}),  $X_0^{2,2-2p}(e)$
  are the zero-order Hansen coefficients (see~\cite{Hughes1980}) and we have used the definitions of  inclination and eccentricity in terms of Delaunay coordinates, that is 
\begin{equation}\label{def:GH}
H=G\cos i\qquad \text{and}\qquad G=L\sqrt{1-e^2}.
\end{equation}
  \item The angle 
 \begin{equation}\label{def:psitilde}
	\psit_{m,p, s}(g,h, \OmegaM)=	2(1 -p)g	
										+	m h 	+ s\left(\Omega_\rM	-	\frac{\pi}{2}\right)	
										-	\ry_{\modu{s}}\pi
	\end{equation}
 defines the relative orientation satellite-Moon.
\item The other  constants are defined as
\begin{equation}\label{def:auxiliaryFunctionsOriginal}
	\begin{aligned}
		\rho_1 	&= &	\frac{\mu\mu_\rM}{(1	-	e_\rM^2)^{3/2}},\quad \quad \quad  \quad \quad  \quad  \quad  \quad
		&\eps_n	&=& 	\left\{
		\begin{array}{ll}
			1		&	\mbox{if $n=0$}\\
			2		&	\mbox{if $n\neq 0$}
		\end{array}\right.,\\
		c_{m,s}	 & = &	(-1)^{\lfloor m/2\rfloor} \frac{\eps_m\eps_s}{2}\frac{(2-s)!}{(2 + m)!},\quad \quad \quad ~
		&\ry_s	&=&	\left\{
		\begin{array}{ll}
			0		&	\mbox{if $s$ is even}\\
			1/2		&	\mbox{if $s$ is odd}
		\end{array}\right. ,
	\end{aligned}
\end{equation}
where $\mu_\rM=	4902.87 \,\mbox{km}^3/\mbox{s}^2$ is the mass parameter of the Moon.
\end{itemize}

\begin{remark}\label{rem:auton}
By Kaula's inclination functions (see \cite{Kaula66}), one obtains that
\begin{align*} 
F_{2,s,1}(0) = \left\{\begin{array}{ll}-\frac12	
   &	\mbox{if $s=0$,}\\[0.5em]
	0			
 & 	\mbox{if $s=1,2$.}
\end{array}\right.
\end{align*}
That is, assuming that the Moon lies on the ecliptic plane (i.e., $\iM=0$), the Hamiltonian $\mathtt{H}$ in \eqref{def:fullHamiltonian} is autonomous, because the angle $\psit_{m,p,s}(g,h, \OmegaM)$ does not depend on $\OmegaM$ for $s=0$ (recall \eqref{eq:OmegaM}). 
\end{remark}

\begin{remark}\label{rmk:units}
	In the numerical computations that will be presented throughout the text, all the variables will be taken in non-dimensional units defined in such a way that the semi-major axis (equal to 29600 km, because we focus on Galileo) is the unit of distance and the corresponding orbital period is 2$\pi$.
\end{remark}

In what follows, we will omit the Keplerian term of the Hamiltonian $\rH_{\rK}$, since it does not contribute to the variation of the orbital elements.
Similarly, we omit the dependence of the Hamiltonian on the variable $L$, since it is a constant of motion.

\subsection{A good system of coordinates
}

The Galileo constellation orbits at an inclination, $i\approx 56^\circ$, such that the dominant term in the lunar perturbation is the $2g+h$ resonance.
This resonance is defined  by the unperturbed  Hamiltonian $ \rHt_0$ in \eqref{def:initialHamiltonians_J2} (i.e., $\alpha = 0$) and the condition  $2\dot{g}+\dot h= 0$. Namely, $2\pa_g \rHt_0+\pa_h\rHt_0=0$.
This is satisfied provided that 
\begin{equation*}
5H^2 - G^2 - HG	=	0,
\end{equation*}
that is, for all $G\neq 0$ (i.e., $e<1$) and $H=G\cos i_\star$ that satisfy
\begin{equation*}
5 X^2 - 1 - X = 0, \quad \quad \quad X = \cos i_\star.
\end{equation*}
Hence, we can distinguish two situations: 
\begin{itemize}
\item the prograde case for $i_{\star}	=	\arccos\left(\frac{1 + \sqrt{21}}{10}\right) \simeq 56,06^{\circ}$,
\item the retrograde case for $i_{\star}	=	\arccos\left(\frac{1 - \sqrt{21}}{10}\right)	\simeq 110,99^{\circ}$.
\end{itemize}
In what follows, $i_\star$ will be mentioned as the inclination of the ``exact $2g+h$''-resonance and we will focus on the prograde case.

The unperturbed Hamiltonian highlights that $x=2g+h$ is constant when $i =i_\star$, while it circulates for $i\neq i_\star$.
Moreover, in a small enough neighborhood of the exact resonance  
the angular variables evolve at different rates: $g$ and $h$ are ``fast'' angles compared to $x$ which undergoes a ``slow'' drift of order $\cO(i-i_\star)$.
 
\paragraph{Slow-fast Delaunay coordinates $(y,x)$.}
Instead of using the Delaunay action-angle variables, in order to take advantage of the timescales separation, we introduce the symplectic transformation
$  
{(G, H, g, h)=\UDelaunay(y, \Gam, x, h) }
$
where
\begin{equation}\label{eq:xy}  
y = \frac{G}2, \qquad \qquad \Gam = H-\frac{G}2, \qquad \qquad x = 2g + h.	
\end{equation}
Note that, the resonant action $y$ does not depend on the inclination.
Hence, the action-angle variables $(y,x)$ are associated with the variation in eccentricity.

In slow-fast Delaunay variables, the secular Hamiltonian  can be written as
\begin{equation*}
\rH(y,\Gam, x, h, \OmegaM;\iM) = \rH_{0}(y, \Gam)+\alpha^3 \rH_1 (y,\Gam, x, h, \OmegaM; \iM)
\end{equation*}
where
\begin{equation}\label{def:hamiltonianDelaunayH1}
\begin{split}
\rH_0(y , \Gam)	&=	(\rHt_0 \circ \UDelaunay)(y,\Gam)= \frac{\rho_0}{128}\frac{y^2  -6y\Gam	 - 3\Gam^2}{L^3y^5} \\
\rH_1(y , \Gam, x, h, \OmegaM;\iM)&=(\rHt_1 \circ \UDelaunay) (y , \Gam, x, h, \OmegaM;\iM).
 \end{split}
\ee
 
In  the coordinates just defined, the $2g + h$–resonance becomes the $x$–resonance, which is defined by
\bes	
\omega(y, \Gam)	=\dron{\rH_0(y, \Gam)}{y}	= \frac{3}{128}\frac{\rho_0}{L^3}\frac{   5\Gam^2 + 8y\Gam - y^2}{y^6} = 0.
\ees
This resonance takes place at the two lines
\begin{equation}\label{def:resonance}
 \Gam=y\frac{-4 + \sqrt{21}}{5} \qquad \text{and}\qquad \Gam	=	y\frac{-4 - \sqrt{21}}{5},
\end{equation}
with $(y, \Gam) \neq (0,0)$, that are associated with prograde and retrograde orbits, respectively.

\paragraph{Poincar\'e coordinates $(\eta,\xi)$.}
The slow-fast variables~\eqref{eq:xy} derived from the Delaunay variables (as happens with the original Delaunay variables) are singular at $e=0$. Thus, in order to study the dynamics of the secular Hamiltonian in a neighborhood of circular orbits, i.e., $0<e\ll 1$, we introduce the  Poincar\'e coordinates
$
{ (y, \Gam, x, h) = \UPoincare(\eta, \Gam, \xi, h) }
$
where
\begin{equation*}
\xi	=\sqrt{2L-4y}\cos \left(\frac{x}2\right), \qquad \qquad \qquad
\eta=\sqrt{2L	-	4y}\sin \left(\frac{x}2\right),
\end{equation*}
which are symplectic. Notice that $\xi$ and $\eta$ are respectively equivalent to $e\cos (x/2)$ and $e\sin (x/2)$ for quasi-circular orbits.

In these coordinates, the secular Hamiltonian becomes
\begin{equation}\label{def:Ham:H:Poinc}
\HH(\eta,\Gam, \xi, h, \OmegaM;\iM) = \HH_{0}(\eta, \Gam,\xi)	+	\alpha^3 \HH_1(\eta,\Gam, \xi, h, \OmegaM;\iM)
\end{equation}
where
\begin{align}
\HH_{0}(\eta, \Gam,\xi)
&=
(\rH_0 \circ \UPoincare)(\eta, \Gam,\xi) \label{def:hamiltonianPoincareH0}\\
&=
\frac{\rho_0}{2} 
\frac{(2L-\xi^2-\eta^2)^2 - 24(2L-\xi^2-\eta^2) \Gam - 48\Gam^2}{L^3(2L-\xi^2-\eta^2)^5}\notag\\
	\HH_1(\eta , \Gam, \xi, h, \OmegaM;\iM) &	=	(\rH_1 \circ \UPoincare) (\eta , \Gam, \xi, h, \OmegaM;\iM).\label{def:hamiltonianPoincareH1}
\end{align}

\subsection{The hierarchy of models}\label{section:hierarchy}


In the present paper, we consider a hierarchy of models which stems from the secular Hamiltonian \eqref{def:Ham:H:Poinc}.
They are what we call the \emph{coplanar secular Hamiltonian} and the \emph{$h$-averaged (coplanar secular) Hamiltonian}.
To construct these ``intermediate models'', we rely on the fact that the model depends on two parameters: the inclination of the Moon with respect to the ecliptic plane,~$i_\rM$, and the semi-major axis ratio between the satellite and the Moon, $\al$.




In the companion paper~\cite{AlessiBGGP23}, using the same hierarchy of models, we construct Arnold diffusion orbits for the secular Hamiltonian \eqref{def:Ham:H:Poinc}, that can lead to a drastic increase of the eccentricity of the satellite. A crucial point in our construction is the analysis of certain hyperbolic structures (normally hyperbolic invariant manifolds, stable and unstable invariant manifolds) of the coplanar secular Hamiltonian which are persistent in \eqref{def:Ham:H:Poinc}. 
Comparing the coplanar model with the full one allows us to construct the drifting orbits. 


On the contrary, in the present paper we focus on the comparison between the coplanar and the $h$-averaged Hamiltonians. There is now an extensive literature (see, e.g.,  \cite{Daquin2022, Legnaro2023b}) analyzing the hyperbolic critical points of the  $h$-averaged Hamiltonian and its invariant manifolds as a sign of existence of unstable motions for the full Hamiltonian \eqref{def:Ham:H:Poinc}. In the present paper, we compare analytically certain dynamics of the $h$-averaged and coplanar models focusing both on what dynamics persist and which ``deviations'' has one model with respect to the other.

%

\paragraph{First reduction: the Coplanar Model.}

Since the inclination of the Moon with respect to the ecliptic plane is relatively small (that is, $i_\rM = 5.15^\circ$), the first reduction that one can do to have an intermediate model is to take $i_\rM=0$, which  corresponds to assume that the orbit of the Moon is coplanar to that of the Earth.

When doing this reduction, the Hamiltonian $\HH$ in \eqref{def:Ham:H:Poinc} becomes $\Omega_\rM$ independent and thus autonomous (see Remark~\ref{rem:auton}).
Starting from \eqref{def:Ham:H:Poinc}, we define the 
$2$ degrees of freedom Hamiltonian
\begin{equation}\label{def:Ham:Coplanar} 
\HH_\CP(\eta, \Gam, \xi, h)=\HH(\eta, \Gam, \xi, h, \Omega_M;0),
\end{equation}
where the subscript $\CP$ stands for coplanar. Since this Hamiltonian is autonomous, $\HH_\CP$ is a first integral of the system. Moreover, it can be written as
\[
\HH_\CP(\eta, \Gam, \xi, h)= \HH_{0}(\eta, \Gam,\xi)	+	\alpha^3 \HH_{\CP,1}(\eta,\Gam, \xi, h),
\]
where $\HH_0$ is the Hamiltonian introduced in \eqref{def:hamiltonianPoincareH0} and $\HH_{\CP,1}$ is the Hamiltonian $\HH_1$ in~\eqref{def:hamiltonianPoincareH1} with $i_\rM=0$, that is,
\begin{equation}\label{def:Hcoplanar:perturb}
\HH_{\CP,1}(\eta,\Gam, \xi, h)=
\HH_1(\eta,\Gam,\xi,h,\OmegaM;0).
\end{equation}
See Appendix~\ref{app:tables}, and in particular~\eqref{eq:expressionPoincareHCP1}, for the explicit expression of $\HH_{\CP,1}$.
	
\paragraph{Second reduction: the $h$-averaged problem.}

When $\al$ is also taken as a small parameter, the autonomous Hamiltonian $\HH_\CP$ in \eqref{def:Ham:Coplanar}  has a timescale separation between the slow and fast angles. Indeed,  
\begin{equation*}
\dot \eta,\dot \xi\sim \alpha^3 \qquad \text{whereas}\qquad \dot h\sim 1.
\end{equation*}
A classical way to exploit this feature is to simplify the Hamiltonian $\HH_\CP$  by averaging out the fast oscillations 
with respect to the longitude of the ascending node $h$. 
 That is,
\[	\HH_\AV(\eta,\Gam, \xi)	=	\frac{1}{2\pi} \int_0^{2\pi}\HH_\CP(\eta, \Gam, \xi, h)\rd h.
\]
Note that, since $\HH_0$ in~\eqref{def:hamiltonianPoincareH0} is $h$-independent, $\HH_\AV$ can be written as 
\begin{equation}\label{def:HamAV}
\HH_\AV(\eta,\Gam, \xi)	=	\HH_{0}(\eta,\Gam,\xi)	+	\alpha^3 \HH_{\AV,1}(\eta,\Gam, \xi),
\end{equation}
with
\begin{equation*}	
\HH_{\AV,1}	(\eta,\Gam, \xi)=\frac{1}{2\pi} \int_0^{2\pi}\HH_{\CP,1}(\eta,\Gam, \xi, h)\rd h.
\end{equation*} 
More specifically, the particular form  of $\HH_{\CP,1}$ in~\eqref{eq:expressionPoincareHCP1} implies that 
\begin{equation} \label{def:Ham:AV1} 
\begin{split}
\HH_{\AV,1}(\eta,\Gam, \xi)=\frac{\rho_1}{L^2} \bigg[&
\frac{1}{2}U_2^{0,0} \DDD_{0,1}(\eta,\Gam,\xi) \left(8L^2+12 L (\xi^2+\eta^2)-3(\xi^2+\eta^2)^2\right) \\
&+ 	\frac{1}{6}U_2^{1,0} \DDD_{1,0}(\eta,\Gam,\xi) (\xi^2-\eta^2)	
\bigg],
\end{split}
\end{equation}
where the constants $U_2^{m,0}$ and the functions $\DDD_{m,p}$ are given in Tables~\ref{tab:U} and~\ref{tab:functionsPoincare}, respectively, in  Appendix \ref{app:tables}. See also this appendix for the deduction of this expression.

To analyze the dynamics of the $h$-averaged Hamiltonian, it will be also convenient to consider it in the slow-fast Delaunay coordinates \eqref{eq:xy}. In these variables, it reads
\begin{equation}\label{eq:expansionsH1averagedDelaunay}
	\rH_{\AV}	(y,\Gam, x)	=\rH_{0}	(y,\Gam)	+\rH_{\AV,1}	(y,\Gam, x)	
	\end{equation}
where $\rH_0$ is the Hamiltonian introduced in \eqref{def:hamiltonianDelaunayH1},
\begin{equation*}
	\rH_{\AV,1}	(y,\Gam, x)	
	=	\frac{\rho_1}{L^2} \left[
	\frac{1}{2}U_2^{0,0} D_{0,1}( y ,\Gam) 
	+ 	\frac{1}{3}U_2^{1,0} D_{1,0}(y,\Gam)\cos x			
	\right],
\end{equation*}
and the functions $D_{m,p}$ are given in Table~\ref{tab:functionsDelaunay}.
We also refer to Appendix \ref{app:tables}, for the deduction of the expression of $\rH_{\AV,1}$.

\section{Main results}\label{sec:mainresults}

In this section, we present the main results of the paper,  which deal with both the $h$-averaged Hamiltonian and the coplanar Hamiltonian. They are stated respectively in Sections~\ref{sec:mainresults:critpoints} and~Section~\ref{sec:mainresults:PO}. 

\begin{remark}\label{rmk_unities}
Along the paper we also present some numerical results. As a rule, the variables will be taken in non-dimensional units defined in such a way that the semi-major axis of the satellite (equal to 29600 km) is the unit of distance and the corresponding orbital period is 2$\pi$. 

However, in some places, with the purpose of keeping the exposition clear, we have kept the standard unit system. It will be pointed out along the text. 
\end{remark}

The results for the $h$-averaged Hamiltonian are related with the existence and stability character of the stationary points. To be more precise,  
\begin{itemize}
\item We study the stability character of the origin of the $h$-averaged Hamiltonian  in Theorem~\ref{thm:origin_average}. Then, in Proposition~\ref{prop:origin_modulus_eigenvalue}, we provide some properties of the eigenvalues of the linearization of $\HH_\AV$ at the origin (in Poincar\'e variables).  
\item Theorem~\ref{thm:AveragedHam_generalcase} deals with the existence and the linear  stability behavior of the eccentric critical points  of the $h$-averaged Hamiltonian $\rH_\AV$. We state these results in Delaunay coordinates. 
\end{itemize}

These results are proven for a specific range of the semi-major axis $a$, namely, we prove the results for values of $a\in [a_{\min}, a_{\max}]$ with 
\begin{equation}\label{defaminmax}
a_{\min}=6378.14 \text{ km}\qquad\text{and}\qquad a_{\max}=30000 \text{ km}.
\end{equation}
The minimum value corresponds to a collision orbit with the Earth assuming $e=0$. Note also that this interval contains the semi-major axis of Galileo $a= 29600 \text{ km}$.

The coplanar Hamiltonian $\HH_\CP$ can be seen as a perturbation of the $h$-averaged Hamiltonian $\HH_\AV$. From this point of view and recalling that the origin $\xi=\eta=0$ is a fixed point of $\HH_\AV$, in Section~\ref{sec:mainresults:PO}, we state the results for the coplanar Hamiltonian $\HH_{\CP}$ related to the periodic orbits emanating from the origin. Some of these results are perturbative results with respect to the parameter $\delta$ defined as 
\begin{equation}\label{def:delta:mainresults}
    \delta = \rho \alpha^3 L^4, \qquad \mathrm{with}\qquad  \alpha= \frac{a}{a_{\rM}}, \qquad \rho = \frac{\rho_1}{\rho_0}
\end{equation}
with $\alpha, \rho_0,\rho_1$ introduced in Section~\ref{sec:secularHam}. 

Next, we study the existence and main properties of the periodic orbits of the coplanar Hamiltonian:
\begin{itemize}
    \item In Theorem~\ref{thm:existence_periodic_orbits}, we characterize the energy levels of the coplanar Hamiltonian $\HH_\CP$ containing a periodic orbit lying in $\{\xi=\eta=0\}$.    
    \item For these values of energy, $\mathbf{E}$, Corollary~\ref{cor:gammaCPAV} relates the $\Gamma$ value of the point of the periodic orbit at the energy level $\{\HH_\CP=\mathbf{E}\}$ at the transverse section $\{h=0\}$ with the corresponding parameter $\Gam$ such that the $h$-averaged Hamiltonian has the same energy level, namely $\{\HH_\AV = \mathbf{E}\}$. We see that, for small enough $\delta$, the two values are order $\delta$ separated. Since for Galileo $\delta\sim 0.11$ (computed in the unities specified in Remark~\ref{rmk_unities}), this shows that $h$-averaged model differs significantly from the coplanar one. At circular motions, deviations in $\Gamma$ are equivalent to deviations in (the cosine of the) inclination.
    \item Proposition~\ref{prop:firstaproximation_periodic} computes a first order approximation with respect to $\delta$ of the periodic orbits and their corresponding periods. 
\end{itemize}
Even when there is no physical meaning for values of $a<a_{\min}$, these results hold true for values of $a\in (0,a_{\max}]$.

The analysis of the stability character of the periodic orbits is the only one that needs the parameter $\delta$ to be small enough. 
\begin{itemize}
    \item Theorem~\ref{prop:eigenvalues_monodromy}, and its refinement Theorem~\ref{thm:eigenvalues_monodromy_more}, give an approximation, at least of order $\mathcal{O}(\delta^{3/2})$, of the eigenvalues of the monodromy matrix associated to the periodic orbits by means of the eigenvalues  of the linearized part at the origin for the $h$-averaged Hamiltonian. As a consequence we are able to elucidate, for small values of $\delta$, the character of the periodic orbits. 
\end{itemize}

\subsection{The critical points of the averaged system}\label{sec:mainresults:critpoints}
We now present the results about the critical points of the $h$-averaged Hamiltonian. We notice that, to simplify the exposition, we state some results referred to the $h$-averaged Hamiltonian in Poincar\'e variables (see~\eqref{def:HamAV}), and others referred to the Hamiltonian expressed in slow-fast Delaunay variables (see~\eqref{eq:expansionsH1averagedDelaunay}).

The specific range of the semi-major axis $[a_{\min},a_{\max}]$ is given in~\eqref{defaminmax}.
We also introduce the notation
\begin{equation}\label{def:Lalfadelta_max}
L_{\min}=\sqrt{\mu a_{\min}}, \qquad L_{\max}=\sqrt{\mu a_{\max}}, \qquad \alpha_{\max}= \frac{a_{\max}}{a_\rM}, \qquad \delta_{\max}= \rho \alpha_{\max}^3 L_{\max}^4,
\end{equation}
where $\mu$ has been introduced in~\eqref{def:mu}. Finally, we define  
\begin{equation}\label{def:pendentresonance}
\sl=\frac{-4+\sqrt{21}}{5} \sim 0.116515138991168,
\end{equation}
the slope of the resonance prograde line is $\Gamma=m y$ (see~\eqref{def:resonance}).

In addition, we notice that, since we want to study the prograde resonance we keep our analysis to $\Gam>0$ and, from the  expressions of $\DDD_{m,p}$ in Table~\ref{tab:functionsPoincare} , we deduce that 
$$0\leq \Gam \leq \frac{L}{2}.$$ 

The first result deals with the character of the origin in Hamiltonian~\eqref{def:HamAV}. Its proof is carried out in Section~\ref{sec:circular_case_Appendix}.

\begin{theorem}\label{thm:origin_average}
The origin is a critical point of the $h$-averaged Hamiltonian $\HH_{\AV}$, see~\eqref{def:HamAV}. 
Moreover, there exist two functions ${\Gamma}_{-}, {\Gamma}_{+}: [L_{\min},L_{\max}] \to \left (0,\frac{L}{2}\right ) $, such that the character of the critical point is described as follows for $L\in [L_{\min}, L_{\max}]$:
\begin{itemize}
    \item if either $\Gamma\in \big (0, \Gamma_-(L) \big )$ or $\Gamma\in \left (\Gamma_+(L), \frac{L}{2} \right )$, then $(0,0)$ is a center. 
    \item If $\Gamma\in \big (\Gamma_-(L),\Gamma_+(L) \big )$, then $(0,0)$ is a saddle.
    \item If $\Gamma=\Gamma_-(L)$ or $\Gamma=\Gamma_+(L)$, the origin is a degenerated equilibrium point with nilpotent linearization. 
\end{itemize}
In addition, $L^{-1} \Gam_\pm(L)$ can be expressed as
$$
\frac{\Gam_{\pm}(L)}{L} = \hat \Gam_{\pm}(\de),\qquad \de=\rho \alpha^3 L^4,
$$
where $\hat \Gam_{\pm} :[0, \delta_{\max}] \to \mathbb{R}$, see~\eqref{def:delta:mainresults} and~\eqref{def:Lalfadelta_max}, and there exists a constant $C_*$ such that for all $L\in [L_{\min},L_{\max}]$,  
\begin{equation}\label{thm:exprGamma12}
\left | \hat \Gam_{\pm}(\de) - \frac{\sl}{2} \mp \frac{5 U_{2}^{1,0}(\sl +3)\sqrt{3- 2 \sl -4 \sl ^2}}{{3(16+20\sl)}}  \de  \right | \leq C_*\de^2.
\end{equation} 
%
\end{theorem}

We denote by $\pm \lambda_{\AV}(\Gamma;L)$ the eigenvalues of the linearization around the origin of the Hamiltonian system associated to~\eqref{def:HamAV}. We assume that, in the saddle case, $\lambda_{\AV}(\Gamma;L)>0$ and in the center case, $-i \lambda_{\AV}(\Gamma;L)>0$. Notice that $\lambda_{\AV}(\Gamma_-(L);L)=\lambda_{\AV}(\Gamma_+(L);L)=0$ where $\Gam_-(L)$ and $\Gam_+(L)$ are the functions defined in Theorem~\ref{thm:origin_average}.
 
\begin{proposition}\label{prop:origin_modulus_eigenvalue}
For any $L\in [L_{\min},L_{\max}]$ we have the following.
\begin{itemize}
\item There exists a unique $\Gamma_*(L) \in \big (\Gamma_-(L), \Gamma_+(L)\big )$ such that for all
$\Gamma \in  
\big (\Gamma_-(L),\Gamma_+(L) \big )$, $
\lambda_\AV(\Gamma;L) \leq \lambda_\AV(\Gamma_*(L);L)$ and 
$$
 \frac{5 \rho_0}{16L^{7}}U_{2}^{1,0} \de(\sl+3) \sqrt{3-2\sl -\sl^2}\leq \lambda_\AV(\Gamma_*(L);L)\leq \frac{ 15 \sqrt{3}\rho_0}{16 L^7} U_2^{1,0} \de.
$$
\item For $\Gamma \in \big (0, \Gamma_-(L))$ 
the function $|\lambda_\AV(\Gamma;L)|$ is smooth and decreasing with respect to $\Gamma$ (and $L$ fixed).
\item For $\Gamma \in  \left (\Gamma_+(L), \frac{L}{2}\right )$, the function  $|\lambda_\AV(\Gamma;L)|$ is smooth and increasing with respect to $\Gamma$. 
\end{itemize}
In addition, there exists a function $\hat \lambda_\AV: \left [0,\frac{1}{2} \right ]\times [0,\delta_{\max}] \to \mathbb{C}$ such that 
$$
\lambda_\AV(\Gam;L)=\frac{3 \rho_0}{4 L^7 } \hat \lambda_\AV \left (\frac{\Gam}{L}; \delta\right ), \qquad \delta = \rho \alpha^3 L^4.
$$
\end{proposition}
\begin{remark}
The eigenvalues $\lambda_\AV(\Gam;L)$ can be explicitly computed, see Remark~\ref{rmk:first_expression_eigenvalues}. 
\end{remark}

This result is proven in Section~\ref{sec:modulus}. We also provide some numerical computations.

To finish with the results related to the equilibrium points of the $h$-averaged system, we state the following theorem concerning ``non-circular'' critical points, namely, those critical points of the $h$-averaged system~\eqref{eq:expansionsH1averagedDelaunay} (in Delaunay coordinates) satisfying $e\neq 0$. 
We consider slow-fast Delaunay variables and the Hamiltonian in~\eqref{eq:expansionsH1averagedDelaunay}.

Note that we are interested in a region of the phase space satisfying the following conditions:
\begin{itemize}
\item We want the satellite to be in a prograde ``elliptic regime''. Equivalently, we want the Delaunay coordinates to be well-defined and, therefore, we consider
\begin{equation}\label{def:DomainDelaunay}
(x,y)\in\mathbb{T}\times(0,L/2).
\end{equation}
\item We want the actions $y$ and $\Gamma$ to be close to the resonance (see, for instance, the first resonance  \eqref{def:resonance}), namely the prograde resonance, so that the $h$-averaged Hamiltonian is a good approximation of the coplanar model. For this reason, we consider
\begin{equation}\label{def:nbdResonance}
\frac{\Gamma}{y}\in \left (0,\frac{1}{2} \right ).
\end{equation}
We recall that the resonance prograde line is $\Gamma=m y$ with 
$\sl$ defined in~\eqref{def:pendentresonance} and we notice that, from  Table~\ref{tab:functionsDelaunay} the values of $\Gam,y$ are restricted to $\Gam<y$. 
\item To avoid collision, the semi-major axis $a$ and the eccentricity $e$ have to satisfy  $a(1-e)\geq R=a_{\min}$. Then $e\leq 1- \frac{R}{a}$. Moreover, since
$\frac{y}{L}= \frac{1}{2} \sqrt{1-e^2}$, we are only interested in values 
\begin{equation}\label{remark_col}
\frac{y}{L} \geq  \frac{1}{2} \sqrt{2\frac{R}{a} - \frac{R^2}{a^2}} \geq \frac{1}{2} \sqrt{2 \frac{R}{a_{\max}} - \frac{R^2}{a_{\max}^2} }=:\hat{y}_{\mathrm{col}}= 0.308224183446436.
\end{equation}
In conclusion, we consider values of $\frac{y}{L} \in \left [\hat{y}_{\mathrm{col}}, \frac{1}{2}\right ]$. 
\end{itemize}

\begin{theorem}\label{thm:AveragedHam_generalcase} 
There exist two constants\footnote{The value $L_*$ (in international unities) corresponds to a semi-major axis $a_* \sim 23893.56218133389\, \text{km}$.} 
$$
L_* \sim 97590.90325766560,\qquad \hat{\Gamma}_0\sim 0.013584391815073,
$$
two decreasing functions $\hat{\Gamma}_{1}(L)$, $\hat{\Gamma}_{2}(L)$ and one increasing function $\hat{\Gamma}_3(L)$ defined for $L\in [L_{\min},L_{\max}]$ such that the Hamiltonian~\eqref{eq:expansionsH1averagedDelaunay} has the following critical points   satisfying \eqref{def:DomainDelaunay}, \eqref{def:nbdResonance}, \eqref{remark_col}: 
\begin{itemize}
\item For $L\in [L_{\min},L_*]$ and $\Gamma\in (0,L/2]$, 
\begin{enumerate}
\item if $\Gamma \notin [ L \hat{\Gamma}_1(L), L \hat{\Gamma}_3(L)]$, has no critical points.
    \item If $\Gamma \in [L \hat{\Gamma}_1(L), L \hat{\Gamma}_0)$, has a unique critical point that is a center and of the form $(\pi,y)$.
    \item If $\Gamma \in [ L \hat{\Gamma}_0, L \hat{\Gamma}_2(L)]$, has only two critical points $(0,y_1), (\pi,y_2)$, which are of saddle and center type respectively, except when  $L=L_*$ and $\Gamma  =L_* \hat{\Gamma}_2(L_*)$ that $(\pi,y_2)$ is parabolic. 
    \item If $\Gamma \in ( L \hat{\Gamma}_2(L), L \hat{\Gamma}_3(L)]$, has a unique critical point that is a saddle of the form $(0,y)$.
\end{enumerate}
\item If $L\in (L_*,L_{\max}]$  and $\Gamma\in (0,L/2]$, there exists a function $\hat{\Gamma}_*(L) \in [\hat{\Gamma}_2(L), \hat{\Gamma}_3(L)]$ such that, 
\begin{enumerate}
\item if $\Gamma \notin [ L \hat{\Gamma}_1(L), L \hat{\Gamma}_3(L)]$, has no critical points.
    \item If $\Gamma \in [L \hat{\Gamma}_1(L), L \hat{\Gamma}_0)$, has a unique critical point that is a center and of the form $(\pi,y)$.
    \item If $\Gamma \in [ L \hat{\Gamma}_0, L \hat{\Gamma}_2(L))$, has only two critical points $(0,y_1), (\pi,y_2)$, of saddle and center type respectively. 
    \item If $\Gamma \in [ L \hat{\Gamma}_2(L), L \hat{\Gamma}_*(L))$, has only three critical points $(0,y), (\pi,y_{1,2})$ where  $y_2>y_{1}>0$ with $(0,y)$, $(\pi,y_2)$ saddles and $(\pi,y_1)$ a center. 
    \item If $\Gamma= L\hat \Gam_*(L)$, has only two critical points $(0,y)$, $(\pi,y_*)$ with $(0,y)$ a saddle and $(\pi,y_*)$ parabolic. 
    \item If $\Gamma \in ( L \hat{\Gamma}_*(L), L \hat{\Gamma}_3(L)]$, has a unique critical point that is a saddle and of the form $(0,y)$.
\end{enumerate}
\end{itemize}
\end{theorem}
 
The complete proof of this result relies on elementary
analytic tools. However, due to the complexity of the expression of the functions
involved in the definition of the Hamiltonian~\eqref{def:HamAV}, it is not straightforward at all and it is postponed to Section~\ref {sec:eccentric_case}. 
 
\begin{remark}  
Along the proof of this result we check some extra properties of the functions involved. In particular we deduce that,
$$
\hat \Gamma_1(L) \in \left [\frac{2}{5} \hat \Gamma_0,\hat \Gamma_0\right ],\qquad 
\hat{\Gamma}_{2}(L) \in \left [ \frac{1}{5} \sl, \frac{1}{2}\sl \right ], \qquad 
\hat{\Gamma}_{3}(L) \in \left [\frac{1}{2}\sl, \frac{3}{4} \sl \right ].
$$
The constant $\hat \Gamma_0$ is defined as $\hat \Gamma_0 = \sl  \hat{y}_{\min}$ with $\hat{y}_{\min}$ defined below in~\eqref{def:ymin}.

 The exact value of $L_*$ is not used along the proof of Theorem~\ref{thm:AveragedHam_generalcase}; it has been computed numerically just for completeness. 
\end{remark}





\subsection{The coplanar Hamiltonian}\label{sec:mainresults:PO}

The coplanar Hamiltonian ${\HH}_{\CP}= \HH_0 + \al^3 {\HH}_{\CP,1}$ (see~\eqref{def:Ham:Coplanar}) can be seen, when $\al^3$ is small enough, as a perturbation of the averaged one. Enlarging the dimension, one can ask if, in some way, the equilibrium points of the $h$-averaged system, survive as periodic orbit in the coplanar system and, in this case, if they conserve the character (saddle or elliptic type) of the equilibrium point.

From expressions of $\HH_0$ in \eqref{def:hamiltonianPoincareH0} and ${\HH}_{\CP,1}$ in  \eqref{eq:expressionPoincareHCP1} in Appendix~\ref{app:tables}, $\Pi=\{\eta=\xi=0\}$ is invariant by the flow associated to the coplanar Hamiltonian ${\HH}_{\CP}$. From now on, we focus our analysis on these types of periodic orbits, which are the ones used in \cite{AlessiBGGP23} to obtain drift in eccentricity.  

The results presented in this section attempt to give a theoretical framework of the previous numerical study in~\cite{AlessiBGGP23} by means of perturbative arguments with respect to the (small) parameter $\delta= \rho \alpha^3 L^4$ defined in~\eqref{def:delta:mainresults}. 

The following result  proves the existence of periodic orbits in suitable energy level of the Hamiltonian $\HH_\CP$. It is proven in Section~\ref{sec:proof_periodicorbits}.

\begin{theorem}\label{thm:existence_periodic_orbits}
For $L\in (0,L_{\max}]$, the functions $\mathbf{E}_{\min,\max}$ defined by
$$
\mathbf{E}_{\max}(L) = {\HH}_\CP (0,0,0,0), \qquad \mathbf{E}_{\min} (L) = {\HH}_\CP(0,0.49 L,0,\pi) 
$$
can be written as 
$$
\mathbf{E}_{\max}(L)= \frac{\rho_0}{16L^6}\widehat{\mathbf{E}}_{\max}(\delta), \qquad 
\mathbf{E}_{\min}(L)= \frac{\rho_0}{16L^6}\widehat{\mathbf{E}}_{\min}(\delta),
$$
with $\delta= \rho \alpha^3 L^4$ introduced in~\eqref{def:delta:mainresults}, and  $\widehat{\mathbf{E}}_{\min,\max}:[0,\delta_{\max}] \to \mathbb{R}$.

Then, for any energy level $\mathbf{E}$ such that $\mathbf{E}_{\min}(L)  \leq \mathbf{E} \leq \mathbf{E}_{\max}(L) $, there exists a periodic orbit $(0, {\Gam}(t;L),0,h(t;L))$ satisfying that $h(0;L)=0$, $\Gam(t;L) \in [0, 0.49 L]$ for all $t\geq 0$, 
$$
\HH_\CP(0, {\Gam}(t;L),0,h(t;L))= \mathbf{E}
$$
and the differential equations
\begin{align*}
\dot{h} =\frac{dh}{dt}= & -\frac{3\rho_0(L+2\Gam)}{4 L^8}
- \al^3 \frac{3\rho_1}{8L^4}\Big(
 U_2^{0,0}(L + 2\Gam)
+
\frac{4}{3} U_2^{1,0}
\frac{L^2-4L\Gam-4\Gam^2}{\sqrt{(L-2\Gam)(3L+2\Gam)}} \cos h
\\
& -\frac{1}{3} U_{2}^{2,0}(L+2\Gam)\cos(2h) \Big)\\
\dot{\Gam} =  \frac{d\Gam}{dt}=&  -\al^3 \frac{\rho_1}{16 L^4}
\Big(
2 U_2^{1,0}
\sqrt{(L-2\Gam)(3L+2\Gam)}(2\Gam+L) 
\sin h 
\\
&+ U_2^{2,0} 
(L-2\Gam)(3L+2\Gam)\sin(2h)
\Big).
\end{align*} 
\end{theorem}
  
\begin{remark}\label{rmk:energyvalues}
    The values $\mathbf{E}_{\min}(L), \mathbf{E}_{\max}(L)$ can be  easily computed. For the case $a=29600\,\mathrm{km}$ these are  
    $$  
    \mathbf{E}_{\min}(1)= -2.558100888960067 \cdot 10^{-5}, \qquad 
     \mathbf{E}_{\max}(1)=2.477266122798186 \cdot 10^{-6} .
    $$ 
    Recall that the units are  the ones in  Remark~\ref{rmk:units}, namely the semi-major axis of Galileo is $a=1$ and the period of Galileo is $2\pi$ (this implies $L=1$). Compare with the numerical results in~\cite{AlessiBGGP23} for hyperbolic periodic orbits in Figure~\ref{fig:PO_example}.
    \begin{figure}[h]
\begin{center}
\hspace{-3cm}
\includegraphics[width=8.2cm]{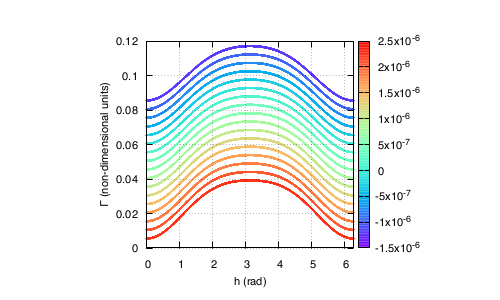}\hspace{-0.1cm}
\end{center}
\caption{Examples of periodic orbits for $\HH_\CP$ in the case $a=29600\, \mathrm{km}$ (Galileo semi-major axis) in non-dimensional units. The colorbar reports the value of ${\HH}_{\CP}$.  }
\label{fig:PO_example}
\end{figure}
\end{remark}  

We emphasize that, by Theorem~\ref{thm:existence_periodic_orbits}, the initial conditions  $\Gam(0;L)=\Gam_0$ for a given value $L\in (0,L_{\max}]$ for a periodic orbit are the ones satisfying
\begin{equation}\label{initial_condition_property}
\mathbf{E}_{\min}(L)\leq\HH_\CP (0,\Gam_0,0,0) \leq \mathbf{E}_{\max}(L),
\end{equation}
with $\mathbf{E}_{\min,\max}(L)$ defined in Theorem~\ref{thm:existence_periodic_orbits}.

For $\sigma~\in \left [0,\frac{1}{2}\right )$, we define 
\begin{align}\label{defhata0c0}
\hat a_0 (\sigma;\delta)&= -1-2\sigma
- \delta \frac{ 1 }{2 } 
 U_2^{0,0}(1 + 2 \sigma),   \\ 
\hat c_0(\sigma;\delta) & =  - \frac{1}{24 \hat a_0(\sigma;\delta)}
\left [
4 U_2^{1,0}
\sqrt{(1-2\sigma)(3+2\sigma)}(2\sigma+1) 
 +   U_2^{2,0} 
(1-2\sigma)(3+2\sigma) 
\right]. \notag 
\end{align}
\begin{remark}
After some tedious computations (see  Section~\ref{proof:corollary_periodic}) one can check that
 \begin{align*}
 \hat a_0 \left ( \frac{\Gam}{L}; \rho\alpha^3 L^4 \right ) & = \frac{4L^7} {3\rho_0} \partial_\Gam \HH_{\AV} (0, \Gam,0) , \\ 
 \hat{c}_0\left ( \frac{\Gam}{L}; \rho \alpha^3 L^4 \right )  &= \frac{\rho}{L^5 }\frac{\HH_{\CP,1}(0,\Gam,0,0)- \HH_{\AV,1}(0,\Gam,0)}{  \partial_\Gam \HH_\AV (0,\Gam,0)}.
 \end{align*}
\end{remark}


The next result is a corollary of Theorem~\ref{thm:existence_periodic_orbits}. 
\begin{corollary}\label{cor:gammaCPAV}
There exists a constant $C_*>0$ such that for a given $L\in (0,L_{\max}]$ and an energy level $\mathbf{E} \in [\mathbf{E}_{\min}(L), \mathbf{E}_{\max}(L)]$, if $\Gam_\AV^{\mathbf{E}}, \Gam_\CP^{\mathbf{E}}$ satisfy
$$
\mathbf{E} = \HH_\AV(0,\Gam_\AV^{\mathbf{E}},0)= \HH_\CP(0,\Gam_\CP^{\mathbf{E}},0,0)
$$
then there exists a function $\gamma_2:\left (0, \frac{1}{2}\right ) \times [0,\delta_{\max}]\to \mathbb{R}$ such that,  
$$
\Gam_\AV^{\mathbf{E}}  = 
\Gam_\CP^{\mathbf{E}} + L \delta  
\hat c_0 \big (\widehat \Gam_{\CP}^{\mathbf{E}}; \delta  \big ) + L \delta^2 \gamma_2(\widehat \Gam_{\CP}^{\mathbf{E}};\delta) , \quad \widehat \Gam_\CP^{\mathbf{E}} = L^{-1} \Gam_\CP^{\mathbf{E}}, \quad \delta = \rho\alpha^3L^4,
$$ 
with $\hat c_0$ in~\eqref{defhata0c0} and $|\gamma_2(\hat \Gam_{\CP}^{\mathbf{E}};\delta)|\leq C_*$.
\end{corollary}
The proof of this result is postponed to Section~\ref{proof:corollary_periodic}.

\begin{remark}
Corollary~\ref{cor:gammaCPAV} gives us a first order in $\delta$ correction of the constant value of $\Gam$ to consider in the $h$-averaged Hamiltonian with respect to the initial condition of the periodic orbit in order to conserve the same energy level. Such a correction will have a relevant role when the stability of the periodic orbit is analyzed.
\end{remark}

As already mentioned, for Galileo $\delta\sim 0.11$. Then, Corollary \ref{cor:gammaCPAV} shows that the two periodic orbits have a significant separation. Moreover, since at circular motions $G=L$ (see \eqref{def:GH}), taking \eqref{eq:xy} into account, one has that 
\[
\Gamma=H-\frac{L}{2}=L\cos i-\frac{L}{2}
\]
Then, since $L$ is a constant for the secular Hamiltonian, deviations in $\Gamma$ correspond to deviations of $\cos i$. Namely, Corollary \ref{cor:gammaCPAV} implies that the inclination between the two periodic orbits is significantly different.

Our next goal is to provide a first order approximation, with respect to the parameter $\delta$, defined in~\eqref{def:delta:mainresults},
of the  periodic orbits $(0,\Gam(t;L),0,h(t;L))$ of the coplanar Hamiltonian system $\HH_{\CP}$ emanating from the origin of the $h$-averaged Hamiltonian~\eqref{def:HamAV}. 

To keep our analysis as concrete as possible, before stating the result, we introduce the functions defined for $\sigma\in \left [0,\frac{1}{2} \right )$
\begin{equation}\label{defhata1c1d1}
\begin{aligned}  
\hat c_1(\sigma;\delta) & = \frac{1}{6 \hat a_0(\sigma;\delta)} U_2^{1,0} \sqrt{(1-2\sigma)(3+2\sigma)} (2\sigma+1),
\\ 
\hat c_2 (\sigma;\delta) & = \frac{1}{24 \hat a_0(\sigma;\delta)} U_2^{2,0} (1-2\sigma)(3+2\sigma),   
\\ 
\hat d_1(\sigma;\delta) &= 
\frac{1}{3 \hat a_0(\sigma;\delta)} U_{2}^{1,0} \frac{5-12 \sigma-12\sigma^2}{ \sqrt{(1-2\sigma)(3+ 2\sigma)}},
\\
\hat d_2(\sigma;\delta) & = 
\frac{1}{24 \hat a_0(\sigma;\delta)} U_{2}^{2,0} \frac{1-12 \sigma-12\sigma^2}{ 1+2\sigma}.
\end{aligned}
\end{equation}


 
\begin{proposition}\label{prop:firstaproximation_periodic} 
There exists a constant $C_*>0$ such that if $L\in (0,L_{\max}]$ and $\Gam_0$
satisfy~\eqref{initial_condition_property}, then the periodic orbit $(0,\Gamma (t;L),0,h (t;L))$ of $\HH_\CP$ with initial condition $\Gam(0;L)=\Gam_0:=L \hat \Gam_0$ is of the form
$$ 
\Gam(t;L)= L \hat \Gam \left ( {\frac{3 \rho_0}{4L^7}} t; \delta \right ), \qquad 
h(t;L)= \hat{h}\left ( {\frac{3 \rho_0}{4L^7}} t ; \delta \right )
$$
where $\delta=\rho \alpha^3 L^4 $ as in~\eqref{def:delta:mainresults} and the functions $\hat \Gam, \hat h:[0,\delta_{\max}] \to \mathbb{R}$ are of the form
$$
\hat \Gam(s;\delta)= \hat \Gam_0  + \delta \hat \Gam_1(s;\delta)+ \delta^2 \hat \Gam_2(s;\delta),\qquad \hat h(s) = \hat h_0(s;\delta) + \delta \hat h_1(s;\delta) + \delta^2 \hat h_2(s;\delta)
$$
with  
$|\hat \Gam_2(s;\delta) |, |\hat h_2(s;\delta) |\leq C_*$ and 
\begin{align*}
\hat \Gam_1(s;\delta) & = \hat c_0(\Gam_0;\delta) + \hat c_1(\hat \Gam_0;\delta) \cos \big (\hat a_0(\hat \Gam_0;\delta) s\big ) + \hat c_2(\Gam_0;\delta) \cos \big (2 \hat a_0(\hat \Gam_0;\delta) s\big ), \\
\hat h_0(s;\delta ) &= \hat a_0(\hat \Gam_0;\delta) s, \\
\hat h_1(s;\delta) & = \partial_{\hat \Gam} \hat a_0(\hat \Gam_0;\delta) \hat c_0(\hat \Gam_0;\delta ) s + \hat d_1(\hat \Gam_0;\delta ) \sin \big (\hat a_0(\hat \Gam_0;\delta) s\big ) + \hat d_2(\hat \Gam_0;\delta) \sin \big (2 \hat a_0(\hat \Gam_0;\delta) s\big ) .
\end{align*}
The constants $\hat a_0,\hat c_{0,1,2}$ and $\hat d_{1,2}$ are defined in~\eqref{defhata0c0} and~\eqref{defhata1c1d1}.

The period $\mathcal{T}(\Gam_0;L) $ of the periodic orbit satisfies that
\begin{equation}\label{formula_period_theorem} 
\mathcal{T}(\Gam_0;L)= \frac{4 L^7 }{3 \rho_0} \widehat{\mathcal{T}}(\hat \Gam_0;\delta),\qquad 
\widehat{\mathcal{T}}(\hat \Gam_0;\delta) = \frac{2\pi + \delta^2 \widehat{\mathcal{T}_2}(\hat \Gam_0;\delta)}{\big | \hat a_0(\hat \Gam_0;\delta) + \delta \partial_{\hat \Gam} \hat a_0(\hat \Gam_0;\delta)  \hat c_0(\hat \Gam_0;\delta) \big |}
\end{equation}
with $\widehat{\mathcal{T}}: \left (0 , \frac{1}{2}\right ) \times [0,\delta_{\max}] \to \mathbb{R}$ and  
$|\widehat{\mathcal{T}}_2(\hat \Gam_0;\delta)|\leq M$. 
\end{proposition}

This proposition is proven in Section~\ref{sec:perturbative analysis_periodic}.
\begin{remark}
    We notice that the dependence on $L$ of the periodic orbit comes through the dependence on suitable functions with respect to the ``small" parameter $\delta$. In~\cite{AlessiBGGP23}, the period of the hyperbolic periodic orbits are numerically computed by some values of $\Gamma/L$ for Galileo ($L=1$). Indeed, in Fig.~\ref{fig:PO_example_periods} there is the comparison between the numerically computed value of the period and the approximated value, $\mathcal{T}^{\mathrm{approx}},$ given by formula~\eqref{formula_period_theorem} with $\widehat{\mathcal{T}}_2\equiv 0$, namely 
    $$
    \mathcal{T}^{\mathrm{approx}} (\Gam_0;L)= \frac{4L^7}{3\rho_0} \frac{2\pi}{\big |\hat{a}_0(\hat{\Gam}_0; \delta) + \delta \partial_{\hat \Gam} \hat a_0 (\hat \Gam_0;\delta) \hat c_0 (\hat \Gam_0;\delta)\big |} 
    $$
\begin{figure}[h]
\begin{center}
\hspace{-1cm}
\includegraphics[width=5cm]{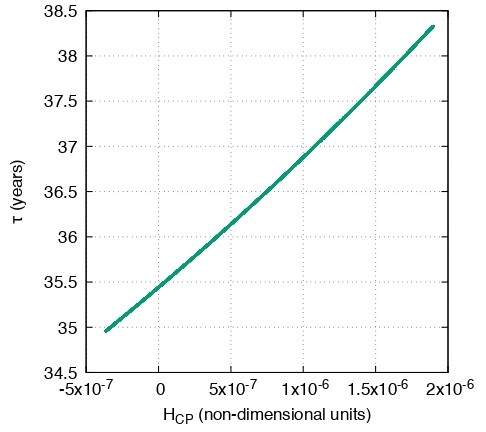}   \hspace{2cm}
\includegraphics[width=4.7cm]{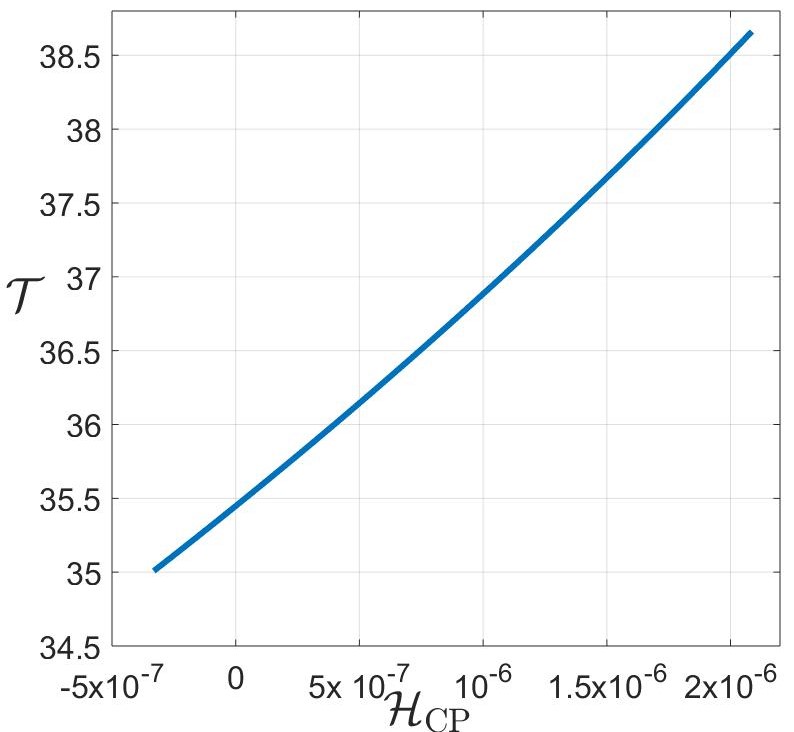}
\end{center}
\caption{On the right, the numerically computed period (left) $ \TTT$ of the hyperbolic periodic orbits as a function of ${\HH}_{\CP}$. On the left the approximated period, $\mathcal{T}^{\mathrm{aprox}}$ computed discarding the $\mathcal{O}(\delta^2)$ terms in~\eqref{formula_period_theorem}.}
\label{fig:PO_example_periods}
\end{figure}

\end{remark}

\begin{remark}\label{rem:period} From definition~\eqref{defhata0c0} of $\hat a_0$, we obtain that, if $\delta$ is small enough 
    $$
     \widehat{\mathcal{T}} (\hat \Gam_0;\delta) =  \frac{2\pi}{2\hat \Gam_0+1} +\mathcal{O}(\delta) \in \big (\pi + \mathcal{O}(\delta), 2\pi +  \mathcal{O}(\delta)\big ). 
    $$
\end{remark}


Our last result  is the comparison between  the stability of a periodic orbit of the coplanar Hamiltonian $\HH_\CP$ and the stability of the origin in the $h$-averaged Hamiltonian
belonging to the same energy level (see Corollary~\ref{cor:gammaCPAV}).

In order to analyze the character of the periodic orbits, we denote
by $X_{\CP}$ the vector field of the Hamiltonian $\HH_\CP$. For $(0, \Gam(t;L),0,h(t;L))$, a periodic orbit, the variational equation around the periodic orbit is given by
$$
\dot{z} = D X_\CP (0, \Gam(t;L),0, h(t;L)) z.
$$
Let $\Phi(t;\Gam_0,L)$ with $\Gam_0=\Gam(0;L)$, be the fundamental matrix satisfying the initial condition $\Phi(0;\Gam_0,L)=\mathrm{Id}$. The monodromy matrix is defined just by 
\begin{equation}\label{def:monodromy_matrix}
 \Phi (\mathcal{T} (\Gam_0;L);\Gam_0,L)
\end{equation}
with $\Gam_0=\Gam(0;L)$ and $\mathcal{T}(\Gam_0;L)$, the period of the considered periodic orbit. Since $\HH_\CP$ is a 2-degrees-of-freedom Hamiltonian, the eigenvalues of the monodromy matrix are of the form 
\begin{equation}\label{def:muCP}
\left ( \mu_\CP(\Gam_0;L)  , \big (\mu_\CP (\Gam_0;L) \big )^{-1}, 1,1 \right ). 
\end{equation}


We are interested now in comparing 
the eigenvalues of the monodromy matrix with the ones coming from the constant linear part around the origin of $\HH_\AV$,  the $h$-averaged Hamiltonian~\eqref{def:HamAV}. That is, we want to compare $\mu_\CP$ with 
$e^{\lambda_\AV   \mathcal{T}}$ where $\lambda_\AV$ has been analyzed in Proposition~\ref{prop:origin_modulus_eigenvalue}.  

In the following result, we provide a first order (up to an error of order $\mathcal{O}(\delta^{3/2})$) approximation of the eigenvalues of the monodromy matrix.

\begin{theorem} \label{prop:eigenvalues_monodromy} 
For $L\in (0,L_{\max}]$ and $\Gam \in \left (0,\frac{L}{2} \right )$, we define 
\begin{equation}\label{def:GammaAV}
\Gam_\CP^{(0)}(\Gam;L) = L \hat \Gam_\CP^{(0)}( \hat \Gam; \delta):= L \big (\hat \Gam+ \delta \hat c_0(\hat \Gam;\delta)), 
\end{equation}
with $\hat c_0$ defined in~\eqref{defhata0c0} and $\delta = \rho \alpha^3 L^4$.
We also introduce  
\begin{equation}\label{def:muCP0}
\mu_{\CP}^{(0)} (\Gam;L) = \mathrm{exp} \big (\mathcal{T}^{(0)}( \Gam;L) \lambda_{\AV} ( \Gam;L) \big ), \qquad \mathcal{T}^{(0)} (\Gam;L)= \frac{3\rho_0}{4L^7} \frac{2\pi}{\big |\hat a_0 (\hat \Gam;\delta) \big |}, 
\end{equation}
where  $\hat a_0$ is defined in~\eqref{defhata0c0} and $\lambda_\AV(\Gam;L)$ (see Proposition~\ref{prop:origin_modulus_eigenvalue} ) is an eigenvalue of the linearization around the origin of the $h$-averaged Hamiltonian $\HH_\AV$ for constant values $\Gam,L$ .

Then, there exist $\delta_0>0$ small enough and a constant $C_*>0$ such that, if  $L \in (0,L_{\max}]$, $\Gam_0= L \hat \Gam_0$ satisfy~\eqref{initial_condition_property} and $\delta=\rho \alpha^3 L^4 \in [0,\delta_0]$,  we have that
\begin{equation*}
|\mu_\CP(\Gam_0;L) - \mu_\CP^{(0)} (\Gam_\CP^{(0)} (\Gam_0;L);L)| \leq C_* \delta^{3/2},
\end{equation*}
where $\mu_\CP (\Gam_0;L)$ is an eigenvalue of the monodromy matrix in~\eqref{def:monodromy_matrix} (see also~\eqref{def:muCP}) corresponding to the linearized system around $(0,\Gam(t;L), 0, h(t;L))$, the periodic orbit of $\HH_\CP$ with initial condition $\Gam(0;L)=\Gam_0$, given in Theorem~\ref{thm:existence_periodic_orbits}.
\end{theorem}
The proof of this result is postponed to Section~\ref{sec:sub_eigenvalues_perturbative} and relies on a first order perturbation analysis. 
\begin{corollary} 
For any $\hat \Gam_0 \in \left (0,\frac{1}{2}\right )$ there exists $\delta_0>0$ such that if $L, \Gam_0=\hat \Gam_0L$ satisfy~\eqref{initial_condition_property} and $\delta=\rho \alpha^3 L^4 \in [0,\delta_0]$, then the periodic orbit with initial condition $\Gam_0=\hat \Gam_0 L$ is of saddle type if $\hat \Gam_0 \in (\hat \Gam_-(\delta), \hat \Gam_+(\delta))$, with $\hat \Gam_{\pm}(\delta)$ defined in Theorem~\ref{thm:origin_average}, and of elliptic type otherwise.
%
%

In addition, there exists  $\delta_0>0$, uniform for $\hat \Gam_0 \in \left \{   \frac{\sl}{2} \right \} \cup \left ( 0,\frac{\sl}{4} \right ] \cup \left [m, \frac{1}{2} \right ) $, such that the periodic orbit with initial condition $\Gam_0=L\hat \Gam_0$  is of saddle type if $\hat \Gam_0=\frac{\sl}{2} $ and of elliptic type if $\hat\Gam_0 \in 
\left ( 0,\frac{\sl}{4} \right ] \cup \left [m, \frac{1}{2} \right ) $. 
\end{corollary}
This result is a direct consequence of Theorems~\ref{prop:eigenvalues_monodromy} and~\ref{thm:origin_average}.

\begin{figure}[h]
\begin{center}
\includegraphics[width=7.8cm]{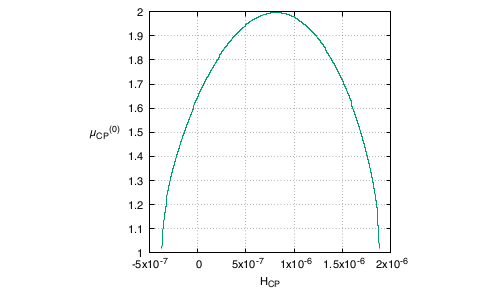}  
\includegraphics[width=5cm]{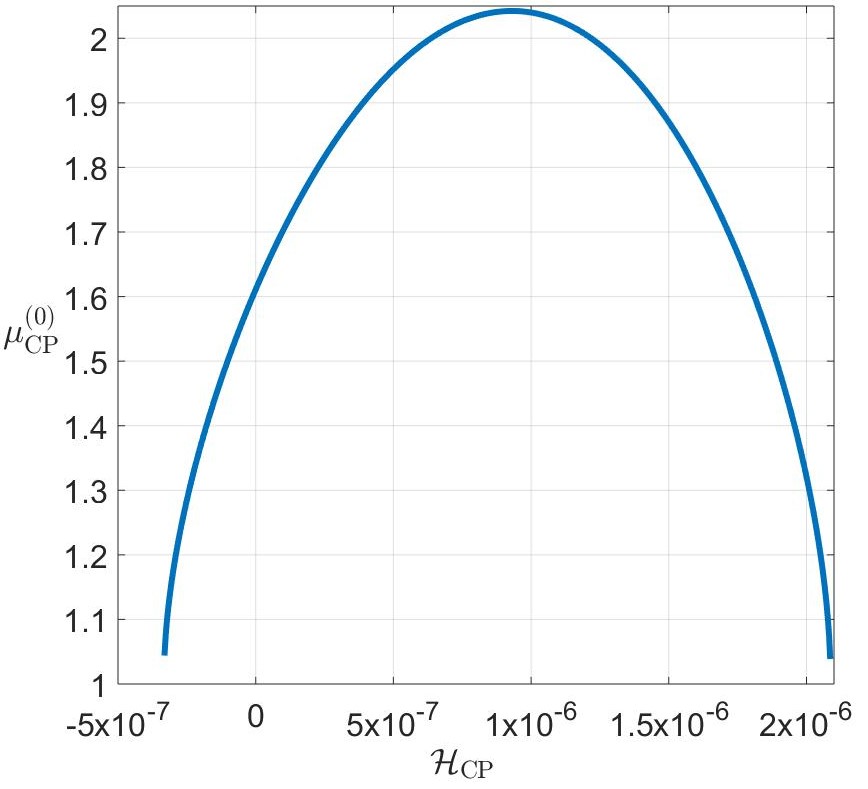}
\end{center}
\caption{The value of $\mu_\CP$ (the eigenvalue greater than 1) as a function of  ${\HH}_{\CP}$ for the hyperbolic periodic orbits. On the right, the numerical computations performed in~\cite{AlessiBGGP23}, and on the left, the corresponding approximated value $\mu_\CP^{(0)}$ in Theorem~\ref{prop:eigenvalues_monodromy}.}
\label{fig:approximated_eigenvalues}
\end{figure}

The theoretical approximation result in Theorem~\ref{prop:eigenvalues_monodromy}, agrees with the ones computed numerically in~\cite{AlessiBGGP23}. Indeed, in  Figure~\ref{fig:approximated_eigenvalues}, we present the comparison between the first order approximated eigenvalue $\mu_\CP^{(0)}$ and the numerical computation of $\mu_\CP$ performed in ~\cite{AlessiBGGP23}. The analysis is restricted to values of $\hat \Gam_0$ close to $\frac{\sl}{2}\sim 0.058257569495584$ so that the periodic orbit could be (and from the numerical point of view is) hyperbolic. 

\begin{remark}
By Corollary~\ref{cor:gammaCPAV}, we obtain that 
$$
\HH_\CP(0,\Gam_0,0,0)= \HH_\AV(0,\Gam_\CP^{(0)}(\Gam_0;L),0) + \mathcal{O}(\delta^2 )
$$
with $\Gam_{\CP}^{(0)}$ defined in~\eqref{def:GammaAV}. 
Therefore, as a straightforward consequence of Theorem~\ref{prop:eigenvalues_monodromy}
we conclude that the eigenvalues of the monodromy matrix of a periodic orbit of period $\mathcal{T}_{\mathbf{E}}$ lying in an energy level $\mathbf{E}$ are well approximated by 
$e^{\mathcal{T}_{\mathbf{E}} \lambda_{\mathbf{E}}}$ with $\lambda_{\mathbf{E}}$ an eigenvalue associated to the origin as critical point of the $h$-averaged hamiltonian in the energy level $\{\HH_{\AV} = \mathbf{E}\}$ .  

\begin{figure}[h]
\begin{center}
\includegraphics[width=5cm]{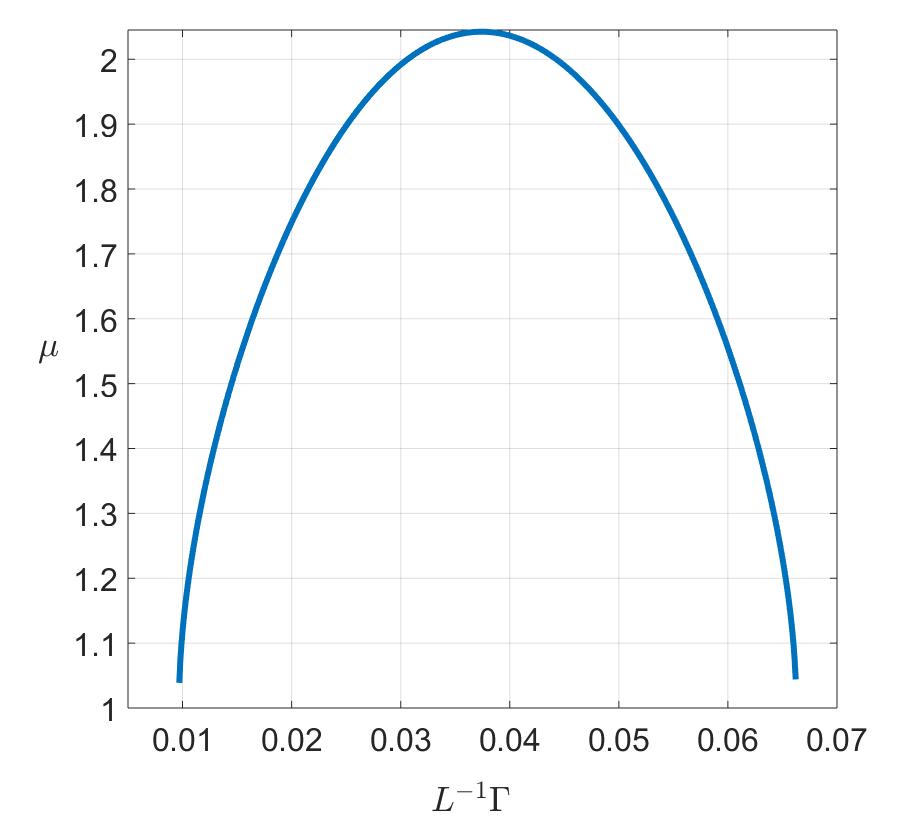}
\hspace{0.1cm}
\includegraphics[width=5cm]{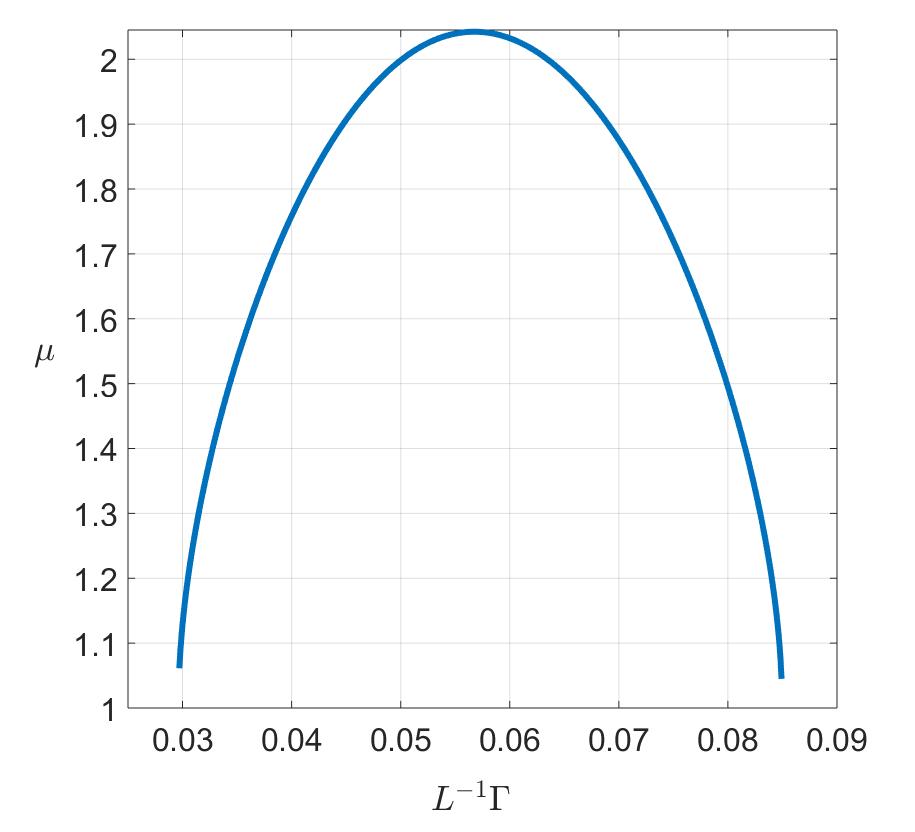}
\end{center}
\caption{Approximated value of the eigenvalue greater than $1$ associated to the monodromy matrix as a function of $L^{-1}  {\Gamma}$. On the left for the Hamiltonian $\HH_\CP$ and on the right the corresponding one for the $h$-averaged Hamiltonian $\HH_\AV$.}
\label{fig:eigenvalues_Gamma}
\end{figure} 

Note, however, that, in terms of the initial condition of the periodic orbit, namely $\Gam(0;L)=\Gam_0$, the information about the eigenvalues of the monodromy matrix associated to a periodic orbit with initial condition $\Gam_0$ of the coplanar Hamiltonian comes from the $h$-averaged Hamiltonian corresponding to the \textit{corrected} value of the parameter $\Gam$ given by $\Gam_\CP^{(0)}(\Gam_0;L)$ in~\eqref{def:GammaAV}. 
To show this discrepancy, for $L=1$, in Figure~\ref{fig:eigenvalues_Gamma} they are represented the eigenvalues of the monodromy matrix associated to the periodic orbit with respect to the initial condition $\Gam/L$ as well as $e^{\mathcal{T}(\Gam;L) \lambda_\AV( \Gam;L)}$ the eigenvalues of the linearization of $\HH_\AV$ corresponding to the same parameter $\Gam/L$.  

For Galileo, which corresponds to $\HH_\CP=0$ and semi-major axis $a=1$ (in the normalized units). The value 
$\hat \Gam_0=0.058788134221194$ corresponding to the initial condition $\HH_\CP (0,L \hat \Gam_0,0,0)=0$, can be numerically computed. Then the corrected value for $\Gam$ is 
$ 
\Gam_\CP^{(0)}(\Gam_0;L) = 0.079342370619096. 
$ 
Therefore, we have that  the approximated eigenvalue is  
$$
\mu_\CP^{(0)}(\Gam_\CP^{(0)}(\Gam_0;L);L)=  1.610955038638576.
$$
However, if the correction in $\Gamma$ is not considered in the $h$-averaged Hamiltonian and the period is only approximated by the averaged Hamiltonian, we obtain
$$  
e^{\lambda_\AV(\Gam_0;L) \TTT(\Gam_0;L)} = 2.038395173606618.
$$ 
Therefore, the difference between the eigenvalues of the monodromy and the corresponding orbit in the $h$-averaged system, namely with the same initial condition $\Gam_0$, can not be \emph{a priori} neglected. 
\end{remark}

The last result of this work can be seen as a refinement of Theorem~\ref{prop:eigenvalues_monodromy}. It gives a more accurate bound for the eigenvalues of the monodromy matrix for the most part of the values of $L, \Gam_0$. In order to state such a result, we introduce the values $\widehat \Gam^{(k)}_0  \in \left (0,\frac{1}{2}\right )$ satisfying that 
\begin{equation}\label{def:Gamma0k}
 \frac{2\pi}{\hat a_0\big (\widehat \Gam^{(k)}_0;0 \big )}   |1- 16 \widehat \Gam^{(k)}_0 - 20 (\widehat \Gam^{(k)}_0)^2|  = k \pi,\qquad k=1,2. 
\end{equation}
That is $  |1- 16 \widehat \Gam^{(k)}_0 - 20 (\widehat \Gam^{(k)}_0)^2| = 2 (1+2\widehat \Gam^{(k)}_0)k$. 
These values, $\widehat \Gam^{(k)}_0 $, can be analytically computed (see Table~\ref{table_sigma} for the numerical value) because they are the zeros of suitable polynomial of degree $2$. 
\begin{table}[h!]
\centering
\begin{tabular}{c |c  } 
$k$ &   $\hat \Gam^{(k)}_0 $ \\
\hline
$1$ & $0.189897948556636 $   \\
$ 2$ &   $0.338516480713450$  
\end{tabular}
\caption{
The values of $\widehat{\Gam}_0^{(k)}$. For $k=3$, $\widehat{\Gam}_0^{(3)}=1/2$ and for $k\geq 4$, the values of $\widehat{\Gam}_0^{(k)}$ are out of the interval $\left (0,\frac{1}{2}\right )$. 
}
\label{table_sigma}
\end{table}

We also introduce the constant $C_0$ defined   through~\eqref{thm:exprGamma12} in Theorem~\ref{thm:origin_average}, as
\begin{equation}\label{def:Cm}
  C_0 = \frac{5 U_{2}^{1,0}(\sl +3)\sqrt{3- 2 \sl -4 \sl ^2}}{3(16+20 \sl)}   
\end{equation}
in such a way that $\Gam_{\pm}(L) L^{-1} = \frac{\sl}{2} \mp C_0 \delta + \mathcal{O}(\delta^2)$. We recall that $\lambda_\AV (\Gam_{\pm}(L);L)=0$. 

\begin{lemma} \label{lem:Gamkpm} There exists $\delta_0>0$ such that if $L$ satisfies that $\delta= \rho \alpha^3 L^4 \in [0,\delta_0]$ then,
there exist $\Gam^{(k)}(L)$, $k=1,2$ such that 
$$ 
\mu_\CP^{(0)} ( \Gam^{(k)}(L);L)  = (-1)^{k},
$$   
where $\Gam^{(k)}(L)= L \widehat{\Gam}^{(k)}(\delta)$ (with $\delta = \rho \alpha^3 L^4$) satisfy
$$
\big | \widehat{\Gam}^{(k)}(\delta) - \widehat{\Gam}^{(k)}_0  \big |\leq M \delta^2,
$$
with the value $\hat \Gam^{(k)}_0$ defined by~\eqref{def:Gamma0k} (see also Table~\ref{table_sigma}) and the constant $M$ is independent of $L$. 
 
In addition, if either $\Gam \neq \Gam^{(k)}(L)$ or $\Gam \neq \Gam_{\pm}(L)$, then $\mu_\CP^{(0)}(\Gam;L) \neq \pm 1$.
\end{lemma}
This lemma is proven in Section~\ref{subsec:lastlema}. 

\begin{theorem}\label{thm:eigenvalues_monodromy_more} 
Assume that we are under the conditions of Theorem~\ref{prop:eigenvalues_monodromy} and denote  $\Gam_\CP^{(0)} (\Gam_0;L) = L\hat \Gam_\CP^{(0)}(\hat \Gam_0;\delta)$, as in~\eqref{def:GammaAV}. 

Fix $C \in (0,C_0)$ (see \eqref{def:Cm}), $\nu \in [1,2]$ and consider the sets, for $k=1,2$, defined for $\delta>0$ as
$$
I_{\pm,\nu}^{(0)} (\delta) = \left \{ \sigma \in \mathbb{R},\, \left |\sigma - \frac{\sl}{2} \mp C_0 \delta  \right |\leq C \delta^\nu  \right \}, \;
I^{(k)}_\nu (\delta)= \left \{  \sigma \in \mathbb{R},\, \left |\sigma -  \widehat{\Gam}_0^{(k)} \right |\leq C \delta^{\frac{1+\nu}{2}} \right \}.
$$
There exists a constant $C_*>0$  and $\delta_0>0$ small enough, such that, for $\delta \in [0,\delta_0]$, 
if  $\hat \Gam_\CP^{(0)}(\hat \Gam_0;\delta) \notin I_{\pm,\nu}^{(0)}(\delta) \cup I_\nu^{(1)}(\delta) \cup I_\nu ^{(2)}(\delta)$  then 
$$
\big |\mu_\CP(\Gam_0;L) - \mu_\CP^{(0)} (\Gam^{(0)}_\CP(\Gam_0;L);L) \big | \leq C_* \delta^{\frac{5-\nu}{2}}.
$$ 
As a consequence when $
\hat   \Gam_\CP^{(0)}(\hat \Gam_0;\delta) \notin I_{\pm,1}^{(0)}(\delta) \cup I_1^{(1)}(\delta)
\cup I_1^{(2)}(\delta)$, then 
$$
\big |\mu_\CP(\Gam_0;L) - \mu_\CP^{(0)} (\Gam^{(0)}_\CP(\Gam_0;L);L) \big | \leq C_* \delta^{2}.
$$
\end{theorem}  

This result is proven in Section~\ref{subsec:lasttheorem}.

The remaining part of this work is devoted to prove all the results stated in this section. 

\section{The circular critical points of the $h$-averaged Hamiltonian}\label{app:AveragedHam}

In Section~\ref{sec:goodexpresion}, we  rewrite the averaged system in a more suitable way to analyze the critical points. This new expression will be used both to study  the circular critical points and the eccentric ones. After that, in Sections~\ref{sec:circular_case_Appendix} and \ref {sec:modulus}, we prove Theorem~\ref{thm:origin_average} and Proposition~\ref{prop:origin_modulus_eigenvalue} respectively.  

\subsection{Rewriting the $h$-averaged system}\label{sec:goodexpresion}

We start with the $h$-averaged Hamiltonian in slow-fast Delaunay variables $(y,x)$. That is, we consider  (see~\eqref{eq:expansionsH1averagedDelaunay})
\begin{equation*} 
\begin{aligned}
\rH_{\AV}(y,\Gam,x;L) & = \rH_0(y,\Gamma;L)  + \alpha^3 \rH_{\AV,1}	(y,\Gam, x;L)\\
&=	\frac{\rho_0}{128} \frac{y^2 - 6y\Gam-3\Gamma^2}{L^3 y^5}+ \alpha^3 \frac{\rho_1}{L^2} \left[
\frac{1}{2}U_2^{0,0} D_{0,1}( y ,\Gam) + \frac{1}{3}U_2^{1,0} D_{1,0}(y,\Gam)\cos x			
\right].
\end{aligned}
\end{equation*}
with $D_{0,1}$ and $D_{1,0}$ as given in Table~\ref{tab:functionsDelaunay}.
In order to rewrite it in a more concise form, we introduce the following polynomials in $(y, \Gam)$, which depend implicitly on the parameter $L$,
\begin{equation*}
\begin{array}{lllll}
P_0(y,\Gamma)&=	y^2 - 6y\Gam - 3\Gam^2, &\quad&
Q_{0}(y;L)&=(5L^2 - 12y^2)y^3,\\
P_1(y,\Gamma)&=	(y - \Gam)(3y + \Gam), &\quad&	
Q_{1}(y,\Gamma;L)&= (3y + \Gam)(L - 2y)(L + 2y)y^3,
\end{array}		
\end{equation*}
and their derivatives with respect to $y$,
\bes
\begin{array}{lllll}
P_0'(y,\Gamma)&= 2(y - 3\Gam),&\quad&
Q_{0}'(y;L)	&= 15(L - 2y)(L + 2y)y^2,\\
P_1'(y,\Gamma)&= 2(3y- \Gam),&\quad&
Q_{1}'(y,\Gamma;L)&= y^2\bigg[3(L-2y)(L+2y)(4y+\Gam) - 
8y^2(3y+\Gam) \bigg].
\end{array}		
\ees
Hence,the $h$-averaged Hamiltonian as well as its derived equations of motion can be written as follows
\begin{equation*} 
\begin{split}
\rH_{\AV}(y,\Gamma,x;L)	=\frac{1}{128}\frac{\rho_0}{L^4y^5}		\left(	\cA_{0,\alpha}(y, \Gam) + \alpha^3 	\sqrt{P_1(y,\Gam)}\cA_{1}(y, \Gam;L)\cos x
\right)	
\end{split}
\end{equation*}
and
\be 
\begin{split}
\dot{x}	&= 	- \frac{1}{256}\frac{\rho_0}{L^4y^6}\frac{1}{\sqrt{P_1(y, \Gam)}}
\left[
\sqrt{P_1(y, \Gam)}\cB_{0,\alpha}(y, \Gam;L) + \alpha^3 \cB_{1}(y, \Gam;L)\cos x 
\right]\\
\dot{y}	&=	\phantom{-}\frac{1}{128}\frac{\rho_0\alpha^3}{L^4y^5} \sqrt{P_1(y, \Gam)}\cA_{1}(y, \Gam;L)\sin x, \end{split}
\label{Eq:EMotMoy}
\ee 
where the functions
\bes
\begin{split}
\cA_{0,\alpha}&=P_0
\left(
L + \alpha^3 d_0 Q_{0}
\right),\\
\cA_{1} &= d_1Q_{1},\\
\cB_{0,\alpha}&=10P_0 \left(L + \alpha^3 d_0 Q_{0}\right) 
- 2y(L + \alpha^3 d_0 Q_{0}) P_0'-2\alpha^3y d_0 P_0 Q_0',\\
\cB_{1}&=d_1(10P_1 Q_1-yQ_{1}P_1'-2yP_1Q_1')
\end{split}
\ees
are polynomial in $(y, \Gam)$, with
\begin{equation}\label{def:d0d1} 
d_0=2 U_2^{0,0}\frac{\rho_1}{\rho_0},		
\quad
d_1=20 U_2^{1,0}\frac{\rho_1}{\rho_0}.
\end{equation}

To analyze the circular critical points (i.e., at  $e=0$), we work with Poincar\'e variables. In this case, the $h$-averaged Hamiltonian~\eqref{def:HamAV} as well as its derived equations of motion can be written as follows
\begin{equation*}
\begin{split}
\HH_\AV (\xi,\eta)	
&=	\frac{1}{128}\frac{\rho_0}{L^4y^5}
\left( \cA_{0,\alpha}+ \alpha^3 d_1\sqrt{P_1}\Qt_1(\xi^2 -\eta^2)\right)	
\end{split} 
\end{equation*}
and
\be
\begin{split}
\dot{\xi}&= -\frac{1}{512}\frac{\rho_0}{L^4y^6}\frac{\eta} {\sqrt{P_1}} \left[
\sqrt{P_1}\cB_{0,\alpha}-8 y \alpha^3 d_1P_1\Qt_1+\alpha^3 \cBt_{1}\times(\xi^2 - \eta^2)
\right]\\
\dot{\eta}&=\frac{1}{512}\frac{\rho_0}{L^4y^6}\frac{\xi}{\sqrt{P_1}}
\left[
\sqrt{P_1}\cB_{0,\alpha}+8y\alpha^3 d_1P_1\Qt_1+\alpha^3 \cBt_{1}\times(\xi^2 - \eta^2) 
\right]\\
\end{split}
\label{Eq:EMotMoyorigin}
\ee
with 
\bes
y =	\frac{L}{2}	-	\frac{\xi^2	+	\eta^2}{4},
\ees
and 
\bes
\begin{split}
\Qt_1&=	\frac{Q_1}{2L-4y}= \frac{1}{2}(3y + \Gam)(L+2y)y^3,	\\
\Qt_1'&=\frac{1}{2}y^2\bigg[3L(\Gam +4y) + 2y(15y + 4\Gam)\bigg],	\\ 
\cBt_{1}&=	d_1(10P_1 \Qt_1 -y\Qt_{1}P_1' -	2yP_1\Qt_1').
\end{split}
\ees
For future purposes, we also decompose 
\begin{equation}\label{decompositioncalB0alpha}
\cB_{0,\alpha}=B_{0} + \alpha^3 B_1
\end{equation}
with
\[
    B_0=L(10 P_0 -2y P_0'), \qquad B_1=d_0 (10P_0 Q_0 - 2y Q_0 P_0'-2yP_0 Q_0')
\]
and we recall that
\begin{equation}\label{def:rho}
U_{2}^{0,0}=0.762646  ,\qquad U_2^{1,0}= 0.547442,\qquad \rho:=\frac{\rho_1}{\rho_0} = 1.762157978551987 \cdot 10^{-18}.
\end{equation}
We also recall that $a_{\max}=30000$ km is the maximum semi-major axis we are going to consider and that $a_{\min}=6378.14$ km (the radius of the Earth) is the minimum one.

Along the proof of the results this notation 
will be used extensively without an explicit mention. 

\subsection{Circular critical points: Proof of Theorem~\ref{thm:origin_average}} \label{sec:circular_case_Appendix} 

The origin, $(\xi,\eta)=(0,0)$, is clearly an equilibrium point of system~\eqref{Eq:EMotMoyorigin}. To prove the statements in Theorem~\ref{thm:origin_average},
we only need to study the linear part of the averaged system at the origin and elucidate the values of the parameters $L,\Gamma$ for which the origin is either a saddle, an elliptic point or a degenerated (parabolic) fixed point.  

We observe that, since
$$
y=y(\xi,\eta)=\frac{L}{2} - \frac{1}{4} (\xi^2 + \eta^2)
$$
satisfies $y(0,0)=L/2$ and $\partial_\xi y(0,0)=\partial_\eta y(0,0)=0$, the variational equation of the averaged system around $(\xi,\eta)=(0,0)$ is $\dot{z}=M(\Gamma;L)z$ where
\begin{equation}\label{origin_critical_variational}
M (\Gamma;L)= g(L) \left (\begin{array}{cc} 0 & -X^{\eta}(\Gamma;L) \\ X^{\xi}(\Gamma;L)  & 0 \end{array}\right ) 
\end{equation}
and 
\begin{align*}
    g(L) & = \frac{\rho_0}{8 L^{10}} \\
    X^{\xi}(\Gamma;L)& = B_0\left (\frac{L}{2}, \Gamma;L\right ) + \alpha^3 B_1 \left (\frac{L}{2}, \Gamma;L\right ) + 4 \alpha^3 d_1 L \sqrt{P_1\left (\frac{L}{2} , \Gamma \right )} \Qt_1 \left (\frac{L}{2}, \Gamma;L\right ) \\ 
    X ^{\eta}(\Gamma;L)& = B_0\left (\frac{L}{2}, \Gamma;L\right ) + \alpha^3 B_1 \left (\frac{L}{2}, \Gamma;L\right ) - 4 \alpha^3 d_1 L\sqrt{P_1\left (\frac{L}{2} , \Gamma \right )} \Qt_1 \left (\frac{L}{2}, \Gamma;L\right ).
\end{align*}
Therefore the eigenvalues $\lambda$ of $M(\Gamma;L)$ satisfy
\begin{equation}\label{expr:eigenvaluese=0}
\lambda^2 = -g^2 (L) X^{\xi}(\Gamma;L) \cdot X^{\eta}(\Gamma;L)
\end{equation}
so that, to determine the character of $(\xi,\eta)=(0,0)$, we need to study the sign of the product $X^{\xi}(\Gamma;L) \cdot X^{\eta}(\Gamma;L)$ with respect to the parameters $\Gamma, L$. To this end, we perform the scaling
$$
\Gamma=L \sigma, \qquad \text{with}\qquad 0< \sigma< \frac{1}{2}
$$
and compute $X^{\xi}(L\sigma;L) \cdot X^{\eta}(L\sigma;L)$. 
One can check (after some tedious computations) that 
\begin{equation}\label{Xxietasigma}
\begin{aligned}
X^{\xi}(L\sigma;L) = \frac{3}{2} L^3 \left [   \beta(L) \mathbf{a}(\sigma) +
     \frac{1}{12} d_1 \alpha^3 L^4 \mathbf{b}(\sigma)
     \right  ] \\ 
X^{\eta}(L\sigma;L) = \frac{3}{2} L^3 \left [   \beta(L) \mathbf{a}(\sigma) -
     \frac{1}{12} d_1 \alpha^3 L^4 \mathbf{b}(\sigma)
     \right  ]
\end{aligned}
\end{equation}
with 
\begin{equation} \label{defaborigin_critical}
\mathbf{a}(\sigma)= 1 - 16 \sigma - 20 \sigma^2,\qquad \mathbf{b}(\sigma)= (2 \sigma +3) \sqrt{3- 4 \sigma -4 \sigma ^2}
\end{equation}
(note that $p(\sigma):=3-4\sigma - 4 \sigma^2 >0$ if $0<\sigma<\frac{1}{2}$ and $p\left (\frac{1}{2}\right )=0$)
and
$$
\beta(L)= 1+ \frac{d_0 \alpha^3 L^4}{4}.
$$ 
\begin{remark}\label{rmk:first_expression_eigenvalues}
Notice that the eigenvalues $\lambda= \lambda(\sigma;L)$ can be explicitly computed. Indeed,  using~\eqref{expr:eigenvaluese=0} and definitions~\eqref{def:d0d1} and~\eqref{def:rho} for $d_0,d_1$ and $\rho$ respectively  
$$
\lambda^2 = -\frac{9\rho_0^2}{256 L^{14}} \left [ \mathbf{a}(\sigma) + \rho \alpha^3 L^4 \mathbf{c}_-(\sigma)\right ] \left [ \mathbf{a}(\sigma) + \rho \alpha^3 L^4 \mathbf{c}_+(\sigma)\right ] 
$$
with 
$$
\mathbf{c}_{\pm}(\sigma) = 
 \frac{U_{2}^{0,0}}{2}\mathbf{a}(\sigma)  \pm  \frac{5U_{2}^{1,0}}{3}  \mathbf{b}(\sigma).
$$
\end{remark}

\begin{lemma}\label{lem:previsigmapm}
Consider 
$$
\mathcal{X}_\pm(\sigma;L):=\beta(L)  \mathbf{a}(\sigma) \pm
     \frac{1}{12} d_1 \alpha^3 L^4 \mathbf{b}\left (\sigma\right).
     $$
For $L\in [0,L_{\max}]$, the function  $\mathcal{X}_+(\cdot;L)$ is strictly decreasing for $\sigma\in \left (0,\frac{1}{2}\right )$ and $\mathcal{X}_-(\cdot ;L)$ 
is strictly decreasing for $\sigma\in \left ( 0, \frac{\sl}{2}\right )$. 

Moreover, there exist two functions $\sigma_\pm : [0, L_{\max}] \to (0,\frac{1}{2})$ such that for any $L\in [0, L_{\max}]$ $\mathcal{X}_\pm (\sigma;L) = 0$ if and only if $\sigma = \sigma_\pm (L)$. In addition, $\sigma_\pm(0)=\frac{\sl}{2}$ and for $L>0$
     $$
     \sigma_-(L) \in \left (0, \frac{\sl}{2} \right ),\qquad \sigma_+(L) \in \left (\frac{\sl}{2} , \frac{1}{4}\right ),$$ 
where $\sl$, defined in~\eqref{def:pendentresonance}, is the slope of the prograde resonance (see also~\eqref{def:resonance}). 
\end{lemma}  
\begin{proof} 
First we notice that $\mathbf{a}, \mathbf{b}$ are decreasing functions for $\sigma\in(0,1/2)$. Indeed, we only need to compute 
$$
\partial_\sigma \mathbf{b}(\sigma) = - 8\frac{\sigma(2 \sigma+3)}{ \sqrt{3 - 4 \sigma -4\sigma^2}}<0
$$
(for $\mathbf{a}$ is obvious). 
Therefore, $ \mathcal{X}_+(\sigma;L)$ is an strictly decreasing function (with respect to $\sigma$) and has at most a unique zero. Using that $\mathbf{a}\left (\frac{\sl}{2}\right )=0$ we obtain that 
$$
     \mathcal{X}_+\left (\frac{\sl}{2} ;L\right )=   \frac{1}{12}d_1 \alpha^3 L^4 \mathbf{b}\left (\frac{\sl}{2}\right )>0
$$
and 
$$
     \mathcal{X}_+\left (\frac{1}{4};L \right )\leq \mathbf{a}\left (\frac{1}{4}\right ) + \frac{1}{12}d_1 \alpha_{\max}^3 L_{\max}^4 \mathbf{b}\left (\frac{1}{4}\right )= -3.744000070111437<0,
     $$
with $\alpha_{\max}=a_{\max}/a_{\rM}$, 
and by Bolzano theorem we obtain the conclusion for $\mathcal{X}_+$.  
 
Concerning $\mathcal{X}_-(\sigma;L)$, we have that 
     $$
     \mathcal{X}_-\left (\frac{\sl}{2};L\right )=-\frac{1}{12}d_1 \alpha^3 L^4 \mathbf{b}\left (\frac{\sl}{2}\right )<0
     $$
and
$$
\mathcal{X}_-(0;L)=\beta(L) - \frac{1}{12}d_1 \alpha^3L^4 3 \sqrt{3} =1+ \frac{\alpha^3 L^4}{4}  \left (  d_0 - d_1  \sqrt{3} \right ).
$$
Therefore, since $d_0-d_1 \sqrt{3} \sim -17.4387  \rho <0$, 
$$
\mathcal{X}_-(0;L) \geq 1 + \frac{\alpha_{\max}^3 L_{\max}^4}{4 }  
\left (  d_0 - d_1  \sqrt{3} \right )\geq  0.477808830620940>0.
$$ 
Again there exists $\sigma_-(L) \in \left (0,\frac{\sl}{2}\right )$ such that $\mathcal{X}_-(\sigma_-(L);L)=0$. In addition, one can easily check that 
$\partial_\sigma \mathcal{X}_- (\sigma;L)<0$ if $\sigma \in \left (0,\frac{\sl}{2}\right )$
so that $\sigma= \sigma_-(L)$ is the only solution of $\mathcal{X}_-(\sigma;L)=0$ belonging to $\left (0 ,\frac{\sl}{2}\right )$. 

To finish with this analysis, we point out that, for $\sigma\in \left ( \frac{\sl}{2}, \frac{1}{2}\right )$,
$$
\mathcal{X}_-(\sigma;L) = \beta(L) \mathbf{a}(\sigma) - \frac{1}{12} d_1 \alpha^3 L^4 \mathbf{b}(\sigma)<0
$$
provided $\beta(L)>0$, $\mathbf{a}(\sigma)<0$ and $\mathbf{b}(\sigma)>0$ if $\sigma\in \left ( \frac{\sl}{2}, \frac{1}{2}\right )$.
\end{proof}

From formula~\eqref{expr:eigenvaluese=0} of the eigenvalues and the previous lemma, recalling that $\Gamma =L \sigma$, we have that 
\begin{itemize}
    \item If either $\Gamma \in (0, L \sigma_-(L))$ or $\Gamma \in \left (L \sigma_+(L), \frac{L}{2} \right )$, then $(\xi,\eta)=(0,0)$ is a center equilibrium point.
    \item If  $\Gamma \in (L \sigma_-(L), L \sigma_+(L))$, then $(\xi,\eta)=(0,0)$ is a saddle equilibrium point.
    \item Using that $\mathcal{X}_+(\cdot;L)$ and $\mathcal{X}_-(\cdot;L)$ are not zero simultaneously, we deduce that, when $\Gamma=L \sigma_-(L)$ or $\Gamma=L\sigma_+ (L)$, then $(\xi,\eta)=(0,0)$ is a degenerated  equilibrium point with nilpotent linear part.  
\end{itemize}
Therefore, in order to prove Theorem~\ref{thm:origin_average}, we only need to prove the expansions of $\sigma_{\pm}(L)$ in~\eqref{thm:exprGamma12}. Indeed, $\sigma_{\pm}(L)$ satisfy
\begin{equation}\label{proof:eqsigma}
\mathcal{X}_{\pm}(\sigma_\pm(L);L)=\beta(L) \mathbf{a}(\sigma_{\pm}(L))\pm \frac{1}{12} d_1 \alpha^3 L^4 \mathbf{b}(\sigma_{\pm}(L))=0. 
\end{equation}
We rewrite condition~\eqref{proof:eqsigma} in a more suitable way. To do so, denoting $\delta=\rho \alpha^3 L^4$, we introduce the functions
$$
\mathbf{A}_{\pm}(\sigma ; \delta):= \mathbf{a}(\sigma) \left (1+ \frac{U_{2}^{0,0}}{2} \delta \right ) \pm \frac{5 U_{2}^{1,0}}{3} \delta \mathbf{b}(\sigma),
$$
which are smooth for $(\sigma,\delta) \in \left (0,\frac{1}{2}\right ) \times (0,\infty)$. Then, by the definition of $d_0,d_1$ in~\eqref{def:d0d1} and also~\eqref{def:rho}, we have that, 
$\mathcal{X}_\pm (\sigma;L)= \mathbf{A}_{\pm} (\sigma;\delta)$ so that
expression~\eqref{proof:eqsigma} is equivalent to
$$
\mathbf{A}_{\pm}(\sigma_{\pm}(L); \delta) =0.
$$ 
Since $\mathbf{A}_{\pm}\left (\frac{\sl}{2};0\right)=0 $ and
$$
\partial_\sigma \mathbf{A}_{\pm} 
\left (\frac{\sl}{2};0\right) = \partial_\sigma \mathbf{a} \left (\frac{\sl}{2} \right ) = -16 -20 \sl  \neq 0,
$$
by the Implicit Function Theorem, $\sigma_{\pm}(L)$ is, in fact, a smooth function of $\delta = \rho \alpha^3 L^4$, namely $\sigma_{\pm}(L)= \tilde{\sigma}_{\pm}(\delta)$ with $\tilde{\sigma}(0)= \frac{\sl}{2}$. Therefore, 
$$
\sigma_{\pm}(L)= \tilde{\sigma}_{\pm}(\delta)= \frac{\sl}{2} + c_{\pm} \delta + \mathcal{O}(\delta^2)
$$
with 
\begin{align*}
c_{\pm} &= \partial_\delta \tilde{\sigma}_{\pm}(0) = -\frac{\partial_\delta \mathbf{A}_{\pm}\left (\frac{\sl}{2};0\right )}{\partial_\sigma\mathbf{A}_{\pm}\left (\frac{\sl}{2};0\right )} = \frac{1}{16 +20 \sl } \left (
\mathbf{a}\left (\frac{\sl}{2}\right )  \frac{U_{2}^{0,0}}{2} \pm \frac{5 U_{2}^{1,0}}{3}  \mathbf{b}\left (\frac{\sl}{2} \right ) \right ) \\
&= \pm \frac{1}{16+20 \sl} \frac{5 U_{2}^{1,0}}{3} \mathbf{b} \left (\frac{\sl}{2} \right ).
\end{align*}
 
Using that, by Lemma~\ref{lem:previsigmapm},  $\mathcal{X}_+$ is a strictly decreasing function (with respect to $\sigma$) and that $\mathcal{X}_-$ is also a strictly decreasing function if $\sigma\in \left (0,\frac{\sl}{2}\right )$, we conclude that $\tilde{\sigma}_{\pm}$ are well defined for $\delta \in \big [0, \rho \alpha_{\max} L_{\max}^4 \big ]$ and the proof of Theorem~\ref{thm:origin_average} is complete.

\subsection{Modulus of the eigenvalues: Proof of Proposition~\ref{prop:origin_modulus_eigenvalue}}\label{sec:modulus}   

We use, along this section, the notation introduced in Section~\ref{sec:circular_case_Appendix}. 
We first define, from~\eqref{expr:eigenvaluese=0} and~\eqref{Xxietasigma},  
\begin{equation}\label{defMcalmodulus} 
\mathcal{M}(\sigma;L) :=    
- g^2(L) \frac{9L^6}{4} \left [ \beta^2 (L) \mathbf{a}^2(\sigma) -  \widetilde{\beta}^2(L) \mathbf{b}^2 (\sigma) \right ] ,
\end{equation}
with $g(L)=\frac{\rho_0}{8L^{10}}$, $\widetilde{\beta}(L)=\frac{d_1 \alpha^3 L^4}{12}$ and $\mathcal{X}_\pm$ were introduced in Lemma~\ref{lem:previsigmapm}. 
It is not difficult to check that 
\begin{align*}
\partial_\sigma \mathcal{M}(\sigma;L) & = 
\frac{-9 \rho_0^2}{256 L^{14}} \left [ 2 \beta^2(L) \mathbf{a}(\sigma) \mathbf{a}'(\sigma) - 2\widetilde{\beta}^2(L)
 \mathbf{b}(\sigma) \mathbf{b}'(\sigma) \right ]
\\ & =\frac{-9 \rho_0^2}{256 L^{14}} \left [ \beta^2(L) (1600 \sigma^3 + 1920 \sigma^2 + 432 \sigma - 32) + \widetilde{\beta}^2(L) (16 \sigma (2\sigma +3)^2)\right ]. 
\end{align*}
Therefore $\partial_\sigma \mathcal{M}$, for any fixed $L$, is a degree three polynomial such that its derivative has no positive zero (all its coefficients are positive). That implies that $\partial_\sigma \mathcal{M}(\cdot;L)$ can at most have one zero for any value of $L$. 
Since, by definition, $\mathcal{M}(\sigma_+(L);L)=\mathcal{M}(\sigma_-(L);L)=0$, by Rolle's theorem, there exists $\sigma_*(L)\in \big (\sigma_-(L), \sigma_+(L) \big )$ such that $\partial_\sigma \mathcal{M}(\sigma_*(L);L)=0$. Moreover, since $\mathbf{a} \left (\frac{\sl}{2}\right )=0$, 
$$
\partial_\sigma \mathcal{M}(0;L) >0,\qquad \partial_{\sigma} \mathcal{M}\left (\frac{\sl}{2};L  \right )<0,
$$
we conclude that, for any $L_{\min}\leq L \leq L_{\max}$, the function $\mathcal{M}(\sigma;L)$ has a maximum at $\sigma_*(L)$. Therefore, for a fixed value of $L$, 
\begin{itemize}
\item on the interval $[\sigma_-(L ), \sigma_+(L )]$, the function $ \lambda_+(\sigma;L)=\sqrt{\mathcal{M}(\sigma;L)}$ has a maximum at $\sigma=\sigma_*(L)$.  
\item When $\sigma\in (0,\sigma_-(L))$ the function 
$\sqrt{\big |\mathcal{M} (\sigma;L)\big |}$ is decreasing with respect to $\sigma$
\item and when $\sigma \in \left (\sigma_+(L),\frac{1}{2} \right )$,  $\sqrt{\big |\mathcal{M} (\sigma;L)\big |}$ is an increasing function. 
\end{itemize}

To finish the proof of Proposition~\ref{prop:origin_modulus_eigenvalue}, we notice that, from~\eqref{defMcalmodulus} and using that $\mathbf{b}$ is a decreasing positive function, we deduce that, if $\sigma \in (\sigma_-(L), \sigma_+(L))$, then the corresponding eigenvalues $\pm \lambda(\sigma;L)$ satisfy
$$
0<\lambda (\sigma;L) \leq g(L) \frac{3 L^3}{24 } d_1  \alpha^3 L^4 \mathbf{b}(0) \leq  \frac{\rho_0}{8L^{10}}  \frac{3 L^3}{24} d_1  \alpha^3 L^4 3\sqrt{3}
= \frac{ 15 \sqrt{3}\rho_0}{16 L^7} U_2^{1,0} \rho \alpha^3 L^4, 
$$
with $\rho=\rho_1/\rho_0$. 
In addition, since $\lambda(\sigma_*(L);L)>\lambda\left (\sigma\left (\frac{\sl}{2}\right );L\right )$,
$$
\lambda (\sigma_*(L);L) \geq  g(L) \frac{3 L^3}{24 } d_1  \alpha^3 L^4 \mathbf{b}\left (\frac{\sl}{2}\right ) = \frac{5 \rho_0}{16L^{7}}U_{2}^{1,0} \rho \alpha^3 L^4 (\sl+3) \sqrt{3-2\sl -\sl^2}.
$$

\begin{remark} As a consequence of our study, if we want to analyze the value of $\mathcal{M}(\sigma;L )$ at some interval $\sigma \in [\sigma_{\min}, \sigma_{\max}] \subset [\sigma_-(L), \sigma_+(L)]$, its minimum value is located either at $\sigma= \sigma_{\min}$ or $\sigma= \sigma_{\max}$. 
\end{remark}

\begin{remark}
Using $L=\sqrt{\mu a}$ and $\alpha= \frac{a}{a_M}$, we write from~\eqref{expr:eigenvaluese=0}, $\widetilde{\mathcal{M}}(\sigma;a):= \mathcal{M}(\sigma;\sqrt{\mu a})$ as 
\begin{equation*}
\begin{aligned}
\widetilde{\mathcal{M}}(\sigma;a)&:= -\frac{9 \rho_0^2 }{256 L^{14}} \left [ \left (1+ \frac{d_0\alpha^3 L^4}{4} \right )^2\mathbf{a}^2 (\sigma) - \frac{1}{144} d_1^2 \alpha^6 L^8 \mathbf{b}^2 (\sigma) \right ]\\ 
&= 
-\frac{9 \rho_0^2}{256 \mu^7 a^7} \left [ 
\mathbf{a}^2(\sigma) \left (1 + \frac{d_0 a^5 \mu^2}{2 a_M^3} + \frac{d_0^2 a^{10} \mu^4}{16 a_M^6}\right ) 
- \frac{1}{144}d_1^2 \mathbf{b}^2 (\sigma) \frac{a^{10} \mu^4}{a_M^6} \right ]
 \\
&= -\frac{9 \rho_0^2}{256 \mu^7 } 
\left [ 
\mathbf{a}^2(\sigma)   
\left ( \frac{1}{a^7} +\frac{d_0\mu^2 }{2 a^2 a_M^3} + \frac{d_0^2 a^3 \mu^4}{16 a_M^6} \right ) 
- \frac{1}{144}d_1^2 \mathbf{b}^2 (\sigma) \frac{a^3 \mu^4}{a_M^6} \right ].
\end{aligned}
\end{equation*}
We have that $\partial_a \widetilde{\mathcal{M}}(\sigma;a)>0$ provided that, for $a\leq a_{\max}$,
$$
-\frac{7}{a^8} - \frac{d_0 \mu^2}{  a^3  a_M^3} + \frac{3d_0^2 a^2 \mu^4}{16 a_M^6 }  = -\frac{1}{a^8} \left (7 + \frac{d_0 \mu^2}{ a_M^3} a^5 - \frac{3d_0^2  \mu^4}{16 a_M^6 } (a^5)^2  \right )<0
$$
(the associated degree two polynomial has a unique positive zero which is greater than $a_{\max}$). 
Therefore, for $a_{\min}\leq a\leq a_{\max}=30000$ we have that 
$$
0<\widetilde{\mathcal{M}}(\sigma;a_{\min}) \leq \widetilde{\mathcal{M}}(\sigma;a) \leq \widetilde{\mathcal{M}}(\sigma;a_{\max}).
$$
\end{remark}

We present now some numerical results. In Figure \ref{fig:lambdas} it is depicted $\lambda_+(\sigma;L)>0$ as a function of   $\sigma$ for some values of $a$.


We can also  compute numerically $\sigma_\pm$ and $\sigma_*$. We compute $\sigma_\pm$ using that $\mathcal{X}_\pm (\sigma_\pm;L)=0$ and $\sigma_*$ as the unique positive zero of 
$$
\beta^2(L) (1600 \sigma^3 + 1920 \sigma^2 + 432 \sigma- 32) + \widetilde{\beta}^2(L) (16\sigma (2\sigma+3)^2).
$$
We show its values in Table \ref{table:sigmas}.

\begin{center}
\begin{figure}
\begin{center}
\subfloat{\includegraphics[width=0.3\textwidth]{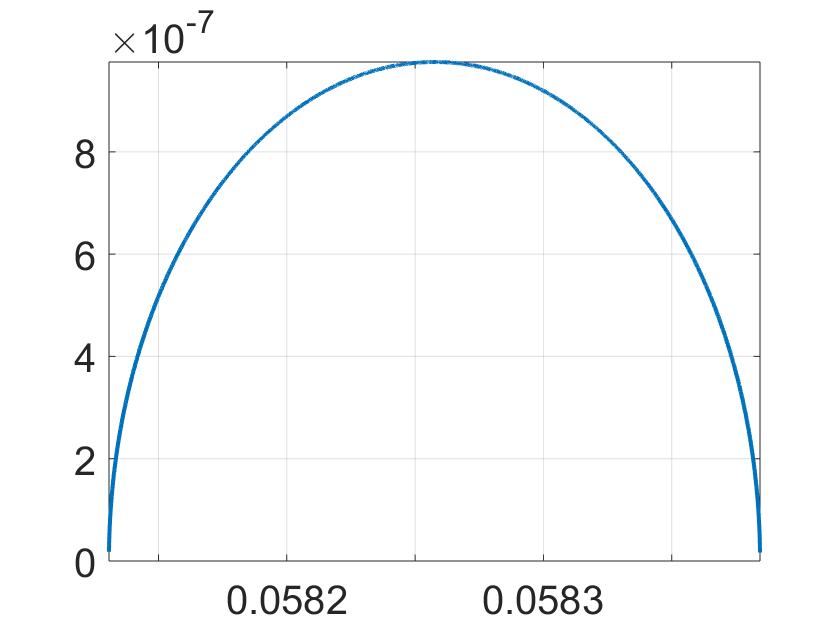}} 
\subfloat{\includegraphics[width=0.3 \textwidth]{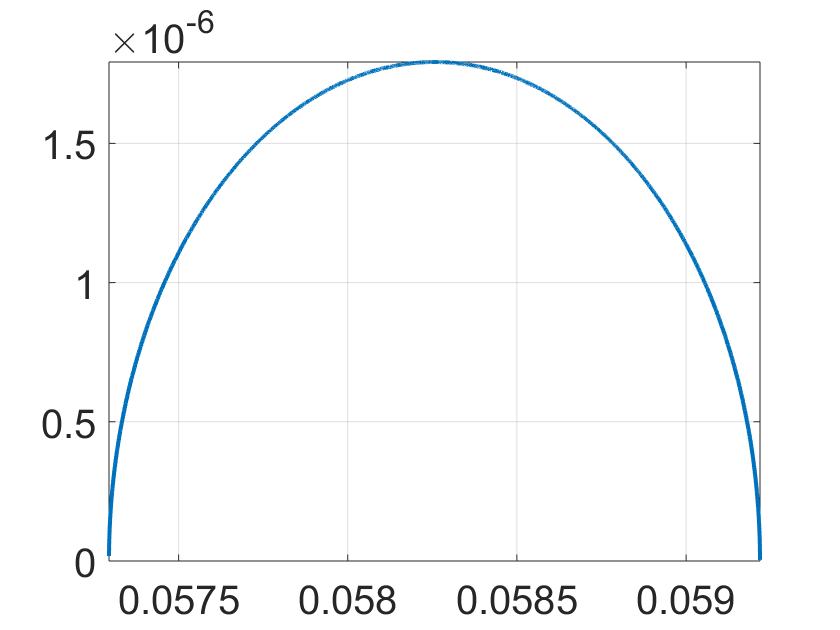}} 
\subfloat{\includegraphics[width=0.3\textwidth]{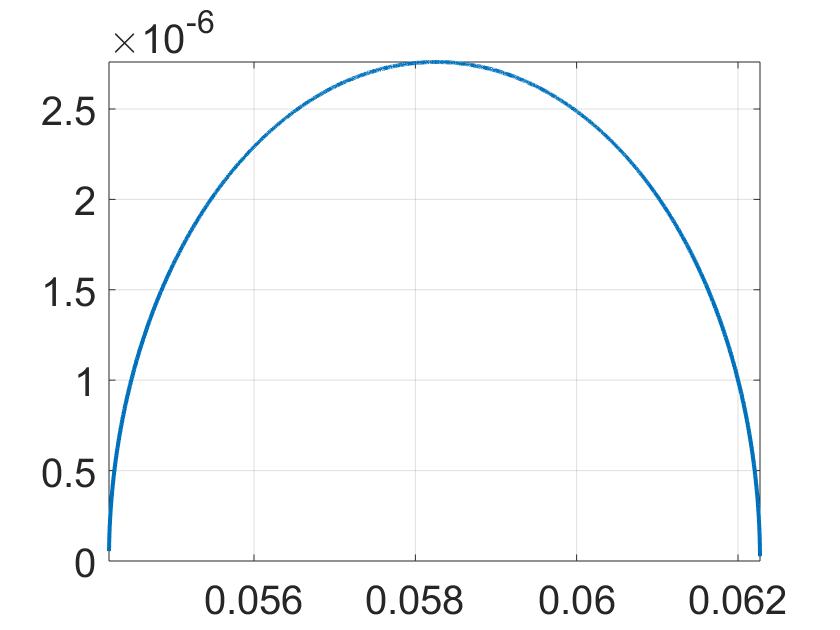}}\\
\subfloat{\includegraphics[width=0.3\textwidth]{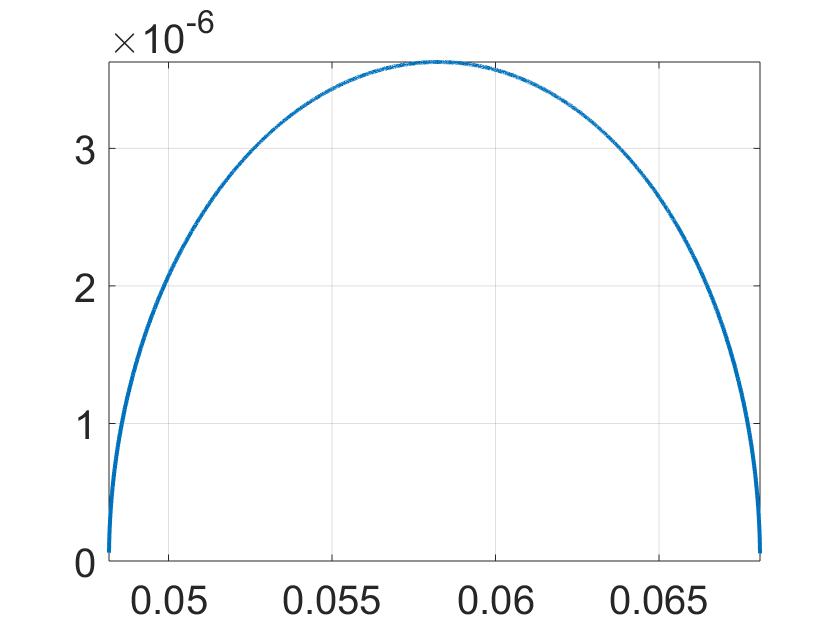}}
\subfloat{\includegraphics[width=0.3 \textwidth]{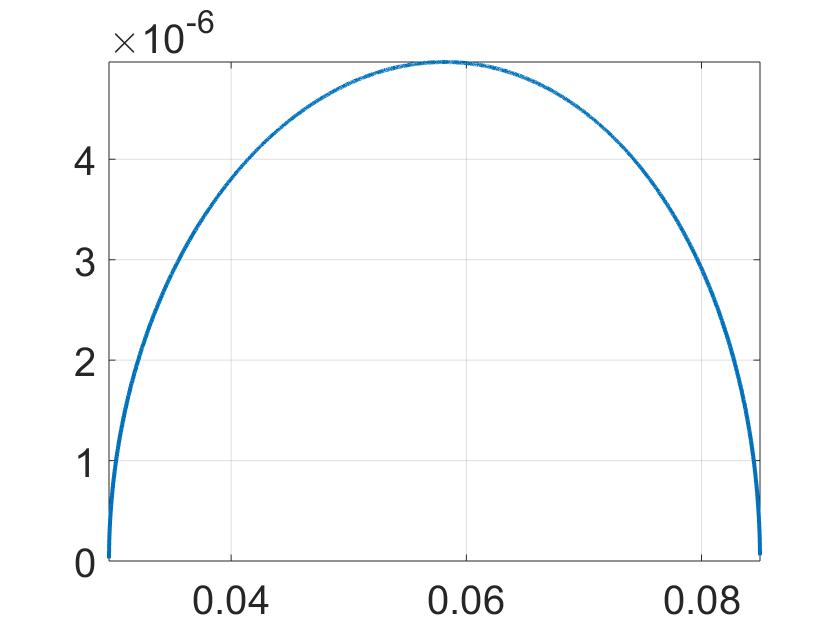}} 
\subfloat{\includegraphics[width=0.3\textwidth]{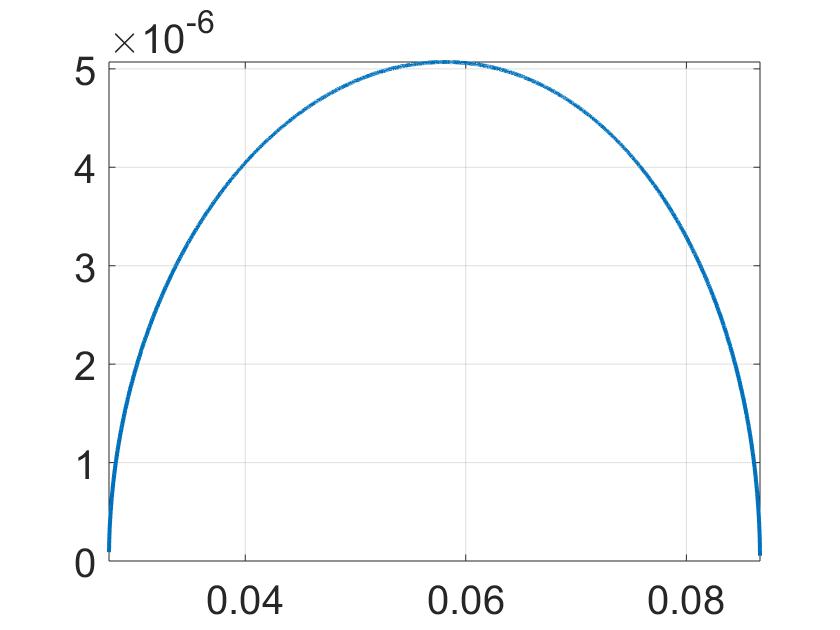} } 
\end{center}
\caption{These figures show $\lambda_+(\sigma;L)>0$ as functions of $\sigma$ for $a=\{10000, 15000, 20000, 24000, 29600, 30000\}\, \text{km}$, when the origin is a saddle point (the units are the ones made explicit in Remark~\ref{rmk_unities}).}
\label{fig:lambdas}
\end{figure}
\end{center}

\begin{table}[ht]
\centering
\begin{tabular}{c|c|c|c}
$a$ & $ \sigma_-(L)$ & $ \sigma_+(L)$ & $\sigma_*(L)$ \\ 
$ 10000$ & 0.058130693719535  & 0.058384404707503 & 0.058257569492108
\\ 
$ 15000$ & 0.057294284878483  &  0.059218520762404 & 0.058257569455989
\\
$ 20000$ & 0.054201084378365  & 0.062272995357407 & 0.058257569273126
\\ 
$ 24000$ & 0.048177348389369  &  0.068087985206979 & 0.058257568831379
\\ 
 29600 & 0.029613649805289   & 0.084971418151141 & 0.058257567158286
\\ 
$ 30000$ & 0.027639200647529 & 0.086680627913961 & 0.058257566962293
\end{tabular}
\caption{Values of $\sigma_\pm$ and $\sigma_*$ for some choices of semi-major axis.}
  \label{table:sigmas}
\end{table}

\begin{remark}
We emphasize that the values of $\sigma_*(L)$ are close to  $\frac{\sl}{2}=0.058257569495584
 $ which is the only positive zero of the polynomial $ \mathbf{a}(\sigma) \mathbf{a}'(\sigma)=1600 \sigma^3 + 1920 \sigma^2 + 432 \sigma -32$. This is because $\beta(L) \sim 6$ for all the values of semi-major axis $a$ considered, meanwhile the range of values for $\widetilde{\beta}(L)$ goes from $\mathcal{O}(10^{-10})$ to $\mathcal{O}(10^{-7})$. Therefore, the zeros of $\partial_\sigma \mathcal{M}$ are close to the ones of $\mathbf{a}(\sigma) \mathbf{a}'(\sigma)$.  
\end{remark}

To finish, in Table \ref{table:lambdamax}, we present the maximum value of the positive eigenvalue in the saddle case, namely $\lambda_+(\sigma_*(L);L)$.
\begin{table}[ht]
\begin{center}
\begin{tabular}{c|c} 
$a$ & $\lambda_+(\sigma_*(L);L)$  \\  \hline

$10000$ &  $0.97568857909377 \cdot 10^{-6}$
\\ 
$15000$ & $1.79245437499162 \cdot 10^{-6}$
\\ 
$20000$ & $2.75966404251617 \cdot 10^{-6}$\\ 
$24000$ & $3.62767259360406 \cdot 10^{-6}$\\
$29600$ & $4.96876878785117 \cdot 10^{-6}$\\ 
$30000$ & $5.06982657623624 \cdot 10^{-6}$.
\end{tabular}
\end{center}
\caption{Maximum value of the positive eigenvalue in the saddle case for some values of the semi-major axis $a$.}
  \label{table:lambdamax}
\end{table}

\section{Eccentric critical points of the $h$-averaged system}\label{sec:eccentric_case}  

We devote this section to prove Theorem  \ref{thm:AveragedHam_generalcase}   and analyze numerically the results obtained in this theorem. Along this section, we will use the notation introduced in Section~\ref {sec:goodexpresion} without any explicit mention. 

The first step  to  compute the location of the equilibrium points of the  $h$-averaged problem~\eqref{Eq:EMotMoy} is to analyze the equation $\dot{y} = 0$. It  implies that a fixed point exists if it satisfies one of the following conditions:	
	\be
		x=k\pi,		\quad k=0,1
			\qtext{or}	
		\sqrt{P_1}Q_1 = 0.
\label{Eq:condFP}
	\ee 
However, since we are assuming that $y\in(0,L/2)$ the  latter condition cannot happen. That is, all critical points of~\eqref{Eq:EMotMoy} must be of the form 
\[
(k\pi,y) \quad \text{with}\quad y\in(0,L/2).
\]
\subsection{Numerical approach}
For  given $L$ and  $\Gam$, equation~\eqref{Eq:condFP} implies that an eccentric orbit can be a critical point if it satisfies the condition  $x=0,\pi$. The location of the  critical point (equivalently its associated eccentricity for a given $\Gam$) is obtained by solving the equations
\be
\begin{split}
\sqrt{P_1(y, \Gam)}\cB_{0,\alpha}(y,\Gam;L) 
+ \alpha^3 \cB_{1}(y,\Gam;L)&= 0, \qtext{for $x = 0$,}\\
\sqrt{P_1(y,\Gam)}\cB_{0,\alpha}(y,\Gam;L) 
- \alpha^3 \cB_{1}(y,\Gam;L)&= 0,\qtext{for $x = \pi$,}
\end{split}
\label{eq:LocEq}
\ee
where $P_1$, $\cB_{0, \alpha}$ and $\cB_{1}$ are polynomials. In other words, they are roots of the following polynomial in $(y,\Gam)$
\begin{equation*}
\begin{split}
P_1(y,\Gam)(\cB_{0,\alpha}(y,\Gam;L))^2	-\alpha^6(\cB_{1}(y,\Gam;L))^2	=	0.								
\end{split}
\end{equation*}
A way to obtain the solutions is to fix $y$ such that  $0< y/L < 1/2$, compute the roots of a polynomial of degree 6 in $\Gam$ and select the ones for which $0<\Gam<L/2$ and~\eqref{eq:LocEq} is satisfied. The results, computed for $a = 19000\, \mbox{km}$, $a = 24000\, \mbox{km}$, $a = 25450\, \mbox{km}$ and $a = 29600\, \mbox{km}$, are depicted in the left panel of the Figure~\ref{fig:Gamhy} in the action space $(\Gamma/L, y/L)$-plane. 

The stability of the two families of equilibria can be computed from the Hessian matrix of the Hamiltonian. More precisely, for fixed $L$ and $\Gam$ and a given fixed point $(x_0, y_0)$, the variational equations read
\begin{equation*}
\begin{pmatrix}
\dot{x}\\
\dot{y}
\end{pmatrix}
=	
M^{\Gam,L}(x_0, y_0) \begin{pmatrix}
x\\ y \end{pmatrix} \qtext{with} M^{\Gam,L} = 			\begin{pmatrix}
\frac{\partial^2\cH^{\Gam,L}}{\partial y \partial x}
&
\frac{\partial^2\cH^{\Gam,L}}{\partial y^2}\\
-\frac{\partial^2\cH^{\Gam,L}}{\partial x^2}
&-\frac{\partial^2\cH^{\Gam,L}}{\partial x\partial y}
\end{pmatrix}.
\end{equation*}
Since $\mathrm{tr}{M^{\Gam,L}} = 0$, $M^{\Gam,L}$ possesses two eigenvalues $(\lam, -\lam)$ with either $\lam\in i\RR$ or $\lam \in \RR$. Hence, the sign of $\det{M^{\Gam,L}} = -\lam^2$ characterizes the linear stability of the fixed point:
\begin{itemize}
\item hyperbolic fixed point (saddle) for $\det{M^{\Gam,L}}<0$;
\item elliptic fixed point (center) for $\det{M^{\Gam,L}}>0$.
\end{itemize} 	
The results are depicted in the right panel of Figure~\ref{fig:Gamhy}. 

 \begin{figure}
\begin{center}
\includegraphics[width=6.5cm]{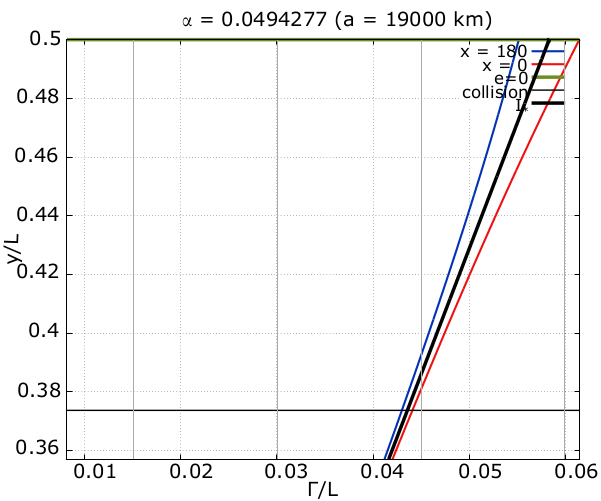}
\includegraphics[width=6.5cm]{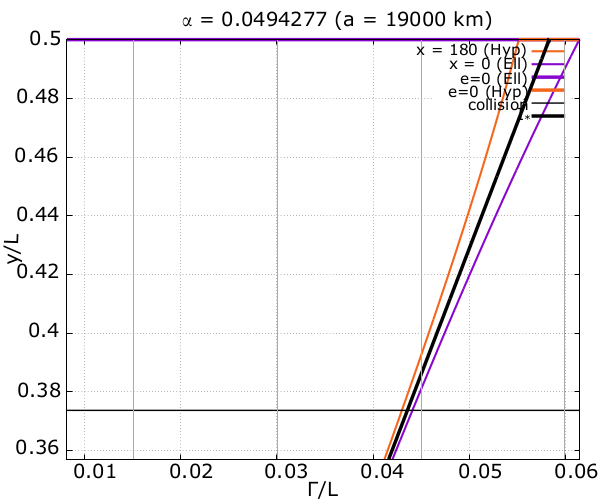}
\includegraphics[width=6.5cm]{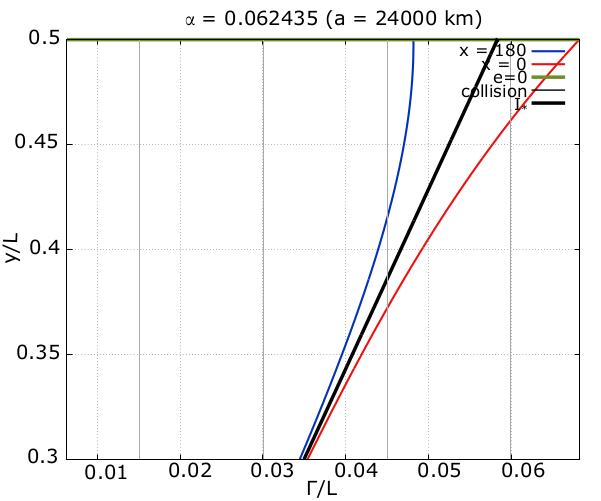}
\includegraphics[width=6.5cm]{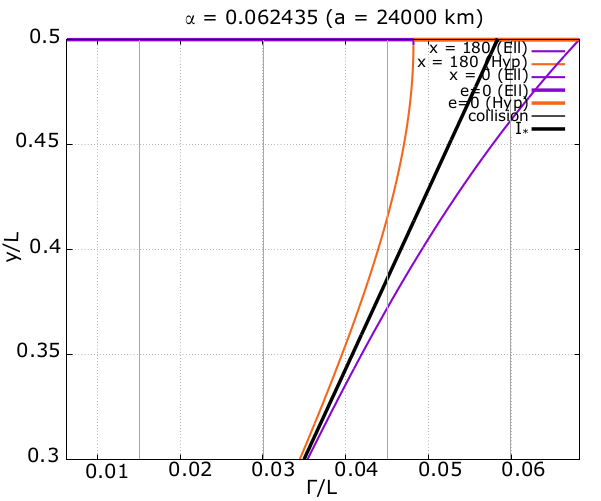}
\includegraphics[width=6.5cm]{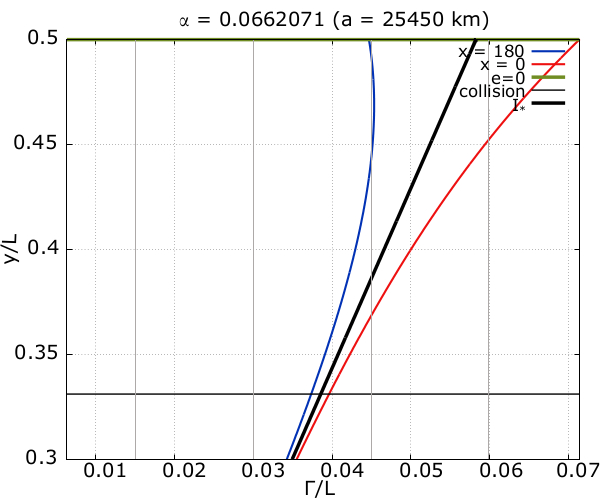}
\includegraphics[width=6.5cm]{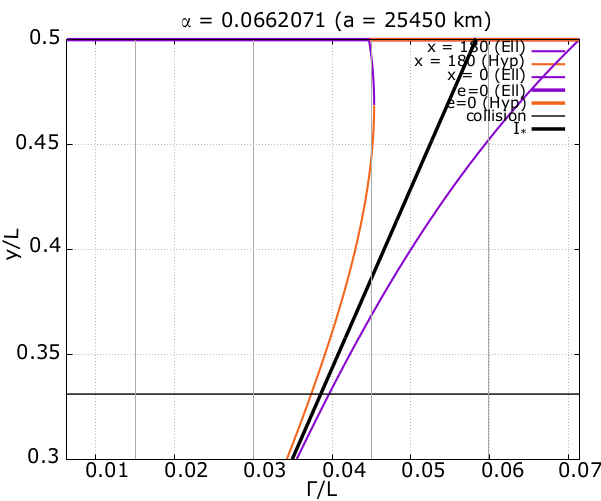}
\includegraphics[width=6.5cm]{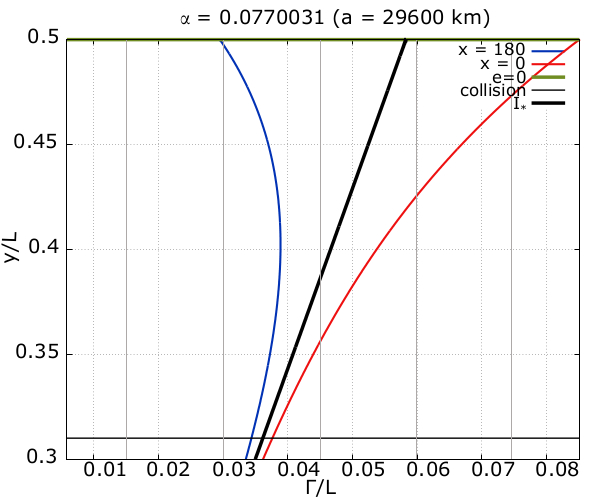}
\includegraphics[width=6.5cm]{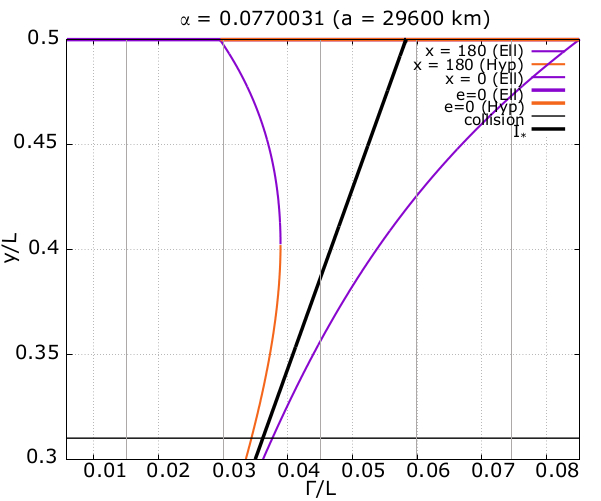}
\end{center}
\caption{Evolution of the fixed point location (left) and stability (right) in the $(\Gamma/L, y/L)$-plane (Action space). In the plots in the first column, the blue and red curves correspond respectively to the fixed point families for $x=0$ and $x=\pi$. In the plots in the second column, the violet and orange curves correspond respectively to the elliptic and hyperbolic part of the $x=0$ and $x=\pi$ fixed point families. 
\label{fig:Gamhy}
}
\end{figure}	

\subsection{Existence of eccentric critical points: {Formulation of the problem}}
To prove Theorem~\ref{thm:AveragedHam_generalcase}, we have to  analyze the existence and character of the fixed points of the form $(0,y)$ and $(\pi,y)$ with $0<y<\frac{L}{2}$.

First, let us now  rewrite condition~\eqref{eq:LocEq} on the fixed points. Using  the decomposition in~\eqref{decompositioncalB0alpha},  the condition~\eqref{eq:LocEq} can be rephrased as
\begin{equation}\label{exprmathcalFfixedpoints}
\mathcal{F}_{\pm}(y,\Gamma;L):=B_0(y,\Gamma;L) + \alpha^3 B_1(y,\Gamma;L) \pm \alpha^3 \frac{\cB_1(y,\Gamma;L)}{\sqrt{P_1(y,\Gamma)}} =0
\end{equation}
where the sign $+$ corresponds to $x=0$ and  the $-$ sign corresponds to $x=\pi$. 

We introduce the scaling 
$$
y=L\hat{y},\qquad \Gamma=y \gg = L \hat{y}\gg.
$$
The values of the variables and parameters we are interested in are $0<\hat{y}  < \frac{1}{2}$, $a_{\min}\leq a\leq a_{\max}$ and $\gg \in (0,\frac{1}{2} )$ 

We rewrite $\mathcal{F}$ in these new variables $(s,\hat{y})$, namely $\widehat{\mathcal{F}}_{\pm}(\gg,\hat{y};L):=\mathcal{F}_{\pm}(L\hat{y},L\hat{y}\gg;L)$. 
Tedious but easy computations lead to
\begin{align*}
\widehat{\mathcal{F}}_\pm (s,\hat{y};L)= &- 6L^3\hat{y}^2( 5 \gg^2 + 8\gg -1)
+12 d_0 \alpha^3 L^7 \hat{y}^5 (4 \hat{y}^2 - 5 \gg - 12 \gg\hat{y}^2 - 5 \gg^2) \\
&\pm 4 d_1 \alpha^3 L^7 \hat{y}^5 \frac{(3+\gg) (12 \hat{y}^2 - 8 \gg \hat{y}^2 - \gg^2)}{\sqrt{3- 2\gg - \gg^2}}
\end{align*} 
where $d_0$ and $d_1$ were introduced in~\eqref{def:d0d1}. 
We define 
\begin{align*}
\mathcal{G}_\pm(\gg,\hat{y};L):= & 5\gg^2 + 8\gg -1  +  2 d_0 \alpha^3 L^4 \hat{y}^3 (5 \gg^2 + 12 \gg\hat{y}^2 + 5 \gg - 4 \hat{y}^2) \\ &  \pm 2 d_1 \alpha^3 L^4 \hat{y}^3 
\frac{(3+\gg)(\gg^2 + 8 \gg\hat{y}^2 - 12 \hat{y}^2)}{3\sqrt{3- 2\gg - \gg^2}} .
\end{align*}
Then, recalling that $y=L\hat{y}$ and $\Gamma=ys=L\hat{y}s$, 
\begin{equation}\label{equivFG_averaged}
 \mathcal{F}_{\pm}(y,\Gam;L)=0 \Longleftrightarrow  \widehat{\mathcal{F}}_{\pm}(s,\hat{y};L)=0     \Longleftrightarrow \mathcal{G}_{\pm}(s,\hat{y};L)=0. 
\end{equation}
The function $\mathcal{G}_\pm$ can be expressed as 
\begin{equation}\label{expr:mathcalG}
\mathcal{G}_{\pm}(\gg,\hat{y};L)=\mathbf{a}(\gg) + 2\alpha^3 L^4 \rho [\mathbf{b}_\pm (\gg) \hat{y}^3 + \mathbf{c}_\pm(\gg) \hat{y}^5 ] 
\end{equation}
with 
\begin{align*}
    \mathbf{a}(\gg)&= 5  {\gg}^2 + 8 {\gg} -1     \\
    \mathbf{b}_\pm (\gg)&=  10U_2^{0,0}  {\gg} (1+  {\gg}) \pm \frac{20 U_2^{1,0} (3+ {\gg})}{3 \sqrt{3 - 2{\gg}-{\gg}^2}}  {\gg}^2 
\\ 
    \mathbf{c}_\pm(\gg) &=  8U_2^{0,0} (3 {\gg}-1 ) \pm \frac{20 U_2^{1,0} (3+  {\gg})}{3 \sqrt{3 - 2{\gg} -  {\gg}^2 }} ( 8 {\gg}-12). 
\end{align*}

\begin{lemma}\label{lem:bcposneg}
The functions $\mathbf{a}$, $\mathbf{b}_{\pm}$, $\mathbf{c}_{\pm}$ are strictly increasing functions at $\left [0,\frac{1}{2} \right ]$. As a consequence, for $\hat{y},L>0$, we have that $\partial_\gg \mathcal{G}_{\pm}(\gg,\hat{y};L)>0$.
 
In addition, $\mathbf{a}(m)=0$, $\mathbf{a}(\gg)<0$ if $\gg\in [0,m)$ and $\mathbf{a}(\gg)>0$ if $\gg >m$. 
\end{lemma}
\begin{proof}
The statements related to $\mathbf{a}$ are immediate. We compute now $\partial_\gg \mathbf{b}_\pm, \partial_\gg \mathbf{c}_\pm$,
\begin{equation*}
\begin{aligned}
    \partial_\gg \mathbf{b}_\pm (\gg)&=   10U_2^{0,0}  (1+  2{\gg}) \pm 20 U_2^{1,0} \frac{ 2 \gg (9-4 \gg^2 -\gg^3)}{3(3 - 2{\gg}-{\gg}^2)^{3/2}}  
\\ 
    \partial_\gg \mathbf{c}_\pm(\gg) &=   24 U_2^{0,0}   \mp 20 U_2^{1,0} \frac{ 8\gg^2 }{3(1-\gg) \sqrt{3 - 2{\gg} -  {\gg}^2 }}   .
\end{aligned}
\end{equation*}
The function $\gg(9-4\gg^2 - \gg^3)$ is an increasing function for $\gg\in \left (0,0.7 \right )$ provided 
\[
9-12 \cdot 0.7^2 -4 \cdot 0.7^3 = 1.748 >0.
\]
Therefore, denoting $s_*=0.35$, if $s\in (0,s_*]$,
\begin{align*}
\left | \pm 20 U_2^{1,0} \frac{ 2 \gg (9-4 \gg^2 -\gg^3)}{3(3 - 2{\gg}-{\gg}^2)^{3/2}} \right | & \leq 
20 U_2^{1,0}  \frac{ 2 \gg_* (9-4 \gg_*^2 -\gg_*^3)}{3(3 - 2{\gg_*}-{\gg_*}^2)^{3/2}}\\ & =6.731992451918273.
\end{align*}
Since $10 U_2^{0,0}  =7.62646$, we conclude that $\partial_\gg \mathbf{b}_\pm(\gg)>0$ if $s\in (0,s_*]$. When $s\in \left (\gg_*, \frac{1}{2} \right ]$,
\[
\left | \pm 20 U_2^{1,0} \frac{ 2 \gg (9-4 \gg^2 -\gg^3)}{3(3 - 2{\gg}-{\gg}^2)^{3/2}} \right | \leq 12.414817621987043
\]
and $10 U_2^{0,0} (1+2s) \geq 10 U_2^{0,0} (1+2s_*)= 12.96498200000000$. Therefore, we conclude that $\partial_\gg \mathbf{b}_\pm (\gg) >0$ if $s\in \left [0,\frac{1}{2}\right ]$. 
On the other hand, if $0<s\leq \frac{1}{2}$, then
\[
    \left |\mp 20 U_2^{1,0} \frac{ 8\gg^2 }{3(1-\gg) \sqrt{3 - 2{\gg} -  {\gg}^2 }} \right |  \leq 20 U_2^{1,0} \frac{ 2 }{\frac{3}{2}  \sqrt{\frac{7}{4}}}   =11.035393441766260
\]
and $24U_2^{0,0} = 18.303504$. Therefore $\partial_\gg \mathbf{c}_{\pm}(\gg)>0$. 

When $s\geq 0$, $\partial_\gg \mathbf{a}(s) = 8 + 10 \gg>0$ and thus $\partial_\gg \mathcal{G}_{\pm} (\gg,\hat{y};L)>0$ provided $\hat{y},L>0$. 
\end{proof}

\subsection{A preliminary existence result}
The next step to compute the eccentric critical points is to solve the equations in \eqref{equivFG_averaged}. We look for zeros of the functions $\GG_\pm$ in $s$.
 
\begin{lemma}\label{lem:existencefixedpointsG}
Let $L$ be satisfying $0<L<L_{\max}$, $L=\sqrt{\mu a}$, and
\begin{equation}\label{def:ymin}
\hat{y}_{\min}= \sqrt{-\frac{\mathbf{b}_+(\sl)}{\mathbf{c}_+(\sl)}}= 0.116589071022807.
\end{equation}
We have that $\mathcal{G}_+(m,\hat{y}_{\min};L)=0$ and 
\begin{enumerate}
\item For any $0<\hat{y}\leq \frac{1}{2}$, there exists a unique analytic function $\gg_{+}(\hat{y};L)$, which belongs to  $\left (0, \frac{3}{2} \sl\right )$, satisfying  $\mathcal{G}_+(\gg_+(\hat{y};L),\hat{y};L)=0$.
In addition, for $\hat{y}_{\min}< \hat{y}\leq \frac{1}{2}$ the function $\gg_{+}(\hat{y};L)$ belongs to $\left (\sl, \frac{3}{2} \sl\right )$ and when $0<\hat{y}<\hat{y}_{\min}$, $\gg_{+}(\hat{y};L) \in (0,m)$.
\item For any $0<\hat{y}\leq \frac{1}{2}$, there exists a unique analytic function $\gg_{-}(\hat{y};L)$, which belongs to  $(\frac{2}{5} \sl ,\sl)$, satisfying
$\mathcal{G}_-(\gg_-(\hat{y};L), \hat{y};L)=0$. 
\end{enumerate}
\end{lemma}
\begin{proof}
To prove the existence of solutions of the equation $\mathcal{G}_{\pm}(\gg,\hat{y};L)=0$, we  fix $\hat{y}, L$ (or $a$) in the ranges $0<\hat{y}\leq \frac{1}{2}$ and $0\leq a\leq a_{\max}$ and we use a Bolzano argument (with respect to the $\gg$ variable). 

We deal first with the $+$ case. Since $\mathbf{b}_+(0)=0$, 
$$
\mathcal{G}_+(0,\hat{y};L)=-1 + 2 \alpha^3 L^4 \hat{y}^5 \mathbf{c}_+(0).
$$
It is clear from the definition of $\mathbf{c}_+$ that 
$\mathbf{c}_+(0)<0$ so that $\mathcal{G}_+(0,\hat{y};L)<0$. 
For fixed $L$, now we prove that the function
$\mathcal{G}_+\left (\frac{3}{2}\sl,\hat{y};L\right )$ is positive. We introduce 
$$
C_+(\hat{y}):=\mathcal{G}_+\left (\frac{3}{2}\sl,\hat{y};L\right )= \hat{\mathbf{a}}  + 2 \alpha^3 L^4 \rho \left  [\hat{\mathbf{b}}_+ \hat{y}^3   + \hat{\mathbf{c}}_+  \hat{y}^5 \right ]
$$
with
\begin{equation}\label{value:bposm32} 
\widehat{\mathbf{a}} =\mathbf{a} \left (\frac{3}{2} \sl\right )=0.550909166052992 ,\quad \widehat{\mathbf{b}}_+=\mathbf{b}_+\left (\frac{3}{2} \sl\right )=1.784508217583372,
\end{equation}
\begin{equation}\label{value:cposm32}
\widehat{\mathbf{c}}_+=\mathbf{c}_+\left (\frac{3}{2} \sl\right )=-78.794344798569341.
\end{equation}
For a fixed $L$, the function $C_+(\hat{y})$ has the minimum value either at the points $\hat{y}=0 ,\hat{y}=\sqrt{\frac{-3\widehat{\mathbf{b}}_+}{5\widehat{\mathbf{c}}_+}}= 0.116570155823739$ or $\hat{y}=\frac{1}{2}$. Evaluating at these points, we obtain that
$C_+(0)=\widehat{\mathbf{a}} >0$, from the fact that $\widehat{\mathbf{c}}_+<0$:
\[
C_+\left (\sqrt{\frac{-3\widehat{\mathbf{b}}_+}{5\widehat{\mathbf{c}}_+}} \right )=\widehat{\mathbf{a}} - 2 \alpha^3L^4 \rho \frac{6 \widehat{\mathbf{b}}_+^2}{25 \widehat{\mathbf{c}}_+ }\sqrt{\frac{-3\widehat{\mathbf{b}}_+}{5\widehat{\mathbf{c}}_+}}>0
\]
and finally, using that $4 \widehat{\mathbf{b}}_+ + \widehat{\mathbf{c}}_+<0$, 
\begin{align}  
C_+ \left (\frac{1}{2}\right ) & = \widehat{\mathbf{a}} + \alpha^3 L^4  \rho \frac{1}{16 } (4 \widehat{\mathbf{b}}_+ + \widehat{\mathbf{c}}_+) \geq \widehat{\mathbf{a}}+ \rho \frac{  a_{\max}^5 \mu^2}{16 a_M^3}   (4\widehat{\mathbf{b}}_+ + \widehat{\mathbf{c}}_+) \label{def:C12}\\ 
&= 0.014481715549738>0. \notag 
\end{align} 
The previous analysis proves that 
\[\mathcal{G}_+(0,\hat{y};L) \cdot \mathcal{G}_+\left (\frac{3}{2} \sl , \hat{y};L \right ) <0
\]
and therefore, there exists $\gg_+ =\gg_+(\hat{y};L)\in \left (0,\frac{3}{2} \sl \right )$ such that $\mathcal{G}_+(\gg_+,\hat{y};L)=0$.

By Lemma~\ref{lem:bcposneg}, $\partial_\gg \mathcal{G}_+(\gg, \hat{y};L)>0$ 
so that, for fixed $\hat{y},L$, $\gg_+=\gg_+(\hat{y};L)$ is the unique solution of $\mathcal{G}_+(\gg,\hat{y};L)=0$. By the implicit function theorem, the dependence on $\hat{y},L$ of $\gg_+(\hat{y};L)$ is analytic. 

To finish the result related to the $+$ case, we consider now
\[
\mathcal{G}_{+}(m,\hat{y};L)=2\alpha^3 L^4 \rho \left [ \mathbf{b}_+(m) \hat{y}^3 + \mathbf{c}_+(m) \hat{y}^5\right].
\]
Notice that its sign does not depend on $L$. 
We have that 
\begin{equation*}
\mathbf{b}_+(m)=1.085189654741836,\qquad  \mathbf{c}_+(m)=-79.834380790596938. 
\end{equation*}
Therefore, for $0<\hat{y}< \hat{y}_{\min}$, $\mathcal{G}_+(m,\hat{y};L)>0$ and when $\hat{y}>\hat{y}_{\min}$, then $\mathcal{G}_+(m,\hat{y};L)<0$. We conclude that for any
$0\leq a \leq a_{\max}$ and $\hat{y}_{\min}<\hat{y}\leq \frac{1}{2}$ the function, $\gg_+(\hat{y};L)\in (\sl,\frac{3}{2} \sl)$ and 
when $0<\hat{y}<\hat{y}_{\min}$, we have that $\gg_+ (\hat{y};L)\in (0,\sl)$.

For the $-$ case, we proceed analogously. We study the function $\mathcal{G}_-(\sl,\hat{y};L)$. 
Using that $\mathbf{a} ( \sl)=0$, 
\[
\mathcal{G}_-(\sl,\hat{y};L) = 2 \alpha^3 L^4 \rho\left [ {\mathbf{b}}_-(m)  \hat{y}^3 + {\mathbf{c}}_-(m) \hat{y}^5 \right ]
\]
with 
\[
\mathbf{b}_- (\sl )=0.899076688965464, \qquad \mathbf{c}_-(\sl )=71.897315415767750,
\]
and hence $\mathcal{G}_-(\sl,\hat{y};L)> 0$ provided $\hat{y},L>0$. Moreover, we have that 
\begin{equation}\label{value:bneg04m}
\mathbf{a} \left(\frac{2}{5}\sl\right) = -0.616290933136957
, \qquad 
\mathbf{b}_-\left(\frac{2}{5}\sl\right)= 0.357833740529904
\end{equation}
and
\begin{equation}\label{value:cneg04m}
\mathbf{c}_-\left(\frac{2}{5}\sl\right) = 70.608271414680971.
\end{equation}
From these values we have that 
\[
\mathcal{G}_-\left(\frac{2}{5}\sl, \hat{y};L\right) \leq \mathcal{G}_- \left (\frac{2}{5}\sl , \frac{1}{2}; L_{\max} \right )= -0.076994095997652 <0.
\]
As a consequence, for any $0<\hat{y}\leq \frac{1}{2}$ and $0<a\leq a_{\max}$, there exists $\gg_-=\gg_-(\hat{y};L) \in (\frac{2}{5}\sl,\sl)$ such that $\mathcal{G}_-(\gg_-,\hat{y};L)=0$. Then, for fixed $\hat{y},L$, since by Lemma~\ref{lem:bcposneg}, $\partial_\gg \mathcal{G}_->0$, $\gg_-(\hat{y};L)$ is the unique zero of $\mathcal{G}_-(\gg,\hat{y};L)$ and by the Implicit Function Theorem $\gg_-$ is an analytic function. 
\end{proof}
 
\begin{remark}
We stress that if the semi-major axis $a\to \infty$ (that is $L\to \infty$), $C_+\left (\frac{1}{2}\right )$ in~\eqref{def:C12} is negative provided $4 \widehat{\mathbf{b}}_+ + \widehat{\mathbf{c}}_+<0$. In order to be a negative quantity, the maximum semi-major axis has to satisfy
\[
\alpha^3 L^4 < \frac{16\widehat{\mathbf{a}} }{ |4 \widehat{\mathbf{b}}_+ + \widehat{\mathbf{c}}_+| \rho }  \Leftrightarrow
a < \left ( \frac{16\hat{\mathbf{a}}\ a_M^3}{ |4 \hat{\mathbf{b}}_+ + \hat{\mathbf{c}}_+| \rho \mu^2 }\right )^{1/5} = 30160.25822948035 \, \mathrm{km},
\]
that corresponds to $a<1.018927642887850
$ in the units specified in Remark~\ref{rmk_unities}. 
This fact leads us to think that $a_{\max}$ is (almost) the optimal value for our arguments to be true.
\end{remark}

\begin{remark} We have that $\mathbf{b}_-(\frac{1}{2} \sl), \mathbf{c}_-(\frac{1}{2} \sl)>0$. Then for semi-major axis satisfying
\[
a\leq \left ( -\frac{16 \mathbf{a}(\frac{1}{2} \sl)  a_M^3   }{\mu^2 \rho (4\mathbf{b}_-(\frac{1}{2}m) + \mathbf{c}_-(\frac{1}{2}m))} \right )^{\frac{1}{5}} = 29700.20301662558\, \mathrm{km}
\]
(corresponding to $a\leq 1.003385237048162$ in the units in Remark~\ref{rmk_unities}) we have that
$\mathcal{G}_-(\frac{\sl}{2} ,\hat{y};L)<0$ and therefore $\gg_-(\hat{y};L) \in \left (\frac{\sl}{2} , \sl \right )$.
\end{remark}  

A straightforward corollary of this preliminary analysis is the following result, see~\eqref{equivFG_averaged}. 
\begin{corollary}\label{prop:existencefixedpointselliptic} 
Let $a$ be a semi-major axis satisfying $0< a \leq a_{\max}$. Then, given $L\hat{y}_{\min} < y \leq \frac{L}{2}$ with $L=\sqrt{\mu a} \in (0,L_{\max}]$, there exist two analytic functions $\Gamma_\pm(y;L)$ such that 
$$
\frac{\Gamma_{+}(y;L)}{ \sl y} \in \left (1, \frac{3}{2} \right  ), \;\frac{\Gamma_{-}(y;L)}{ \sl y} \in \left(\frac{2}{5}   ,1 \right),\qquad 
\mathcal{F}_{\pm}(y,\Gamma_{\pm}(y;L);L) =0.
$$
For $y,L$ fixed, $\Gamma_\pm:=\Gamma_{\pm}(y;L)$ are the unique solutions of $\mathcal{F}_\pm(y,\Gamma_\pm;L)=0$ belonging to $\left (0, \frac{y}{2}  \right)$.  
We also have that  $\partial_\Gamma \mathcal{F}_\pm(y,\Gamma;L)<0$ for $\Gamma\in \left (0, \frac{y}{2}  \right )$. 
\end{corollary}
\begin{proof}
For the existence and regularity of $\Gamma_{\pm}$, we only need to define $\Gamma_{\pm}(y;L)=y \gg_{\pm}\left (\frac{y}{L};L\right )$ and to apply Lemma~\ref{lem:bcposneg}. In addition, using~\eqref{equivFG_averaged}, we have that 
\[
\mathcal{F}_\pm(y,\Gamma;L)=\hat{\mathcal{F}}_\pm \left ( \frac{\Gamma}{y}, \frac{y}{L};L\right ) = -6 L y^3 \mathcal{G}_{\pm} \left ( \frac{\Gamma}{y}, \frac{y}{L};L\right )  
\]
and then 
\[
\partial_\Gamma \mathcal{F}_\pm(y,\Gamma;L) = -6Ly^2 \partial_\gg \mathcal{G}_{\pm} \left ( \frac{\Gamma}{y}, \frac{y}{L};L\right )<0,
\]
where we have used that, by Lemma~\ref{lem:bcposneg}, $\partial_\gg \mathcal{G}_\pm(\gg,\hat{y};L)>0$. 
\end{proof}

Corollary~\ref{prop:existencefixedpointselliptic} proves the existence of critical points. However, we can not elucidate neither the range of the parameters $\Gamma,L$ for which they exist nor how many critical points the system possesses for given values of the parameters. To do so, let us study the functions 
\[\hat{\Gamma}_\pm(\hat{y};L)  = \frac{\Gamma_\pm(L\hat{y}; L)}{L} = \hat{y} \gg_\pm(\hat{y};L),
\]
where $\gg_{\pm}(\hat{y};L)$ are defined implicitly by the equations $\mathcal{G}_{\pm}(\gg,\hat{y};L)=0$ (see Lemma~\ref{lem:existencefixedpointsG}). That is
 $\mathcal{G}_\pm(\gg_\pm(\hat{y};L),\hat{y};L)=0$. Therefore, writing $\gg_\pm=\gg_\pm(\hat{y};L)$ and taking the derivative with respect to $\hat{y}$ in the equation $\mathcal{G}_{\pm}(s,\hat{y};L)=0$, we have that
$$
\partial_{\hat{y}} \mathcal{G}_\pm(\gg_\pm,\hat{y};L) + \partial_{\gg} \mathcal{G}_\pm(\gg_\pm,\hat{y};L) \partial_{\hat{y}} \gg_\pm (\hat{y};L)=0 .
$$
We recall that by Lemma~\ref{lem:bcposneg}, $\partial_\gg \mathcal{G}_\pm >0$. This fact allows computing $\partial_{\hat{y}} \gg_{\pm} (\hat{y};L)$ in terms of $\gg_{\pm}$ and, consequently,  
\begin{align}\label{def:derGammahat}
\partial_{\hat y} \hat{\Gamma}_\pm (\hat{y};L) & = \gg_\pm(\hat{y};L) + \hat{y} \partial_{\hat{y}} \gg_\pm (\hat{y};L) \notag \\ &= \frac{1}{\partial_{\gg} \mathcal{G}_\pm(\gg_+,\hat{y};L)} \left [ \gg_\pm (\hat{y};L)\partial_{\gg} \mathcal{G}_\pm(\gg_\pm,\hat{y};L) - \hat{y} \partial_{\hat{y}} \mathcal{G}_\pm(\gg_\pm,\hat{y};L) \right].
\end{align}
We want to study the possible changes in the sign of $\partial_{\hat{y}} \hat{\Gamma}_\pm$. We notice that, following the same kind of computations, 
\begin{equation}\label{expr:partialLGammahat}
\partial_{L}\hat{\Gamma}_{\pm}(\hat{y};L) = -\frac{\partial_{L} \mathcal{G}_\pm(\gg_\pm,\hat{y};L) }{\partial_{\gg} \mathcal{G}_\pm(\gg_+,\hat{y};L)}  \hat{y}.
\end{equation} 

\begin{remark}\label{rmk:ymin} Notice that the value $\hat{y}_{\min}$ defined in Lemma~\ref{lem:existencefixedpointsG} satisfies that $\hat{y}_{\min} < \hat{y}_{\mathrm{col}}$, the \textit{collision} value $\hat{y}_{\mathrm{col}}$ in ~\eqref{remark_col}.
For that reason, from now on, we will restrict our analysis to values of $\hat{y} \in \left [\hat{y}_{\min} , \frac{1}{2}\right )$. 
\end{remark} 

As we will see below, it turns out that the analysis of the functions $\hat{\Gamma}_+$ and $\hat{\Gamma}_-$ are quite different, being the corresponding to $\hat{\Gamma}_-$ more involved. 

We start with the result for the fixed points of the form $(0,y)$, which corresponds with the $+$ case, in Section~\ref{subsec:fixedpoints_pos}, and we postpone the study of the fixed points of the form $(\pi,y)$ to Section~\ref{subsec:fixedpoints_neg}. 

\subsection{The fixed points of the form $(0,y)$}\label{subsec:fixedpoints_pos}
For any $L_{\min}\leq L \leq L_{\max}$ we define
\begin{align*}
\hat{\Gamma}_{\min}^+ &= \hat{y}_{\min} \sl = 0.013584391815073, \\ 
\hat{\Gamma}_{\max}^+(L) &= \hat{\Gamma}_+\left (\frac{1}{2}; L \right) = \frac{1}{2} \gg_+\left (\frac{1}{2};L\right ).
\end{align*} 
We emphasize that, by Lemma~\ref{lem:existencefixedpointsG},  $\gg_+\left (\frac{1}{2};L \right ) \in \left (\sl, \frac{3}{2}\sl\right )$ and therefore
\begin{equation}\label{intGammamaxpos}
\hat{\Gamma}_{\max}^+(L)    \in \left [\frac{1}{2}\sl,\frac{3}{4}\sl \right ]= [0.058257569495584,0.087386354243376].
\end{equation}
\begin{proposition}\label{prop:fixedpointscasepos}
The function 
$\hat{\Gamma}_{\max}^+(L)$ is strictly increasing on $[L_{\min},L_{\max}]$ and therefore
$$
\hat{\Gamma}_{\max}^+(L) \in [\hat{\Gamma}_{\max}^+(L_{\min}), \hat{\Gamma}_{\max}^+(L_{\max})]=[ 0.058270961487710, 0.086680627913961]
$$
where these values have been  computed numerically\footnote{ Notice that there is not a significant difference with the estimates in~\eqref{intGammamaxpos}.}. 
Moreover, 
$$
\hat{\Gamma}^+(\hat{y};L) \in [\hat{\Gamma}_{\min}^+, \hat{\Gamma}_{\max}^+(L)],\qquad \text{for all }\;\hat{y}_{\min} \leq \hat{y} \leq \frac{1}{2},
$$
and $\hat{\Gamma}^+(\cdot; L)$ is an injective increasing function. 

With respect to the equilibrium points, for any $L\in [L_{\min},L_{\max}]$:
 \begin{enumerate}
\item If $\hat{\Gamma} \notin [\hat{\Gamma}_{\min}^+, \hat{\Gamma}_{\max}^+(L)]$, system~\eqref{Eq:EMotMoy} has no fixed points of the form $(0,L\hat y)$ satisfying $\hat{y}_{\min} <\hat y\leq \frac{1}{2}$ and $\hat\Gamma\in \left (0,\frac{1}{2} \hat y \right )$. 
\item If $\hat{\Gamma}\in [\hat{\Gamma}_{\min}^+, \hat{\Gamma}_{\max}^+(L)]$, there exists a unique $\hat{y}_+=\hat{y}_+(\hat{\Gamma};L)$ such that $(0,L\hat{y}_+)$ is a fixed point of the system~\eqref{Eq:EMotMoy}. In addition, $\hat{y}_+(\cdot; L)$ is a strictly increasing function. 
\end{enumerate}
\end{proposition}
\begin{proof}
We start with the statement related to $\hat{\Gamma}_{\max}^+$. Using~\eqref{expr:partialLGammahat}, we have that
\[
\partial_L \hat{\Gamma}_{\max}^+(L)= \partial_L \hat{\Gamma}^+ \left (\frac{1}{2}; L\right )= -
\frac{\partial_{L} \mathcal{G}_+ \left (\gg_+,\frac{1}{2};L\right) }{2 \partial_{\gg} \mathcal{G}_+\left (\gg_+,\frac{1}{2};L\right )}.
\]
The sign of $\partial_{L} \mathcal{G}_+\left (\gg_+,\frac{1}{2};L\right)$ is the sign of $\mathbf{b}_+(\gg_+) + \frac{1}{4} \mathbf{c}_+(\gg_+)$ (see definition~\eqref{expr:mathcalG}). Then,  since by Lemma~\ref{lem:bcposneg}, $\mathbf{b}_+,\mathbf{c}_+$ are increasing functions and by Lemma~\ref{lem:existencefixedpointsG}, $\gg_+ \in \left (\sl,\frac{3}{2}\sl\right )$
\[
\mathbf{b}_+(\gg_+) + \frac{1}{4} \mathbf{c}_+(\gg_+)\leq \mathbf{b}_+\left (\frac{3}{2}\sl \right) + \frac{1}{4} \mathbf{c}_+ \left (\frac{3}{2}\sl\right )<0,
\]
where we have used~\eqref{value:bposm32} and~\eqref{value:cposm32} to estimate
$\mathbf{b}_+\left (\frac{3}{2}\sl \right) + \frac{1}{4} \mathbf{c}_+ \left (\frac{3}{2}\sl\right )$. Therefore, $\partial_L\mathcal{G}_+<0$ and, 
using that $\partial_s \mathcal{G}_+>0$ by Lemma~\ref{lem:bcposneg}, $\partial_L\Gam_{\max}^+>0$. 

Now we prove the rest of the properties. We write
\[
\partial_{\hat{y}} \mathcal{G}_+ ( \gg, \hat{y};L)= 2 \alpha^3 L^4 \rho \left [3\mathbf{b}_+(\gg) \hat{y}^2 + 5 \mathbf{c}_+(\gg) \hat{y}^4 \right ].
\]
By Lemma~\ref{lem:bcposneg}, $\mathbf{b}_+,\mathbf{c}_+$ are increasing functions and, by Lemma~\ref{lem:existencefixedpointsG}, $\gg_+ \in \left (\sl , \frac{3}{2} \sl\right )$.  Therefore, for $\gg \in \left (\sl , \frac{3}{2} \sl \right )$,
\[
\partial_{\hat{y}} \mathcal{G}_+ ( \gg, \hat{y};L) \leq  2 \alpha^3 L^4 \rho \left [3\mathbf{b}_+\left (\frac{3}{2} \sl  \right) \hat{y}^2 + 5 \mathbf{c}_+ \left (\frac{3}{2} \sl \right ) \hat{y}^4 \right ].
\]
Using the values of $\mathbf{b}_+\left (\frac{3}{2} \sl \right), \mathbf{c}_+\left (\frac{3}{2} \sl \right )$ in~\eqref{value:bposm32} and~\eqref{value:cposm32}, if 
\[
\hat{y} \geq  \sqrt{-\frac{3\mathbf{b}_+\left (\frac{3}{2} \sl\right )}{5 \mathbf{c}_+ \left (\frac{3}{2} \sl \right )}} = 0.116570155823739,
\]
we have then 
$\partial_{\hat{y}} \mathcal{G}_+(\gg,\hat{y};L) \leq 0$. In particular, the same happens if $\hat{y} \geq  \hat{y}_{\min}$ (see Lemma~\ref{lem:existencefixedpointsG} for the exact value of $\hat{y}_{\min}$). 
As a consequence, from~\eqref{def:derGammahat}, $\partial_{\hat{y}} \hat{\Gamma}_+(\hat{y};L)>0$ if $\hat{y}_{\min}  \leq  \hat{y}  \leq  \frac{1}{2}$ and $L_{\min}\leq L \leq L_{\max}$ so that $\hat{\Gamma}_+(\cdot; L)$ is an injective (strictly increasing) function for any fixed $L\in [L_{\min},L_{\max}]$. The range of values of $\hat{\Gamma}_+$ for a given $L$ is then 
\[
\hat{\Gamma}^+(\hat{y};L)\in \left [\hat{\Gamma}_+(\hat{y}_{\min};L), \hat{\Gamma}_{+}\left (\frac{1}{2}; L\right ) \right ].
\]
We notice that, since  $\mathcal{G}_+(\sl,\hat{y}_{\min};L)=0$ (see Lemma~\ref{lem:bcposneg}), 
\[
\hat{\Gamma}_+(\hat{y}_{\min};L) = \hat{y}_{\min} \gg_+(\hat{y}_{\min};L) = \hat{y}_{\min} \sl =   \hat{\Gamma}_{\min}^+.
\]
For a given $L$, let $\hat{y}_+(\hat{\Gamma};L)$ be such that 
\[
\hat{\Gamma}_+(\hat{y}_+(\hat{\Gamma};L);L)=\hat{\Gamma}.
\]
It is clear that $\hat{y}_+(\cdot; L)$ is an increasing function and it is defined  for $\hat{\Gamma} \in [\hat{\Gamma}_{\min}, \hat{\Gamma}_{\max}^+(L)]$.   
\end{proof}

\begin{remark}
If we want to control the range of $\hat{\Gamma}$ for $\hat{y}_{\mathrm{col}} \leq \hat{y} \leq \frac{1}{2}$, we notice that, since $\hat{\Gamma}_+(\cdot; L)$ is an increasing function, the ``new'' $\hat{\Gamma}_{\min}^+ = \hat{\Gamma}_{\min}^+(L) = \hat{\Gamma}_+(y_{\mathrm{col}};L)$. Then, 
since   $\hat{\Gamma}_+(\hat{y}_{\mathrm{col}}; L)= \hat{y}_{\mathrm{col}} \gg_+(\hat{y}_{\mathrm{col}};L)$,  
$$
\hat{\Gamma}_{\min}^+(L) \in \left [ \hat{y}_{\mathrm{col}}\sl , \hat{y}_{\mathrm{col}} \frac{3}{2} \sl \right ]=[0.035912783574701,0.053869175362051].
$$  
We emphasize that $ \hat{\Gamma}_+(\hat{y}_{\mathrm{col}};\cdot )$ is also an increasing function. Indeed, writing $\gg_+ = \gg_+(\hat{y}_{\mathrm{col}};L)$, the sign of $-\partial_L \hat{\Gamma}_+(\hat{y}_{\mathrm{col}};\cdot )$ is the sign of
$$
\mathbf{b}_{+}(\gg_+) + \hat{y}_{\mathrm{col}}^2 \mathbf{c}_+(\gg_+) \leq \mathbf{b}_{+}\left (\frac{3}{2} \sl\right) + \hat{y}_{\mathrm{col}}^2 \mathbf{c}_+ \left (\frac{3}{2}\sl \right ) = -5.701123730321832<0
$$
where, again, we have used~\eqref{value:bposm32} and~\eqref{value:cposm32} for $\mathbf{b}_+\left (\frac{3}{2}\sl\right )$ and $\mathbf{c}_+\left (\frac{3}{2}\sl \right )$. 
Therefore, we numerically obtain that 
$$
 \hat{\Gamma}_{\min}^+(L)    \in [\hat{\Gamma}_+(y_{\mathrm{col}};L_{\min}), \hat{\Gamma}_+(y_{\mathrm{col}};L_{\max}) ]= [0.035913449550478,0.037427016836718].
$$
We notice that in this case, the numerical computation induces a more accurate range of values of $\hat{\Gamma}_{\min}^+(L)$. 
\end{remark}

\subsection{The fixed points of the form 
$(\pi,y)$}\label{subsec:fixedpoints_neg}
Now we pay attention to the $-$ case. In this case, $\hat{\Gamma}_-$ is not an injective function and for that reason, for studying the behavior of $\hat{\Gamma}_-$, we must control the existence of critical points, namely the values of $\hat{y},L$ such that $\partial_{\hat{y}} \hat{\Gamma}_-(\hat{y};L)=0$. 
To this end, we introduce 
\begin{equation}\label{defmathcalH}
\mathcal{H}(\gg,\hat{y};L):=\gg \partial_{\gg} \mathcal{G}_-(\gg,\hat{y};L) - \hat{y} \partial_{\hat{y}}\mathcal{G}_-(\gg,\hat{y};L), \qquad \partial_{\hat{y}} \hat{\Gamma}_{-}(\hat{y};L)=
\frac{\mathcal{H} (\gg_-,\hat{y};L)}{\partial_{\gg} \mathcal{G}_{-} (\gg_{-},\hat{y};L)}
\end{equation}
with $\gg_-= \gg_-(\hat{y};L)$. By the uniqueness statement in Lemma~\ref{lem:bcposneg}, we have that 
$\mathcal{G}_-(\gg,\hat{y};L)=0$ if and only if $\gg= \gg_-(\hat{y};L)$. Moreover $\gg_-(\hat{y};L) \in \left (\frac{2}{5}\sl, \sl \right )$. Therefore, the equation
$$
\partial_{\hat{y}} \hat{\Gamma}_- (\hat{y};L)=0
$$
is equivalent to the existence of solutions of $$\mathcal{H}(\gg,\hat{y};L)=0, \qquad \mathcal{G}_-(\gg,\hat{y};L)=0
$$ 
under the restrictions $\gg\in \left (\frac{2}{5}\sl, \sl \right )$, $\hat{y}\in \left [\hat{y}_{\min}, \frac{1}{2}\right ]$ (see Remark~\ref{rmk:ymin}), $L\in [L_{\min},L_{\max}]$. 


From expression~\eqref{expr:mathcalG} of $\mathcal{G}_-$, we write the system in a more suitable way  
\begin{equation}\label{systemGammaneg}
\begin{aligned}
    \mathcal{H}(\gg,\hat{y};L)=\overline{\mathbf{a}}(\gg) + 2 \alpha^3 L^4 \rho \big [\overline{\mathbf{b}}(\gg) \hat{y}^3 + \overline{\mathbf{c}}(\gg) \hat{y}^5] =0\\ 
    \mathcal{G}_-(\gg,\hat{y};L)= \mathbf{a}(\gg) + 2 \alpha^3 L^4 \rho \big [\mathbf{b}_- (\gg)\hat{y}^3 + \mathbf{c}_-(\gg) \hat{y}^5  \big  ] =0
\end{aligned}
\end{equation}
with 
\begin{equation} \label{defbarabc}
\begin{aligned}
    \overline{\mathbf{a}}(\gg) & = \gg \partial_\gg \mathbf{a}(\gg) = 2 \gg (5\gg + 4) \\ 
    \overline{\mathbf{b}}(\gg) & = \gg \partial_\gg \mathbf{b}_-(\gg) - 3 \mathbf{b}_-(\gg)  
     = -10 U_2^{0,0} \gg(\gg + 2) - \frac{20 U_2^{1,0} 
   \gg^2(\gg^2 + 4\gg - 3)}{3(1-\gg) \sqrt{3 - 2 \gg -\gg^2}}
  \\ 
    \overline{\mathbf{c}}(\gg) & = \gg \partial_\gg \mathbf{c}_-(\gg) - 5 \mathbf{c}_-(\gg)  = 
    8 U_2^{0,0} (5 - 12\gg) - \frac{80 U_2^{1,0} (8\gg^3 + 5\gg^2 - 60\gg + 45)}{3(1- \gg) \sqrt{3 - 2 \gg -\gg^2 }}.
 \end{aligned}
 \end{equation}
Hence, in order to solve system~\eqref{systemGammaneg}, we write from the second equation  
$$
2\alpha^3 L^4 \rho = -\frac{\mathbf{a}(\gg)}{\mathbf{b}_-(\gg) \hat{y}^3 + \mathbf{c}_-(\gg) \hat{y}^{5}} 
$$
and substituting this value into the first equation, we obtain that 
\begin{equation*}
\mathcal{H}(\gg,\hat{y};L)= \frac{ \overline{\mathbf{a}}(\gg) \mathbf{b}_-(\gg) - \mathbf{a}(\gg) \overline{\mathbf{b}}(\gg) + \hat{y}^2 \big ( \overline{\mathbf{a}}(\gg ) \mathbf{c}_-(\gg) - \mathbf{a}(\gg) \overline{\mathbf{c}}(\gg) \big )}{ \mathbf{b}_-(\gg) + \mathbf{c}_-(\gg) \hat{y}^2} =0
\end{equation*}
so that 
$$
\hat{y}^2= \frac{\mathbf{a}(\gg) \overline{\mathbf{b}}(\gg)-\overline{\mathbf{a}}(\gg) \mathbf{b}_-(\gg)}{\overline{\mathbf{a}}(\gg ) \mathbf{c}_-(\gg) - \mathbf{a}(\gg) \overline{\mathbf{c}}(\gg)}.
$$
Therefore system~\eqref{systemGammaneg} has a unique solution for those values of $\gg\in \left (\frac{2}{5} \sl, \sl \right ) $ such that the following restrictions hold
\begin{equation}\label{exp:yabif}
\begin{aligned}
{\hat{y}_{\min}} \leq \hat{y}& =\hat{y}_*(\gg) := \left ( \frac{\mathbf{a}(\gg) \overline{\mathbf{b}}(\gg)-\overline{\mathbf{a}}(\gg) \mathbf{b}_-(\gg)}{\overline{\mathbf{a}}(\gg ) \mathbf{c}_-(\gg) - \mathbf{a}(\gg) \overline{\mathbf{c}}(\gg)}\right )^{\frac{1}{2}} \leq \frac{1}{2}, 
\\ a_{\min}\leq a&=a_*(\gg):= \left (-\frac{\mathbf{a}(\gg) a_M^3}{2 \mu^2 \rho (\mathbf{b}_-(\gg) \hat{y}_*^{3}(\gg)  + \mathbf{c}_-(\gg) \hat{y}_*^5 (\gg))}\right )^{\frac{1}{5}} \leq a_{\max}.
\end{aligned}
\end{equation} 


\begin{lemma}\label{lem:yaast}
For those values of $\gg$ such that $\hat{y}_{*}(\gg) \in \left (\hat{y}_{\min}, \frac{1}{2}\right ]$, $\hat{y}_*(\gg)$ is an increasing function and $a_*(\gg)$ is decreasing. Defining implicitly $\gg_{\min}$ and $\gg_{\max}$ such that 
\[
a_*(\gg_{\min})=a_{\max}, \qquad \hat{y}_*(\gg_{\max})=\frac{1}{2},
\]
we have that, for $\gg \in [\gg_{\min},\gg_{\max}]$,  $\hat{\Gamma}_-(\hat{y};L)$ has a unique critical point at $(\hat{y},L)= (\hat{y}_*(\gg), L_*(\gg))$ with $L_*(\gg)=\sqrt{\mu a_*(\gg)}$. The values of $\gg_{\min}$ and $\gg_{\max}$ and $L_*(\gg_{\max}), \hat{y}_*(\gg_{\min})$ can be computed numerically:
\[
{\gg_{\min} = 0.096577225237580,} \qquad \gg_{\max}= 0.096796816334740, 
\]
and then $a_*(\gg_{\max})=23893.56218133389 \, \mathrm{km}$, 
\[
  L_*(\gg_{\max})=97590.90325766560\, \mathrm{km^2/s}
 ,\qquad \hat{y}_*(\gg_{\min})= 0.397273020602216.
\] 
In addition, if $L\in [L_{\min},L_*(\gg_{\max})]$ there are no critical points of $\hat{\Gamma}_-(\cdot;L)$ belonging to $\left [\hat{y}_{\min},\frac{1}{2} \right )$. 
\end{lemma}
\begin{proof}
One can perform an analytic thorough study of $\hat{y}_*(\gg), a_*(\gg)$ as a function of $\gg$. However, we have preferred just to draw the explicit functions $\hat{y}_*(\gg), a_*(\gg)$ in order to convince the reader about this result. For those values of $\gg$ such that $\hat{y}_*(\gg) \in \mathbb{R}$ with $0<\hat{y}_*(\gg)\leq \frac{1}{2}$, we draw the corresponding semi-major axis $a$, see Figure~\ref{figyahatvss}.

%

\begin{figure}[h]
\centering
\begin{tabular}{cc} 
\includegraphics[width=6cm]{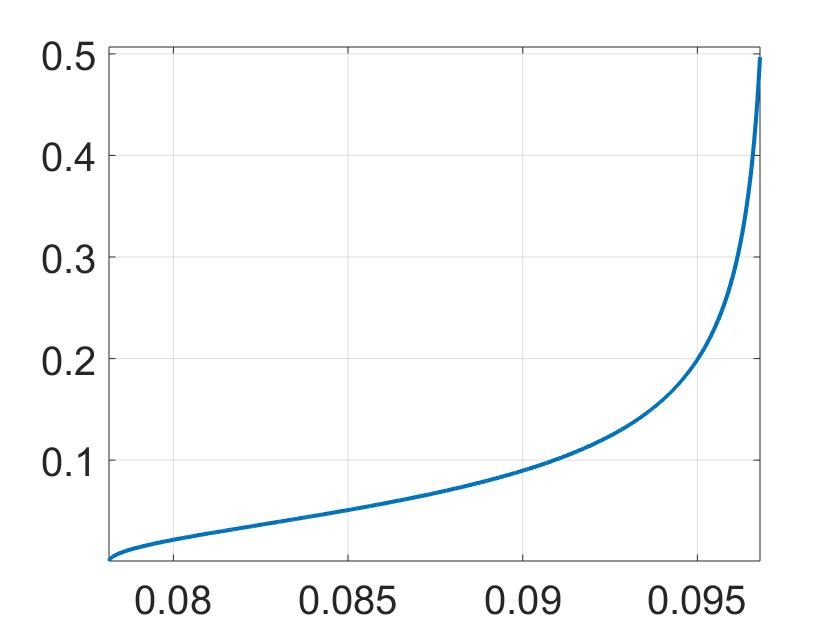} & 
\includegraphics[width=6cm]{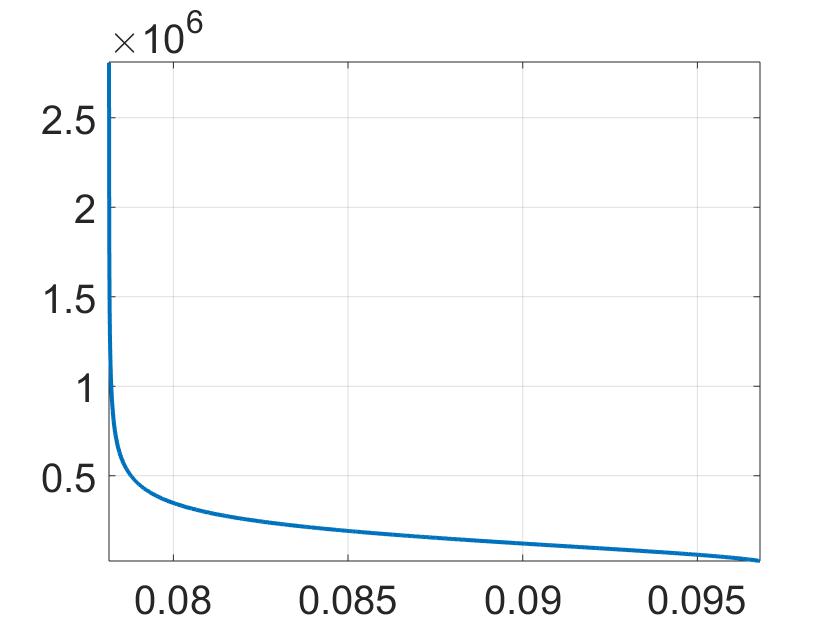}.
\end{tabular}
\caption{On the left, $\hat{y}_*(s)$ and on the right $a_*(s)$. Only the values of $s$ such that $\hat{y}_*(s) \in \left [0,\frac{1}{2}\right ]$ are considered. In this figure, the semi-major axis $a$ is measured in km.}
\label{figyahatvss}
\end{figure}

These figures illustrate that indeed $\hat{y}_*(\gg)$ is an increasing function of $\gg$ meanwhile $a_*(\gg)$ is decreasing. 

If we only consider the values of $\gg$ such that $\hat{y}_*(\gg) \in \left[\hat{y}_{\min}, \frac{1}{2}\right ]$ and $a_*(\gg) \in [a_{\min}, a_{\max}]$ (which are the ones we are interested in) we observe the behaviour in Figure~\ref{figyahatvssmin}.
\begin{figure}[h]
\centering
\begin{tabular}{cc} 
\includegraphics[width=6cm]{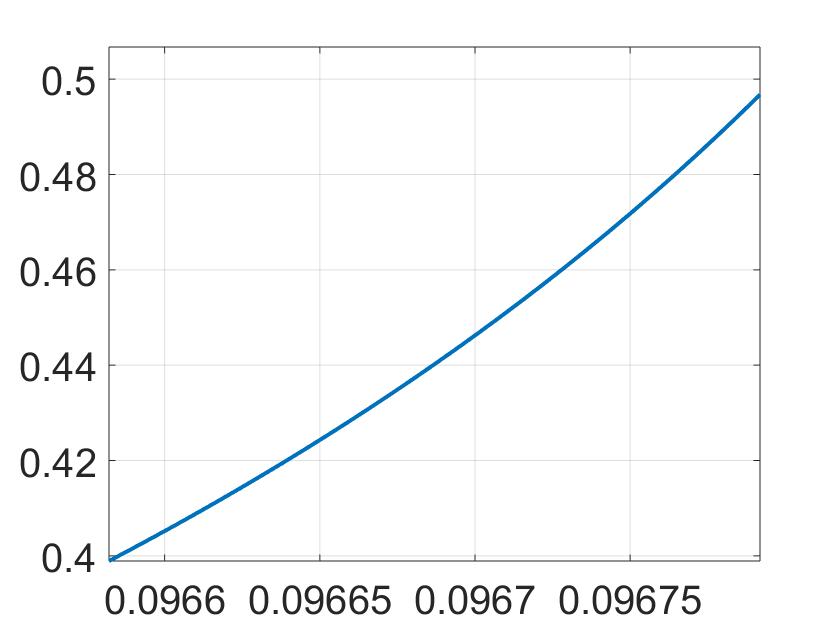} & 
\includegraphics[width=6cm]{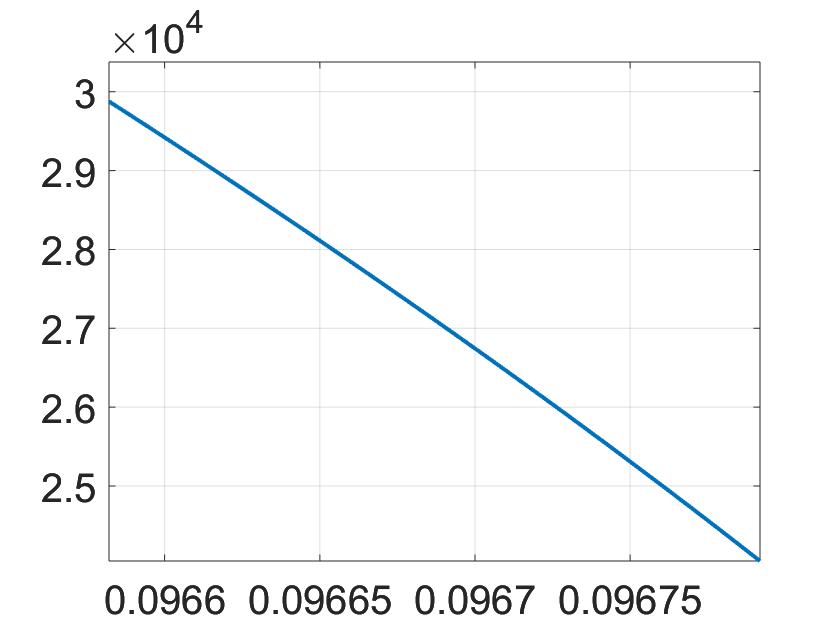}.
\end{tabular}
\caption{As in Figure~\ref{figyahatvss}, $\hat{y}_*(s)$ (left) and $a_*(s)$ (right). The values $s$ considered as the ones such that $\hat{y}_{\min}\leq \hat{y}_*(s) \leq \frac{1}{2}$ and $a_{\min}\leq a \leq a_{\max}$ in km.}
\label{figyahatvssmin}
\end{figure}

Summarizing, for the values of $\gg_{\min}, \gg_{\max} $ in the lemma we have that for any $\gg \in [\gg_{\min},\gg_{\max}]$,  
$$
\partial_{\hat{y}}\hat{\Gamma}_-\big(\hat{y}_*(\gg); L_*(\gg)\big )=0
$$
with $L_*(\gg)=\sqrt{\mu a_*(\gg)}$ and $\hat{y}_*(\gg), a_*(\gg)$ defined in~\eqref{exp:yabif}. In addition,
these are the only possible critical points of $\hat{\Gamma}_-$. 
\end{proof}

We introduce now the boundary values of $\hat{\Gamma}$. We first define, for $L\in [L_{*}(\gg_{\max}), L_{\max}]$,
$$
\hat{\Gamma}_{*}^-(L)= \hat{y}_*(\gg )\gg , \qquad \text{with} \quad  s=s(L) \quad \text{ such that} \quad  L=L_*(\gg).
$$
Then we introduce, for $L\in [L_{\min},L_{\max}]$
\begin{equation}\label{defGminmax-}
\hat{\Gamma}_{\min}^-(L)=\hat{\Gamma}_{-}(\hat{y}_{\min};L), \qquad \hat{\Gamma}_{\max}^-(L)= \hat{\Gamma}_-\left (\frac{1}{2} ; L\right )
\end{equation}
that, by Corollary~\ref{prop:existencefixedpointselliptic}, satisfy
\begin{equation}
\label{defGminmax-intervals}
\begin{aligned}
\hat{\Gamma}_{\min}^-(L)& \in \left [\hat{y}_{\min} \frac{2}{5} \sl, \hat{y}_{\min} \sl  \right ] =[0.005433756726029,0.013584391815074 ]\\ 
\hat{\Gamma}_{\max}^-(L)& \in \left [ \frac{1}{5} \sl, \frac{1}{2} \sl\right ] =[0.023303027798234,0.058257569495584]  .
\end{aligned}
\end{equation}

By Lemma~\ref{lem:yaast}, $\hat{y}_*(\gg)$ is an increasing function whereas $L_*(\gg)$ is decreasing. Let $L_*^{-1}$ be its inverse. We have that $\hat{\Gamma}_{*}^-(L)=\hat{y}_*(L^{-1}_*(L)) L^{-1}_*(L)$ is a decreasing function. Therefore 
$$
\hat{\Gamma}_{*}^-(L) \in [ \hat{\Gamma}_{*}^-(L_{\max}), \hat{\Gamma}_*^- ( L_*(\gg_{\max}))]
$$
and,  by definition of $\gg_{\min},\gg_{\max}$ in Lemma~\ref{lem:yaast},
$$
\hat{\Gamma}_{*}^-(L) \in \left [\hat{y}_*(\gg_{\min}) \gg_{\min},  \frac{1}{2} \gg_{\max} \right ] = [0.038367525991514, 0.048398408167370 ].
$$

Next lemma studies the monotonicity properties of the function $\hat \Gamma_-$.
\begin{lemma}\label{lem:Gammanegincreasing}
For $\hat{y}_{\min}\leq \hat{y} \leq \frac{1}{2}$, the function $\hat{\Gamma}_-(\hat{y};\cdot)$ is  strictly decreasing (with respect to $L$) on $[L_{\min},L_{\max}]$. Therefore, 
$\hat \Gamma_{\min,\max}^-(L)$ (see definition~\eqref{defGminmax-}) are also decreasing and moreover 
$$
\hat{\Gamma}_{\min}^-(L) \in [ \hat{\Gamma}_{\min}^-(L_{\max}),  \hat{\Gamma}_{\min}^-(L_{\min})]=[0.013575332545548,0.013584387878826]
$$ 
and 
$$
\hat{\Gamma}_{\max}^-(L) \in [ \hat{\Gamma}_{\max}^-(L_{\max}),  \hat{\Gamma}_{\max}^-(L_{\min})]=[0.027639200647529, 0.058244177051364]
$$ 
where the values have been computed numerically. 

The function $\hat{\Gamma}_-(\cdot; L)$ satisfies, for $\hat{y}_{\min}\leq \hat{y}\leq \frac{1}{2}$, that
\begin{enumerate}
    \item For $L\in [L_{\min}, L_{*}(\gg_{\max})]$, the function $\hat{\Gamma}_-(\cdot;L)$ is strictly increasing and  moreover
    $\hat{\Gamma}_-(\hat{y};L) \in [\hat{\Gamma}_{\min}^-(L), \hat{\Gamma}_{\max}^-(L)]$. 
\item For $L\in (L_*(\gg_{\max}),L_{\max}]$, $\hat{\Gamma}_-(\cdot, L)$ has a maximum at $\hat{y}_{\max}(L) = \hat{y}_*(L_*^{-1}(L))$. In this case, 
$$
\hat{\Gamma}_-(\hat{y};L) \in [\hat{\Gamma}_{\min}^-(L), \hat{\Gamma}_*^-(L)].
$$
\end{enumerate}
\end{lemma}
\begin{remark}
The statement of Lemma~\ref{lem:Gammanegincreasing} can be graphically represented as in Figure~\ref{fig:abstractGammaneg}.
\begin{figure}[h]
\centering
\includegraphics[width=7.2cm]{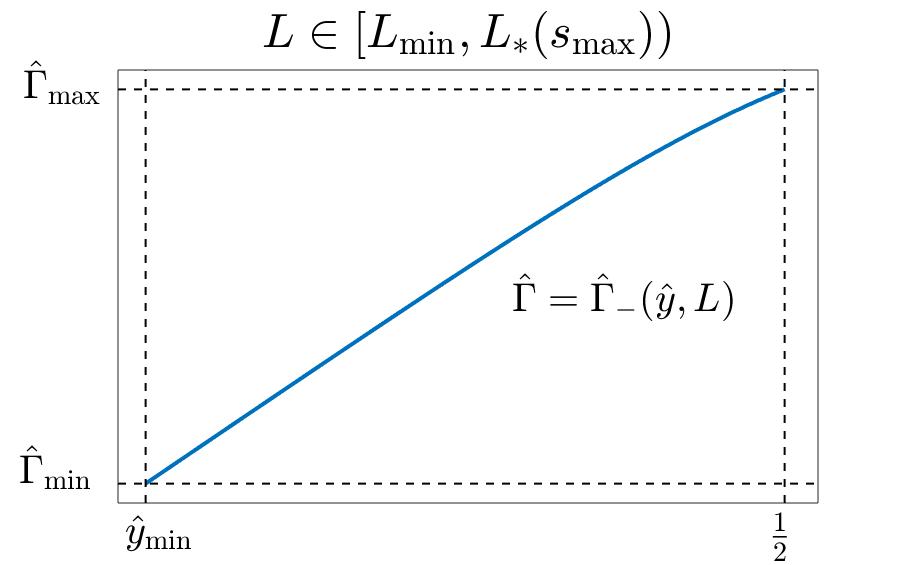}
\includegraphics[width=7.2cm]{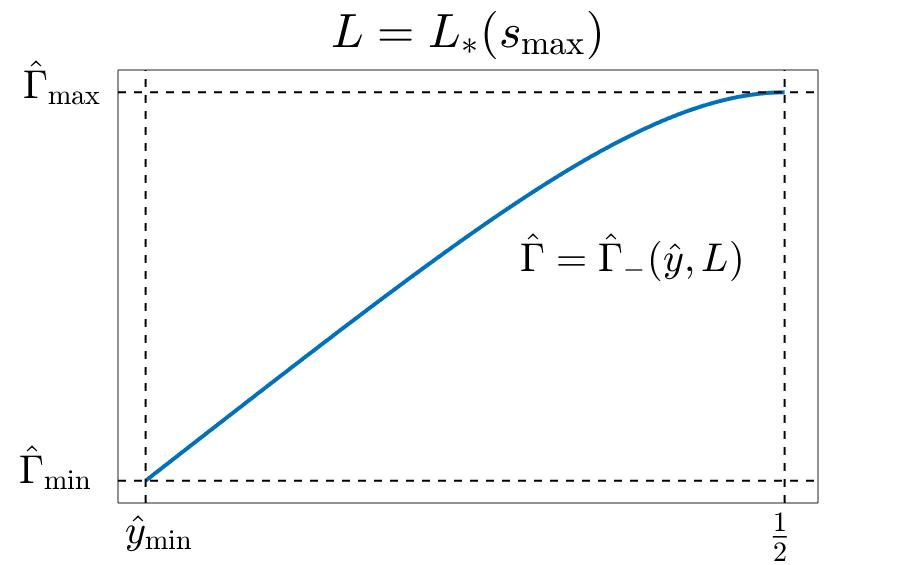}

\includegraphics[width=7.2cm]{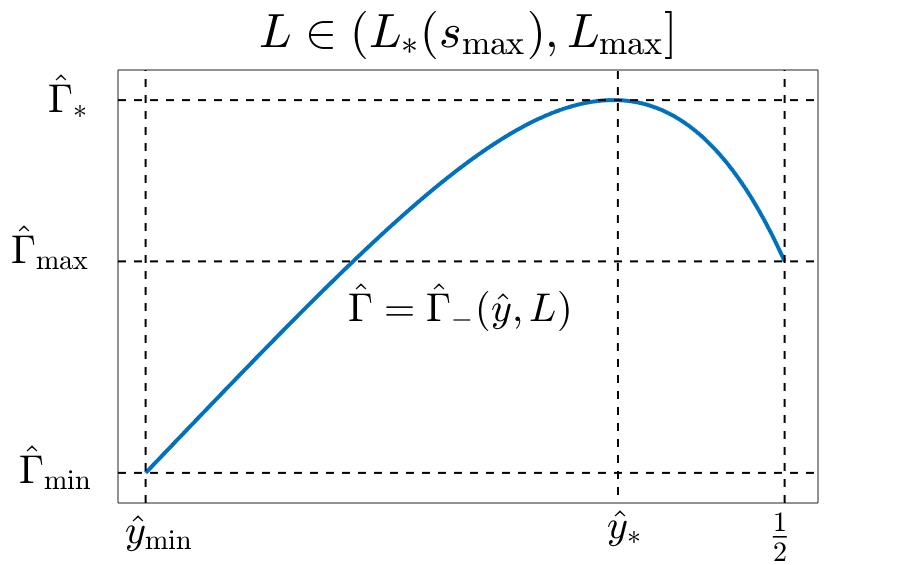}
 
    \caption{For a given value of $L$, the figures represent qualitatively the function $\hat{\Gamma}_-(\hat{y},L)$ (as a function of $\hat{y}$).
    There are depicted the three different behaviours depending on the value of $L$. See also Figure~\ref{figgammaneg} for numerical computations.} 
    \label{fig:abstractGammaneg}
\end{figure}
\end{remark}
\begin{proof} 
From 
formula~\eqref{expr:partialLGammahat} of $\partial_L \hat{\Gamma}_-$ and using that, by Lemma~\ref{lem:bcposneg}, $\partial_s \mathcal{G}>0$, we have that the sign of $\partial_L\hat{\Gamma}_-$, is the same as for 
$ 
-\partial_L\mathcal{G}_-(s_-,\hat{y},L) 
$ or, in other words, the same as 
$-\mathbf{b}_-(\gg_-) - \hat{y}^2 \mathbf{c}_-(\gg_-)$ with $\gg_-= \gg_-(\hat{y},L)$. Again, from Lemma~\ref{lem:bcposneg}, we have that $\mathbf{b}_-,\mathbf{c}_-$ are increasing functions and using their values~\eqref{value:bneg04m} and~\eqref{value:cneg04m} at $s=\frac{2}{5}\sl$, we conclude that $\mathbf{b}_-(\gg_-) + \hat{y}^2 \mathbf{c}_-(\gg_-)>0$ and so  $\hat{\Gamma}_{-}(\hat{y};\cdot)$ is decreasing.  

By Lemma~\ref{lem:yaast}, when $L\in [L_{\min},L_*(\gg_{\max})]$ there are no critical points of $\hat{\Gamma}_-$ in the  interval $\left [\hat{y}_{\min}, \frac{1}{2}\right )$. That implies that $\hat{\Gamma}_-(\cdot; L)$ is either increasing or decreasing. Notice that, by definition~\eqref{defGminmax-} of $\hat{\Gamma}_{\min,\max}^{-}$ and bounds~\eqref{defGminmax-intervals},
$$
\hat{\Gamma}_-(\hat{y}_{\min};L)=\hat{\Gamma}_{\min}(L)\leq \hat{\Gamma}_{\max}(L)
=
\hat{\Gamma}_-\left (\frac{1}{2}; L \right) 
$$
and therefore $\hat{\Gamma}_-(\cdot; L)$ is a strictly increasing function in $[L_{\min}, L_*(s_{\max}))$. 

When $L=L_*(\gg_{\max})$ we have that the corresponding critical point is $\hat{y}=y_*(\gg_{\max}) = \frac{1}{2}$. Therefore, the same argument as before allows to conclude that also in this case $\hat{\Gamma}_-(\cdot; L)$ is strictly increasing for $\hat{y}_{\min} \leq \hat{y} < \frac{1}{2}$.

Take now $L\in (L_*(\gg_{\max}), L_{\max}]$ and let $\gg =L_*^{-1}(L)\in [\gg_{\min}, \gg_{\max}]$ be such that $L=L_*(\gg)$. By Lemma~\ref{lem:yaast}, we already know that $\hat{y}_*(s)$ is a critical point of $\hat{\Gamma}_-(\cdot;L)$ and that $\hat{\Gamma}_-(\hat{y}_*(s); L_*(s))= \hat{\Gamma}_*^-(L)$. So we only need to check that 
$$
\partial^2_{\hat{y}} \hat{\Gamma}_-(\hat{y}_*(L^{-1}_*(L));L) = \partial^{2}_{\hat{y}} \hat{\Gamma}_-(\hat{y}_*(\gg);L_*(\gg))<0, \qquad s=L^{-1}_*(L).
$$
We first 
recall expression~\eqref{defmathcalH} of $\hat{\Gamma}_-$ in~\eqref{defmathcalH}:
$$
\partial_{\hat{y}} \hat{\Gamma}_-(\hat{y};L)= \frac{\mathcal{H}(s_-(\hat{y};L),\hat{y};L)}{\partial_s \mathcal{G}_-(s_-(\hat{y};L),\hat{y};L)}.  
$$
Then, using that, by construction, $\mathcal{H}(\gg, \hat{y}_*(\gg);L_{*}(\gg))=0$ and $s_-(\hat{y}_*(s);L_*(s))=s$ 
\[
    \partial_{\hat{y}}^2 \hat{\Gamma}_-(\hat{y}_*(\gg); L_*(\gg)) = 
    \frac{\partial_{\hat{y}} \mathcal{H}(s,\hat{y}_*(s);L_*(s)),   }{\partial_\gg \mathcal{G}_-(\gg, y_*(\gg);L_*(\gg))},
\]
and, since by Lemma~\ref{lem:bcposneg}, $\partial_s \mathcal{G}>0$, we need to compute the sign of $\partial_{\hat{y}} \mathcal{H} (s, \hat{y}_*(s); L_*(s))$. From~\eqref{systemGammaneg} we have that 
$$
\partial_{\hat{y}} \mathcal{H}(s,\hat{y};L)= 2\alpha^3L^4 \rho \big [3\overline{\mathbf{b}}(s) \hat{y}^2 + 5\overline{\mathbf{c}}(s) \hat{y}^4].
$$
Then, using that $\mathcal{H}(\gg, \hat{y}_*(\gg);L_{*}(\gg))=0$, we deduce that
$$
\partial_{\hat{y}} \mathcal{H} (s, \hat{y}_*(s); L_*(s)) = \frac{1}{\hat{y}_*(s)} \left ( -3 \overline{\mathbf{a}}(s) + 4 \alpha^3 L^4 \rho \overline{\mathbf{c}}(s) (\hat{y}_*(s))^5\right ).
$$
On the other hand, from expression~\eqref{defbarabc} for $\overline{\mathbf{c}}$,
$$
\overline{\mathbf{c}}(s) \leq 8U_2^{0,0} (5-12s)
$$
where we have used that $45-60s+5s^2 +8s^3 \geq 45-60s \geq 45-60\sl>0$. Therefore, using again that $s\in (\frac{2}{5}\sl,\sl)$, $0<\hat{y}_*(s) \leq \frac{1}{2}$ and $a_{\min}\leq a \leq a_{\max}$, we obtain
\begin{align*}
\hat{y}_*(s) \partial_{\hat{y}} \mathcal{H}(s,\hat{y};L) &\leq -6s(5s+4) + 32\alpha^3 L^4 \rho (\hat{y}_*(s) )^5 U_2^{0,0} (5-12s) \\
& \leq -6 \cdot \frac{2}{5}\sl (2\sl +4) +  \frac{a_{\max}^3}{a_\rM^3} L_{\max} ^4 \rho  U_2^{0,0} \left(5-\frac{24}{5} \sl\right) \\ & =-0.778057059724233<0.
\end{align*}
As a conclusion $\partial_{\hat{y}} \mathcal{H}(\gg_-,\hat{y};L)<0$ if $(\hat{y},L)=(\hat{y}_*(\gg),L_*(\gg))$ and that implies that $\partial_{\hat{y}}^2 \hat{\Gamma}_-(\hat{y}_*(\gg); L_*(\gg)) <0$ so that $\hat{y}_*(\gg)$ is a maximum of $\hat{\Gamma}(\cdot; L_*(\gg))$. 
\end{proof}

\begin{remark}
The values of $\hat{\Gamma}_-(\hat{y};L)$ can be numerically computed for any fixed $\hat{y},L$ as $\hat{\Gamma}_-(\hat{y};L)=s_-(\hat{y};L) \hat{y}$ with $s_-$ the zero of the function
$\mathcal{G}_-(s,\hat{y};L)$. In Figure~\ref{figgammaneg} we present some representative values of $a$ (recall that $L=\sqrt{\mu a})$ where we can find the different behaviour described in Lemma~\ref{lem:Gammanegincreasing}.
\begin{figure}[h]
\centering
    \begin{tabular}{cc} 
\includegraphics[width=6cm]{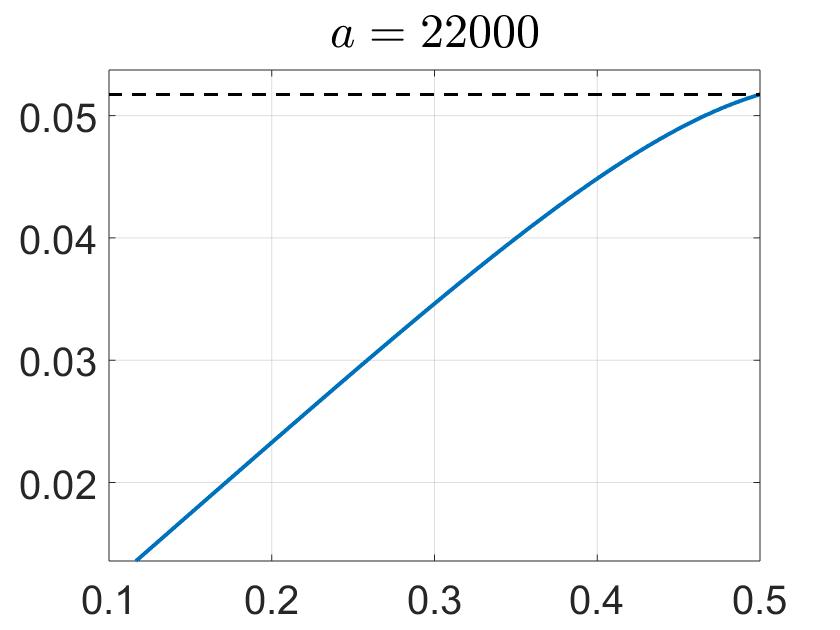}
 & 
\includegraphics[width=6cm]{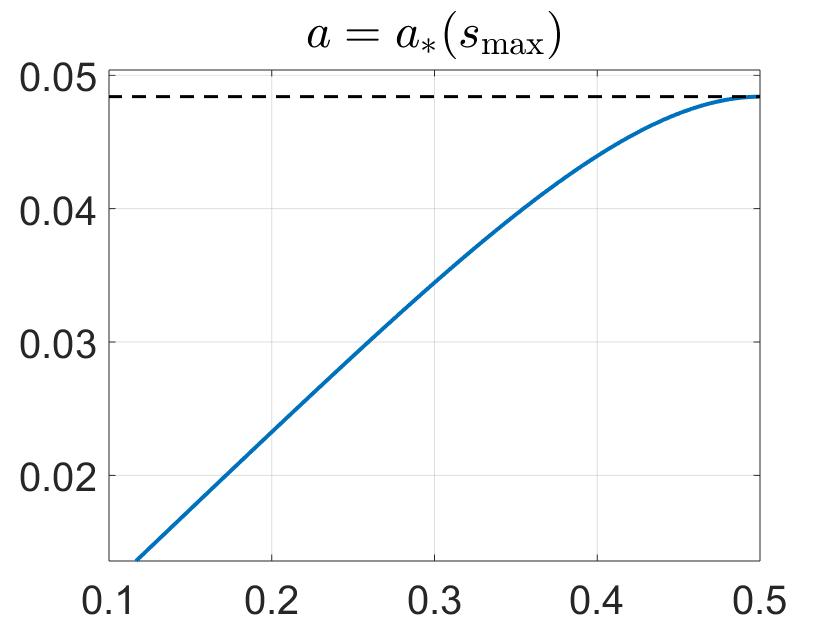}
\\
\includegraphics[width=6cm]{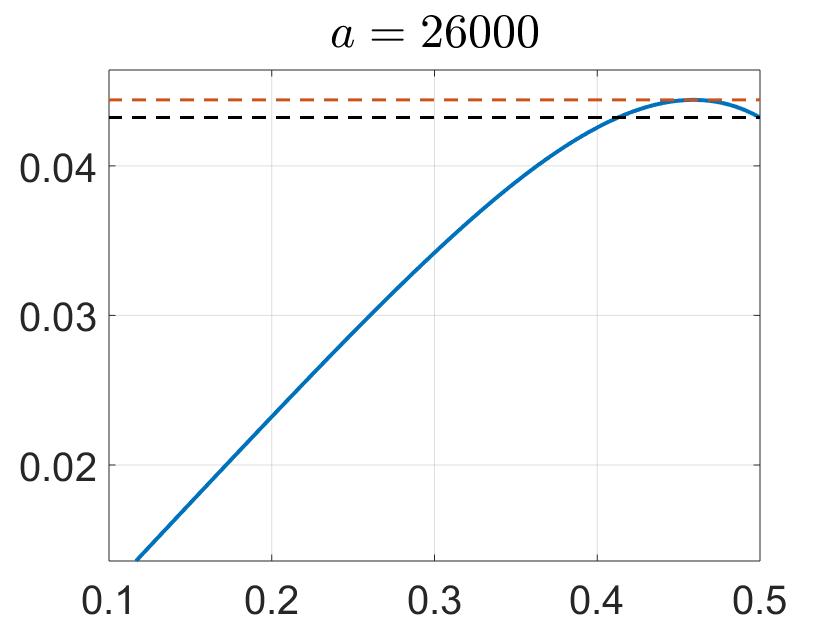}
&
\includegraphics[width=6cm]{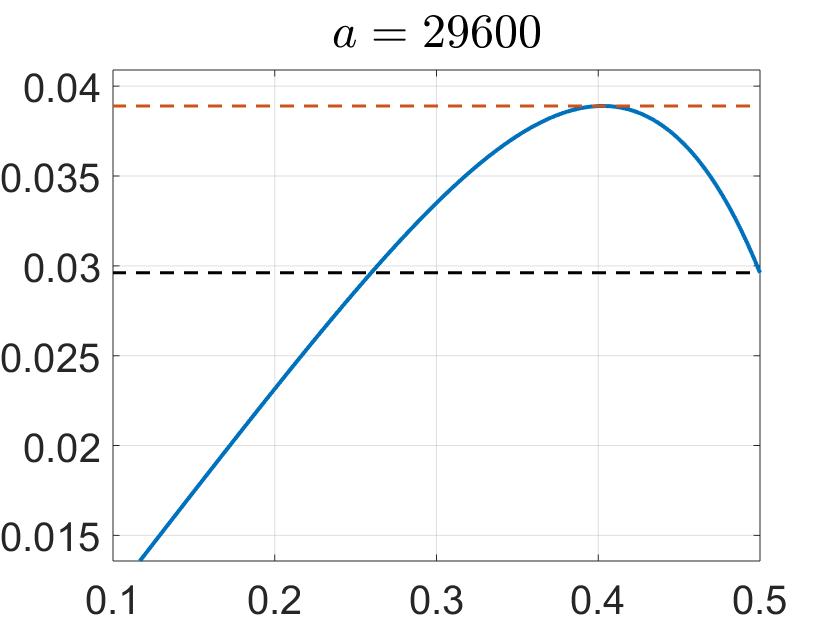}
\end{tabular} 
    \caption{The function $\hat{\Gamma}_-(\hat{y},L)$ for $a=22000, a_*(s_{\max}), 26000, 29600,\, \mathrm{km}$. Notice that for $a=22000, a_*(s_{\max})$ the function is strictly increasing but when $a=a_*(s_{\max})$ $\hat{\Gamma}_-(\hat{y}; L_*(s_{\max}))$ has a critical point at $\hat{y}=\frac{1}{2}$. The black dotted line is for $\hat{\Gamma}_{\max}^-(L)$ meanwhile the orange one is for $\hat{\Gamma}_*^-(L)$.} 
    \label{figgammaneg}
\end{figure}
\end{remark}
\begin{remark}\label{rmk:positionGammaast}  
We first notice that 
$
\hat{\Gamma}_*^-(L_*(s_{\max}))= \hat{\Gamma}_{\max}^-(L_*(s_{\max}) ).
$
Then, using that by Lemma~\ref{lem:Gammanegincreasing}, $\hat{\Gamma}_{\max}^-$ is decreasing (in its variable $L$), we also have that for $L\in (L_*(s_{\max}),L_{\max}]$, 
$$
\hat{\Gamma}^-_{\max}(L)< \hat{\Gamma}^-_{*}(L).
$$  
Moreover, since $s_{\max}= s_- \left ( \frac{1}{2}, L_*(s_{\max})\right )$, we have that $s_{\max} \leq \sl$. Therefore, using that by~\eqref{intGammamaxpos}, $\hat\Gamma^+_{\max}(L) \geq \frac{1}{2} \sl$, we have that
$$
\hat{\Gamma}^-_{*}(L) \leq \frac{1}{2} s_{\max} \leq \frac{1}{2} \sl \leq \hat{\Gamma}^+_{\max}(L).
$$
From~\eqref{defGminmax-intervals}
it is also clear that the constant $\hat \Gamma_{\min}^+=\hat y_{\min} \sl $ satisfies that 
$$
\hat{\Gamma}_{\min}^-(L) \leq \hat\Gamma_{\min}^{+} \leq \hat \Gamma_{\max}^-(L).
$$
\end{remark}

From this analysis, it is straightforward to deduce the following result about the existence of equilibrium points of the form $(\pi,y)$ (that is, the existence results of  Theorem~\ref{thm:AveragedHam_generalcase}). 

\begin{proposition}
Let $L\in [L_{\min},L_{\max}]$. If $\hat{\Gamma} \in [\hat{\Gamma}_{\min}^-(L), \hat{\Gamma}_{\max}^-(L)]$,  there exists a unique $\hat{y}_-(\hat{\Gamma};L)$ such that system~\eqref{Eq:EMotMoy} has a fixed point of the form $(\pi,y) = (\pi, L\hat{y}_-(\hat{\Gamma};L))$. The function $\hat{y}_{-}(\cdot; L)$ is strictly increasing.

In addition, for $L\in [L_{\min}, L_{*}(\gg_{\max})]$:
\begin{enumerate}
    \item If $\hat{\Gamma} \notin [\hat{\Gamma}_{\min}^-(L), \hat{\Gamma}_{\max}^-(L)]$,   system~\eqref{Eq:EMotMoy} has no fixed points of the form $(\pi ,L\hat y)$ with $\hat{y}_{\min} \leq  \hat y \leq \frac{1}{2}$ and $\hat\Gamma\in \left (0,\frac{1}{2} \hat y \right )$. 
\end{enumerate}
and when $L\in (L_*(\gg_{\max}), L_{\max}]$, we have that 
\begin{enumerate}
    \item If $\hat{\Gamma} \in [\hat{\Gamma}_{\max}^-(L), \hat{\Gamma}_*(L))$, there exist only two functions $\hat{y}_-^{1,2}(\hat{\Gamma};L)$ satisfying 
    $$
    \hat{y}_-^1 (\hat{\Gamma};L) \leq \hat{y}_{\max}(L) \leq \hat{y}_-^2 (\hat{\Gamma};L)
    $$
    with $\hat{y}_{\max}(L)$ defined in Lemma~\ref{lem:Gammanegincreasing}, 
    such that system~\eqref{Eq:EMotMoy} has two fixed point of the form $(\pi,y) = (\pi, L\hat{y}_-^{1,2}(\hat{\Gamma};L))$. In addition, 
    $\hat{y}^1_-(\cdot; L)$ is strictly increasing and $\hat{y}^2_-(\cdot;L)$ is strictly decreasing.  
    \item If $\hat \Gam = \hat \Gam_*(L)$, there exists only one fixed point $(\pi, \hat{y}_{\max}(L))$  with $\hat{y}_{\max} (L)$ a decreasing function.
    \item If $\hat{\Gamma} \notin [\hat{\Gamma}_{\min}^-(L), \hat{\Gamma}_*(L)]$,  system~\eqref{Eq:EMotMoy} has no fixed points of the form $(\pi ,L \hat y)$ with $\hat{y}_{\min} \leq \hat y \leq \frac{1}{2}$ and $\hat\Gamma\in \left (0,\frac{1}{2} \hat y \right )$.   
\end{enumerate}
\end{proposition}

 \subsection{Linearization around the critical points}
To complete the proof of Theorem~\ref{thm:AveragedHam_generalcase}, we are going to study the character of the fixed points of the form $(0,y), (\pi,y)$. Using the decomposition of $ \cB_{0,\alpha}$ in~\eqref{decompositioncalB0alpha}, we write system~\eqref{Eq:EMotMoy} as
\begin{align*}
\dot{x} &= -g(y;L)\big [X_1(y,\Gamma;L) +X_2(y,\Gamma;L) \cos x] \\
\dot{y}&= h(y,\Gamma;L) \sin x
\end{align*}
with 
\[
g(y;L)= \frac{1}{256} \frac{\rho_0}{L^4 y^6}, \qquad h(y,\Gamma;L)= \frac{1}{128} \frac{\rho \alpha^3}{L^4 y^5} \sqrt{P_1(y,\Gamma)} \mathcal{A}_1(y,\Gamma;L)
\]
and   
\[
X_1(y,\Gamma;L)=B_0(y,\Gamma;L) + \alpha^3 B_1(y,\Gamma;L), \qquad X_2(y,\Gamma;L) = \alpha^3 \frac{\cB_1(y,\Gamma;L)}{\sqrt{P_1(y,\Gamma)}}.
\]
By Corollary~\ref{prop:existencefixedpointselliptic},  we can characterize the fixed points as the sets
\[
\{(0,y,\Gamma_+(y;L),L)\}, \qquad \{(\pi,y,\Gamma_-(y;L);L)\}
\]
with $L\in [L_{\min},L_{\max}]$ and $\hat{y}_{\min} L \leq y <\frac{L}{2}$. They satisfy 
$
\mathcal{F}_\pm(y, \Gamma_{\pm}(y;L);L)=0,
$
where $\mathcal{F}_{\pm}$ are defined in~\eqref{exprmathcalFfixedpoints}, the sign $+$ corresponds to $(0,y)$ and $-$ otherwise. Equivalently we have that 
$$
\mathcal{F}_\pm(y, \Gamma_{\pm}(y;L);L)=X_1(y,\Gamma_{\pm}(y;L);L) \pm X_2(y,\Gamma_\pm(y;L);L)=0.
$$
The variational equation  at the fixed points $(0,y,\Gamma_+(y;L),L)$ or $(\pi,y,\Gamma_-(y;L),L)$, $\dot{z}=M^L(y) z$, is given by
\[
 M^{L}(y):=
\left (\begin{array}{cc} 0 & -g(y; L) \big [\partial_{y}X_1(y,\Gamma_\pm;L) \pm \partial_y X_2(y,\Gamma_\pm;L) \big ]  
 \\ \pm h(y,\Gamma_\pm;L) & 0 \end{array}\right ) ,
\]
where $\Gamma_\pm = \Gamma_\pm (y;L)$. 
The eigenvalues of $M^L$ then satisfy
\[
\lambda^2  = g(y; L)h(y,\Gamma_\pm;L)\big [\partial_y X_2(y,\Gamma_\pm;L) \pm \partial_y X_1(y,\Gamma_\pm;L)].
\] 
Since clearly $g(y;L) h(y,\Gamma;L)>0$, we need to study the sign of  
\begin{equation*}
\mathcal{E}_{\pm}(y;L) =\partial_y X_2(y,\Gamma_{\pm}(y;L);L) \pm \partial_y X_1(y,\Gamma_{\pm}(y;L);L)  
= \pm \partial_ y \mathcal{F}_\pm (y, \Gamma_\pm; L)
\end{equation*}
for $L \in [L_{\min},L_{\max}]$ and 
$\hat{y}_{\min} L \leq y \leq \frac{L}{2}$.  

On the other hand,  from $\mathcal{F}_{\pm}(y,\Gamma_{\pm}(y;L);L)=0$, we have that 
\[
\partial_y \mathcal{F}_\pm (y,\Gamma_\pm(y;L);L) + \partial_{\Gamma} \mathcal{F}_\pm (y,\Gamma_\pm(y;L);L) \partial_{y} \Gamma_\pm(y;L)=0 
\]
and then 
\begin{equation}\label{finexprmathcalE}
\mathcal{E}_{\pm } (y;L)= \mp \partial_\Gamma \mathcal{F}_{\pm}(y,\Gamma_\pm(y;L);L) \partial_y \Gamma_\pm(y;L).
\end{equation}
 
\begin{proposition}\label{prop:fixedpointscharacter}
We have that 
\begin{enumerate}
    \item For $L\in [L_{\min},L_{\max}]$ and 
    $\Gamma\in [L \hat{\Gamma}_{\min}^+, L \hat{\Gamma}_{\max}^+(L)]$, the unique fixed point of the form $(0,y)$
satisfies that 
    $y=y_+(\Gamma;L):= L \hat{y}_+ \left (\frac{\Gamma}{L};L \right )$ and it is a saddle.
    \item For $L\in [L_{\min}, L_{\max}]$ and $\Gamma\in [L \hat{\Gamma}_{\min}^-(L), L\hat{\Gamma}_{\max}^-(L))$, the unique fixed point of the form $(\pi,y)$
satisfies that 
    $y=y_-(\Gamma;L):= L \hat{y}_- \left (\frac{\Gamma}{L};L \right )$ and it is a center.
    \item For $L\in (L_*(\gg_{\max}),L_{\max}]$ and $\Gamma \in [L \hat{\Gamma}_{\max}^-(L), L \hat{\Gamma}_{*}(L))$ there are two fixed points of the form $(\pi,y)$ with $y=y_{-}^{1,2}(\Gamma;L) := L \hat{y}_{-}^{1,2} \left (\frac{\Gamma}{L};L\right )$. Assume that $y^1<y_{\max}<y^2$. Then $(\pi, y_-^1(\Gamma;L))$ is a center whereas $(\pi, y_-^2(\Gamma;L))$ is a saddle. 
    \item When $L\in [L_*(\gg_{\max}), L_{\max}]$ and $\Gamma= L \hat{\Gamma}_*^-(L)$ the unique fixed point $(\pi,y_-(\Gamma;L))$ is parabolic.  
\end{enumerate}
\end{proposition}
\begin{proof}
It is clear that if $\mathcal{E}_\pm>0$, the corresponding fixed point is a saddle and if $\mathcal{E}_\pm<0$, it is a center. 

Note that~\eqref{finexprmathcalE} implies $\mathcal{E}_\pm= \mp \partial_\Gamma \mathcal{F}_\pm \partial_y \Gamma_\pm$ and Corollary~\ref{prop:existencefixedpointselliptic} implies $\partial_\Gamma \mathcal{F}_\pm<0$. Then, when $\partial_y \Gamma_{+} >0$ the fixed point will be a saddle and if $\partial_y \Gamma_{\color{red}+} <0$ the fixed point will be a center. Conversely, when $\partial_y \Gamma_{-} >0$ the fixed point will be a center and if $\partial_y \Gamma_{-} <0$ the fixed point will be a saddle.

From Proposition~\ref{prop:fixedpointscasepos}, we have that  $\partial_y \Gamma_+>0$ which proves the first item of the lemma. 

For the values of $L, \Gamma$ in the second item, by Lemma~\ref{lem:Gammanegincreasing} (see also Figure~\ref{fig:abstractGammaneg}), $\partial_y \Gamma_- = L \partial_{\hat{y}} \hat{\Gamma}_->0$ that implies that $\mathcal{E}_-<0$ and the result follows. 

With respect to the third item, we have that $\hat{\Gamma}_-(\cdot; L)$ has a maximum at $\hat{y}_{\max}$ so that $\partial_{\hat{y}} \hat{\Gamma}_-(\hat{y};L)>0$ for $\hat{y}< \hat{y}_{\max}$ and negative if $\hat{y}>\hat{y}_{\max}$ and the result holds true since
$\hat{y}^1_- < \hat{y}_{\max}$ and  $\hat{y}^2_->\hat{y}_{\max}$. 

The last item follows from the fact that, if $L=L_*(\gg)$ and $\hat{\Gamma}=\hat{\Gamma}_*(L) $, then $\partial_{\hat{y}}\hat{\Gamma}_-(\hat{y}_{*}(\gg);L_*(\gg))=0$ where we recall that $\hat{y}_*(\gg)=\hat{y}_{\max}$ and 
$\hat{\Gamma}_*^-(L)=\hat{y}_*(\gg) \gg$. 
\end{proof}

Theorem~\ref{thm:AveragedHam_generalcase} is a straightforward consequence of Proposition~\ref{prop:fixedpointscharacter} keeping track of the range of values of $\hat{\Gamma}_{\min,\max}^{\pm}$ and $\hat{\Gamma}_*$. Indeed, we only need to rename $\hat \Gamma_1 = \hat \Gamma_{\min}^-$, 
$\hat \Gamma_2 = \Gamma_{\max}^-$, $\hat \Gamma_0= \hat \Gamma_{\min}^+$, $\hat \Gamma_3 = \hat \Gamma_{\max}^+$, $\hat \Gam_* = \hat \Gam_*^-$, $L_*=L_*(s_{\max})$  and to rewrite Proposition~\ref{prop:fixedpointscharacter} in the terms of Theorem~\ref{thm:AveragedHam_generalcase}, that is fixing the values of $L$. Notice that by 
Remark~\ref{rmk:positionGammaast} we have that 
$$
\hat \Gamma_1(L) \leq \hat \Gamma_0 \leq \hat \Gamma_2 (L) \leq \hat \Gamma_*(L) \leq \Gamma_3(L), 
$$
with the convection that $\hat \Gamma_*(L)=\Gamma_2(L)$ when $L\in [L_{\min}, L_*(s_{\max})]$. To finish, we note that by Proposition~\ref{prop:fixedpointscasepos}, $\hat \Gamma_3$ is increasing and by Lemma~\ref{lem:Gammanegincreasing}
$\hat \Gamma_{1,2}$ are decreasing.

\section{Periodic orbits of the coplanar Hamiltonian}\label{appendix_monodromy} 

\subsection{Existence of periodic orbits. Proof of Theorem~\ref{thm:existence_periodic_orbits}} \label{sec:proof_periodicorbits}

Consider the coplanar Hamiltonian $\HH_\CP = \HH_0 + \alpha^3 \HH_{\CP,1}$. 
We first emphasize that, by expressions of $\HH_0$ in~\eqref{def:hamiltonianPoincareH0} and ${\HH}_{\CP,1}$ in~\eqref{eq:expressionPoincareHCP1} in Appendix \ref{app:tables}, $(\eta,\xi)=(0,0)$ is invariant by the flow of $\HH_\CP$ and every orbit of the form
$(0,\Gamma(t;L),0,h(t;L))$ has to 
be in an energy level
$\HH_\CP(0,\Gamma,0,h)=\mathbf{E}$ for some energy $\mathbf{E}$
with $\Gamma(t), h(t)$ satisfying the differential equations in Theorem~\ref{thm:existence_periodic_orbits}. 

We now prove that for a certain range of energies $E$ (to be determined), the solutions $(0,\Gamma(t;L), 0,h(t;L))$ are periodic orbits which are a graph over the variable $h$.

Fix $E \in \mathbb{R}$ and $L\in (0,L_{\max} ]$. Writing $\Gamma= L\hat{\Gamma}$,  
\begin{equation}\label{energylevelperiodic}
\begin{aligned}
E= \widehat{\HH}_\CP(0,\hat \Gam,0,h):=&\frac{\rho_0}{16L^6} (1- 12 \hat{\Gam} - 12 \hat{\Gam}^2) + 
\alpha^3 \frac{\rho_1 U_2^{0,0}}{32 L^2} (1- 12 \hat{\Gam} - 12 \hat{\Gam}^2 )\\
&-\alpha^3 \frac{\rho_1 U_2^{1,0}}{8L^2} \sqrt{(1- 2\hat\Gam)(3+2\hat\Gam)} (2\hat\Gam+1) \cos h \\
&-\alpha^3 \frac{\rho_1U_2^{2,0}}{32 L^2} (1- 2 \hat\Gam)(3 + 2\hat\Gam) \cos 2h. 
\end{aligned}
\end{equation}

We impose that $\partial_{\hat \Gam} \widehat{\HH}_\CP(0,\hat \Gam,0,h) \neq 0$ for all $h$, in other words $\dot{h} \neq 0$. This condition will give a set of possible values for $L,\hat{\Gam}$. From the differential equations in Theorem~\ref{thm:existence_periodic_orbits}, we need that
\begin{equation*}
\begin{aligned}
 \dot{h} = &-\frac{3\rho_0(1+2\hat \Gam)}{4 L^7}
- \al^3 \frac{3\rho_1}{8L^3} 
 U_2^{0,0}(1 + 2\hat \Gam) \\ &
- \alpha^3 \frac{3\rho_1}{24L^3} \Big(
4 U_2^{1,0}
\frac{1-4\hat \Gam- 4 \hat \Gam^2}{\sqrt{(1-2\hat \Gam)(3+2\hat \Gam)}} \cos h
  -  U_{2}^{2,0}(1+2\hat \Gam)\cos(2h) \Big ) <0
\end{aligned}
\end{equation*}
for all $h\in [0,2\pi]$ or equivalently
\begin{equation}\label{conditionpartialGamma}
    \begin{aligned}
        (1+2 \hat \Gam)+\frac{1}{2} \rho \alpha^3 L^4 U_2^{0,0} (1+ 2\hat \Gam) & +\frac{2}{3} \rho \alpha^3 L^4 U_2^{1,0} \frac{1- 4 \hat \Gam - 4 \hat \Gam^2}{\sqrt{(1-2 \hat \Gam)(3 + 2 \hat \Gam)}} \cos h \\
        &- \frac{1}{6} \rho \alpha^3 L^4 U_2^{2,0} (1+ 2\hat \Gam) \cos 2h>0,
    \end{aligned}
\end{equation}
where $\rho$ has been introduced in \eqref{def:rho}.

To avoid cumbersome notations, we introduce 
\begin{align*}
A=A(\hat{\Gam};L) & = (1+2 \hat \Gam) + \frac{1}{2} \rho \alpha^3 L^4 U_2^{0,0} (1+ 2\hat \Gam) \\ 
B=B(\hat\Gam;L) & =    \frac{2}{3} \rho \alpha^3 L^4 U_2^{1,0} \frac{1- 4 \hat \Gam - 4 \hat \Gam^2}{\sqrt{(1-2 \hat \Gam)(3 + 2 \hat \Gam)}} \\
C=C(\hat\Gam;L) & = -\frac{1}{6} \rho \alpha^3 L^4 U_2^{2,0} (1+ 2\hat \Gam),
\end{align*}
so that condition~\eqref{conditionpartialGamma} reads as
\[
f(h)=f(h;\hat \Gam, L):=A + B \cos h + C\cos 2h>0.
\]
With respect to $h$, $f(h)$ has its global minimum either at $h=0,\pi$ or, if $\left | \frac{B}{4C}\right |<1$ at $h_{1,2}$ satisfying
$$
\cos h_1 = \cos h_2 = -\frac{B}{4C}.
$$
Using that $\cos 2 h = 2 \cos^2 h -1$, we have that, 
$$
f(h)=A+B\cos h +C (2\cos^2 h -1)
$$
and then when $\left | \frac{B}{4C}\right |<1$
$$
f(h_1)=f(h_2)=A - \frac{B^2}{8C} -C.
$$
Since $A>0$ and $C<0$, $f(h_1), f(h_2)$ are positive. 
Therefore we only need to impose 
$f(0),f(\pi)>0$ 
Notice that 
$f(0)=A+B+C$ and $f(\pi)=A-B+C$ so that both conditions can be written as
$$
A+C>|B|.
$$
When $|B|< 4|C| $, this last condition is satisfied provided $a_{\max} = 30000\, \mathrm{km}$. Indeed, in this case 
$$
 {5|C|} \leq \frac{ {5}}{6} \rho \frac{a_{\max}^3}{a_{\rM}^3} L_{\max}^{4} U_2^{2,0} (1+ 2\hat \Gam) =  0.023691360650697 (1+ 2 \hat \Gam) < A
$$
and then $|B| <  4 |C| \leq 5|C| + C < A+ C $. 
Summarizing we only need to impose
\begin{equation}\label{conditionpartialGammaABC} 
        A+C- |B|>0 \quad  \text{if    }\quad  \left |\frac{B}{4C} \right | \geq 1  
\end{equation}
\begin{lemma} \label{lem:firstcondition_op} For $\hat \Gam \in [0,0.49]$ and $L\in [L_{\min},L_{\max}]$, 
we have that $\frac{5}{4} |B(\hat \Gam, L)| \leq 1 <A$. 
\end{lemma}
\begin{proof}
   It is clear that 
   $$
   \frac{5}{4} |B(\hat \Gam,L)|\leq \frac{5}{6} \rho \frac{a_{\max}^3 }{a_{\rM}^3} L_{\max}^4 U_2^{1,0} \frac{|1- 4 \hat \Gam- 4 \hat \Gam^2|}{\sqrt{(1- 2 \hat \Gam)(3 + 2 \hat \Gam)}} \leq 
   \frac{5}{3} \rho \frac{a_{\max}^3 }{a_{\rM}^3} L_{\max}^4 U_2^{1,0} \frac{1}{\sqrt{3(1-2 \cdot 0.49)}}.
    $$
Computing this value, we have that 
$|B(\hat \Gam,L)| \leq 0.446157052927936$ and we are done. 
\end{proof}

Lemma~\ref{lem:firstcondition_op} implies that condition $A+C-|B|>0$ is always satisfied  if $\hat \Gam \in [0,0.49]$ and $|B|\geq 4 |C|$ because
$
|B|-C \leq \frac54 |B| < A.
$

As a consequence of the previous analysis, 
\[
\partial_{\hat \Gamma}\widehat{\HH}_\CP\neq 0 \qquad\text{for}\quad L\in[L_{\min},L_{\max}], \,h\in[0,2\pi],\,\hat\Gamma\in [0,0.49].
\]
Therefore, any energy level $\widehat{\HH}_\CP=E$ which belongs to the cylinder $h\in [0,2\pi], \hat \Gamma\in[0,0.49]$ is a closed curve, which is a graph over $h$ and moreover its dynamics is periodic. Then, it only remains to characterize such  energy levels.
 
We notice that, since $L\hat{\Gamma}(t;L), h(t;L)$ satisfy the differential equations in Theorem~\ref{thm:existence_periodic_orbits}, one deduces that $\dot{\hat \Gam} =0$ if and only if $h=0,\pi, 2\pi$. Indeed, $\dot{\hat \Gam}=0$ is equivalent to
$$
\sin h \left (U_2^{1,0} (2\hat \Gam + 1) + U_2^{2,0} \sqrt{(1-2\hat \Gam) ( 3 + 2 \hat \Gam) }\cos h \right )=0 
$$
and since 
$$
- \frac{U_2^{1,0}}{U_2^{2,0}} \frac{2\hat \Gam +1}{\sqrt{(1-2\hat \Gam) ( 3 + 2 \hat \Gam) }} \leq 
-\frac{U_2^{1,0}}{U_2^{2,0} \sqrt{3}} = -1.331627517097770
<-1
$$
the conclusion is obvious. Clearly $h=\pi$ corresponds to a maximum and $h=0$ corresponds to a minimum of $\hat \Gam$ as function of $h$. 


We denote by
$
\hat\Gam_{\min}(E), \hat \Gam_{\max}(E)
$ the values of $\hat\Gam$ such that 
\[
\widehat{\HH}_\CP(0,\hat \Gam_{\min}(E),0,0)=E, \qquad \widehat{\HH}_\CP(0,\hat \Gam_{\max}(E),0,\pi)=E.
\]
The level curves $E= \widehat{\HH}_\CP (0, \hat \Gam, 0, h)$ cannot intersect and, since $\partial_{\hat \Gam} \widehat{\HH}_\CP \neq 0$, we have that 
$\partial_E \hat \Gam_{\min}(E), \partial_E \hat \Gam_{\max}(E) \neq 0$. Therefore, in order to check the range of $E$ allowed in our analysis, we  compute 
\[
\mathbf{E}_{\max} = \mathbf{E}_{\max}(L)=\widehat{\HH}_\CP (0,0,0,0), \qquad \mathbf{E}_{\min}=  \mathbf{E}_{\min}=(L)\widehat{\HH}_\CP(0,0.49,0,\pi)
\]
and Theorem~\ref{thm:existence_periodic_orbits} is proven, taking into account that, from~\eqref{energylevelperiodic}, we easily deduce that  
$$
\mathbf{E}_{\max}(L)= \frac{\rho_0}{16L^6} \widehat{\mathbf{E}}_{\max}(\delta), \qquad \mathbf{E}_{\min}(L)= \frac{\rho_0}{16L^6} \widehat{\mathbf{E}}_{\min}(\delta)
$$
for some function $\widehat{\mathbf{E}}_{\min,\max}$ defined for $\delta= \rho \alpha^3 L^4 \in [0, \delta_{\max}]$ (recall that $\delta_{\max}$ is defined in~\eqref{def:Lalfadelta_max}).  

%

To obtain the values in Remark \ref{rmk:energyvalues}, it is enough to recall that, for Galileo,  $a=29600\, \mathrm{km}$ (equivalently $L=1$),  and therefore 
$\mathbf{E}_{\min,\max}(1)$ with 
\[
\mathbf{E}_{\max}(1)= 2.477266122798186 \cdot 10^{-6}, \qquad \mathbf{E}_{\min}(1)= -2.558100888960067 \cdot 10^{-5}.
\]


\subsubsection{Proof of Corollary~\ref{cor:gammaCPAV}} \label{proof:corollary_periodic}
    We fix $\mathbf{E}$, a given energy level, and let $\Gam_\CP^{\mathbf{E}}$ be such that $\HH_{\CP}(0,\Gam_\CP^{\mathbf{E}}, 0,0)=\mathbf{E}$. Then, using formula~\eqref{energylevelperiodic} for $\widehat{\HH}_{\CP} (0, \hat \Gam, 0 ,h)= \HH_{\CP} (0, L \hat \Gam,0,h)$ one has that 
$$ 
\frac{16 L^6}{\rho_0} \mathbf{E} = \widehat{\HH}_{\AV} (0,\hat \Gam_\CP^{\mathbf{E}} ,0) + \delta \widehat{\HH}^*_{\CP,1}(0,\hat \Gam_\CP^{\mathbf{E}} , 0, 0), \qquad \hat \Gam_\CP^{\mathbf{E}} = L^{-1} 
\Gam_\CP^{\mathbf{E}},
$$
with $\delta =\rho \alpha^3 L^4$, 
\begin{align*}
\widehat{\HH}_{\AV} (0,\hat \Gam,0) &= 
 { \frac{16 L^6}{\rho_0} \HH_{\AV}(0,L\hat\Gam,0) }
=
(1-12 \hat \Gam -12 \hat \Gam^2) \left (1+  {\delta} \frac{U_{2}^{0,0}}{2} \right ), 
\\
\widehat{\HH}_{\CP,1}^{ {*}}(0,\hat \Gam, 0,0 ) 
&=  {\frac{16 L^2}{\rho_1} \left(
\HH_{\CP,1}(0,L\hat\Gam,0,0)-\HH_{\AV,1}(0,L\hat\Gam,0) \right)}
\\
&= - 2   U_2^{1,0}\sqrt{(1- 2\hat\Gam)(3+2\hat\Gam)} (2\hat\Gam+1)  - \frac{U_2^{2,0}}{2} (1- 2 \hat\Gam)(3 + 2\hat\Gam)  . 
\end{align*} 
Let $\hat \Gam_\AV^{\mathbf{E}}$ be such that 
$$
\frac{16 L^6}{\rho_0}\mathbf{E} = \widehat{\HH}_{\AV} (0,\hat \Gam_\AV^{\mathbf{E}} ,0).
$$
One has, from definition~\eqref{defhata0c0} of $\hat a_0,\hat c_0$, that
$$
\hat \Gam_\CP^{\mathbf{E}} - \hat \Gam_\AV^{\mathbf{E}} = -\delta \frac{ \widehat{\HH}^*_{\CP,1} (0, \hat \Gam_{\CP}^{\mathbf{E}},0,  0)}{\partial_{\hat \Gam} \widehat{\HH}_{\AV}  (0,\hat \Gam_\AV^{\mathbf{E}},0)} + \mathcal{O}(\delta^2) = -\delta \hat c_0 \big (\hat \Gam_{\CP}^{\mathbf{E}};\delta  \big ) + \mathcal{O}(\delta^2 )
$$
and the proof is finished.
  
\subsection{Perturbative analysis. } 
We fix $L\in [L_{\min},L_{\max}]$ and we consider the Hamiltonian $\HH_\CP$ in~\eqref{def:Ham:Coplanar}. Let $(\eta,\Gam,\xi,h)=(0,\Gam(t;L),0,h(t;L))$ be a periodic orbit satisfying $\Gam(0;L)= \Gam_0$ and $h(0;L)=0$ (see Theorem~\ref{thm:existence_periodic_orbits}). We recall that in~\eqref{def:delta:mainresults} we have introduced the parameter $\delta = \alpha^3 L^4 \rho$. 

\subsubsection{Proof of Proposition~\ref{prop:firstaproximation_periodic}}\label{sec:perturbative analysis_periodic}

The proof of Proposition~\ref{prop:firstaproximation_periodic} relies on a simple perturbative analysis. In order to make our analysis independent on $L$, we perform the change    
$$
\hat \Gam = \Gam/L, \qquad s= 3\rho_0 t / (4L^7)
$$
to the differential equation in Theorem~\ref{thm:existence_periodic_orbits} and we obtain the new system 
\begin{align} \label{diff_eq_hatGammah}
\hat {h}' =  
& \, a_0(\hat \Gam;L)+ \delta   a_1(\hat \Gam, \hat h;L) \notag\\ :=  &-1-2\hat \Gam
- \al^3 L^4  \frac{ \rho }{2 } 
 U_2^{0,0}(1 + 2 \hat \Gam)
+ \delta \frac{1}{6} \Big(
4 U_2^{1,0}
\frac{1-4 \hat \Gam-4 \hat \Gam^2}{\sqrt{(1-2\hat \Gam)(3+2 \hat \Gam)}} \cos \hat h
\notag \\
& -  U_{2}^{2,0}(1+2\hat \Gam)\cos(2\hat h) \Big)\\
\hat \Gam' =& -\delta \frac{1}{12}
\Big( 
2 U_2^{1,0}
\sqrt{(1-2\hat \Gam)(3+2 \hat \Gam)}(2\hat \Gam+1) 
\sin \hat h 
 + U_2^{2,0} 
(1-2\hat \Gam)(3+2 \hat \Gam)\sin(2 \hat h)
\Big) \notag.
\end{align}
By the analyticity  of system~\eqref{diff_eq_hatGammah} with respect to $\delta$, we have that 
\[
\hat \Gamma(s) = \hat \Gamma_0 + \delta \hat \Gamma_1(s) +   \delta^2 \hat \Gam_{2}(s,\de) , \qquad \hat h(s) = \hat h_0(s) + \delta \hat h_1(s) + \delta^2 \hat h_2(s, \de)
\]
with $\hat \Gam_0 \in \mathbb{R}$ the initial condition and
\begin{equation*}
\begin{aligned}
\hat h'_0 = &  a_0( \hat \Gam_0;L), \\ 
\hat \Gam'_1 = &  - \frac{1}{12}
\Big(
2 U_2^{1,0}
\sqrt{(1-2\hat \Gam_0)(3+2 \hat \Gam_0)}(2\hat \Gam_0+1) 
\sin \hat h_0(s)
\\
&+ U_2^{2,0} 
(1-2\hat \Gam_0)(3 +2\hat \Gam_0)\sin(2 \hat h_0(s))
\Big), \\ 
\hat h'_1 =  &\partial_{\hat \Gam }  a_0(\hat \Gam_0;L) \hat \Gam_1(s) +  a_1(\hat \Gam_0, \hat h_0(s);L).
\end{aligned}
\end{equation*}
Then, imposing that $\hat h(0)=0$ and $\hat \Gam(0)= \hat \Gamma_0$, that is $\hat \Gam_1(0)=0$, we have that
\begin{align*}
h_0(s) = &  a_0(\hat \Gam_0;L) s \\
\hat \Gam_1(s) = &     c_0(\hat \Gam_0;L) +   c_1( \hat \Gam_0;L) \cos \big (   a_0( \hat \Gam_0;L) s \big ) +   c_2(\hat \Gam_0;L) \cos \big ( 2 a_0( \hat \Gam_0;L) s \big ) \\ 
h_1(s)  =&  \partial_{\hat \Gam}  a_0(\hat \Gam_0;L)   c_0(\hat \Gamma_0;L) s+ d_1(\hat \Gam_0;L) \sin \big (  a_0(\hat \Gam_0;L) s \big ) +  d_2(\hat \Gam_0;L) \sin \big ( 2  a_0(\hat \Gam_0;L) s \big )
\end{align*}
where it is straightforward to check that the constants $a_0,c_{0,1,2}, d_{ 1,2}$ correspond to $\hat a_0,\hat{c}_{0,1,2}, \hat{d}_{1,2}$, the ones defined in~\eqref{defhata0c0} and~\eqref{defhata1c1d1}, taking into account that the dependence with respect to the parameter $L$ comes from the fact that 
$$
a_0(\hat \Gam_0;L)= \hat {a}_0\left (\hat \Gam_0; \rho \alpha^3 L^4  \right ).
$$
With respect to the period of the periodic orbit, if $\hat h(\widehat{\mathcal{T}}(\hat \Gam_0; \delta)) = 2\pi$, that is 
$\widehat{\mathcal{T}}(\hat{\Gam}_0; \delta )$ is the period of the periodic orbit, then
$$ 
\widehat{\mathcal{T}} (\hat \Gam_0;\delta )
= \frac{2\pi + \mathcal{O}(\delta^2) }{ |\hat a_0 (\hat \Gam_0;\delta ) + \delta \partial_{\hat \Gam} \hat a_0(\hat \Gam_0;\delta ) \hat c_0 (\hat \Gam_0;\delta )|}= 
\frac{2\pi + \mathcal{O}(\delta)}{ |\hat a_0(\hat \Gam_0;\delta ) |}.
$$ 
Undoing the change of coordinates we have that 
$\mathcal{T}(\Gam_0;L)= \widehat{\mathcal{T}}(\Gam_0 L^{-1}; \delta)  {\frac{4L^7}{3\rho_0}}$ and 
$$
\Gam (t;L)= L \hat \Gam\left ( {\frac{3\rho_0}{4L^7}} t ;\rho \alpha^3 L^4\right ),
\qquad 
h(t;L)= h \left ( {\frac{3\rho_0}{4L^7}} t ;\rho \alpha^3 L^4\right ).
$$
This completes the proof of Proposition~\ref{prop:firstaproximation_periodic} .

\subsubsection{Proof of Theorem~\ref{prop:eigenvalues_monodromy}}\label{sec:sub_eigenvalues_perturbative}

Along this section, whenever there is no danger of confusion we will omit the dependence on $L\in [L_{\min},L_{\max}]$. 

It is convenient to introduce the (small) parameter $\epsilon = \alpha ^3 $ and 
$X(\xi,\eta, \Gam, h;\epsilon)$\footnote{With this new order in the variables the linearized vector field around the periodic because has block structure.}, the vector field associated to $\HH_\epsilon:=\HH_\AV  + \epsilon(\HH_{\CP,1} - \HH_{\AV,1} )$.
 
The first step is to characterize the variational equation of $X$ around the periodic orbit, $(0,0,\Gam(t),h(t))$.   
Note that  
$$
DX(0,0, \Gam(t),h(t);\epsilon)= \left (\begin{array}{cc} \partial_{\xi,\eta} X_{\xi,\eta}(0,0,  \Gam(t), h(t);\epsilon) & 0_{2\times 2} \\
0_{2\times 2} & \partial_{\hat \Gam,h} X_{  \Gam,h} (0,0,  \Gam(t), h(t);\epsilon) \end{array} \right ).
$$
Then, if $\Phi(t;\epsilon)$ is the fundamental matrix of the linearized system
$$
\dot{z} = DX(0,0,  \Gam(t),h(t);\epsilon)z
$$
satisfying $\Phi(0;\epsilon)=\mathrm{Id}$, it has also the block form
\begin{equation}\label{def:fundamentalmatrix_1}
\Phi(t;\epsilon) = \left (\begin{array}{cc} \Psi(t;\epsilon) & 0_{2\times 2} \\ 0_{2\times 2} & \widetilde {\Psi}(t;\epsilon) \end{array} \right ), \qquad \Psi(0;\epsilon)=\widetilde{\Psi}(0;\epsilon)=\mathrm{Id}.
\end{equation}
Since the system is Hamiltonian, $\mathrm{det} \Phi(t;\epsilon)=1$ for all $t$ and therefore we also have that $\mathrm{det} \Psi(t;\epsilon)= \mathrm{det} \widetilde{\Psi}(t;\epsilon)=1$. 

The character of the periodic orbit $(0,0,\Gam(t),h(t))$, is determined by the eigenvalues of the monodromy matrix $\Psi(\mathcal{T}; \epsilon)$, with $\mathcal{T} = \mathcal{T}(\Gam(0);L)$ the period of the periodic orbit. To study the monodromy matrix, we first notice that 
\[
\partial_{\xi,\eta} X_{\xi,\eta} (0,0,\Gam(t) , h(t); \epsilon)=
\partial_{\xi,\eta} X_{\AV} (0,0,\Gam(t)) + \epsilon \frac{\rho_1}{L^2} \tilde{B}(t)
\] 
where $X_{\AV}(\xi,\eta,\Gam)$ is the vector field associated
to the $h$-averaged system $\HH_\AV$ and,  using formulas~\eqref{def:Ham:AV1} and~\eqref{eq:expressionPoincareHCP1} of $\HH_{\AV,1}$ and $\HH_{\CP,1}$ respectively, it is not difficult to see that the matrix $\tilde{B}$ has the form
\[
\tilde{B}(t) = \left ( \begin{array}{cc} \sum_{j=1}^3  \tilde{b}_{11}^j(\Gamma(t),L) \sin j h(t) & \sum_{j=1}^3 \tilde{b}_{12}^j( \Gam(t),L) \cos j h(t) \\
\sum_{j=1}^3 \tilde{b}_{21}^j(\Gam(t),L) \cos j h(t) & - \sum_{j=1}^3  \tilde b_{11}^j(\Gamma(t),L) \sin j h(t) \end{array} 
\right ).
\]
The elements $\tilde{b}_{ij}$ depend on $L,\Gam$ through the functions $\mathcal{D}_{l,k} (0, \Gam,0;L)$ (see Table~\ref{tab:functionsPoincare}). For future computations, we note that $\tilde{b}_{ik}^j$ are linear functions on $\mathcal{D}_{i,k} (0,\Gam,0;L)$, namely
\begin {equation}\label{tildebij}
\tilde{b}_{ik}^j (\Gamma,L)= \mathbf{b}_{ik}^j \big ( \{\mathcal{D}_{i,k}(0,\Gam,0;L) \}_{i=0,1,2,\, k=0,2}, L \{\mathcal{D}_{i,1}(0,\Gam,0;L)\}_{i=1,2}  \big),  
\end{equation}
with $\mathbf{b}_{ik}^j$ linear functions. Therefore the matrix $\Psi(t;\epsilon)$ (see~\eqref{def:fundamentalmatrix_1}) is the fundamental matrix of the two dimensional linear system
\begin{equation} \label{first_variational_epsilon}
\dot{\zeta} =  \left [\partial_{\xi,\eta} X_{\AV} (0,0, \Gam(t)) + \epsilon \frac{\rho_1}{L^2} \tilde{B}(t) \right ] \zeta, \;\;\text{such that 
}\;\;\Psi(0,\epsilon)=\mathrm{Id}, \;\mathrm{det}\, \Psi(t,\epsilon)= 1.
\end{equation}


\begin{lemma}\label{lem:scaled_eigenvalues}
Let $(0,0,\Gam(t),h(t))$ be a $\mathcal{T}$-periodic orbit of the coplanar Hamiltonian $\HH_\CP$ with 
$\Gam(0)=\Gam_0, h(0)=0$. 

Consider the new variable $\hat \Gam = \Gam /L$, the new independent time $s= 3\rho_0 t / (4L^7) $, the initial condition $\hat \Gam_0 = \Gam_0 /L$ and the parameters $\delta$ and $\sigma_0$ defined by 
\begin{equation*}
\delta = \epsilon L^4 \rho = \alpha^3 L^4 \frac{\rho_1}{\rho_0},\qquad \sigma_0=\sigma_0(\hat \Gam_0;\delta) = \hat \Gam_0 + \delta \hat c_0(\hat \Gam_0;\delta)
\end{equation*}
with $\hat c_0$ introduced in~\eqref{defhata0c0}. Then
\begin{enumerate}
\item In the new variables, the variational equation~\eqref{first_variational_epsilon} becomes
\begin{equation}\label{final_variational_delta}
\zeta' = \left [ \mathbf{X}(\sigma_0 )   + \delta \widehat{B}(s) + \mathcal{O}(\delta^2)\right ]\zeta, \qquad 
\int_{0}^{\widehat{\mathcal{T}}} \widehat{B}(u)\, du=\mathcal{O}(\delta), 
\end{equation}
where $\mathbf{X}(\sigma_0)$ is defined as
\begin{equation} \label{expression_variational_periodic}
\begin{aligned}
\mathbf{X}(\sigma_0 ) := &\frac{4 L^7}{3 \rho_0} \partial_{\xi,\eta} X_{\AV} (0,0, L\sigma_0) \\ 
= &\frac{1}{4} \left (\begin{array} {cc} 0 & -(\mathbf{a} (\sigma_0) + \delta \mathbf{c}_-(\sigma_0) )\\  \mathbf{a} (\sigma_0) + \delta \mathbf{c}_+(\sigma_0) & 0 \end{array}\right )
\end{aligned}
\end{equation}
where $\mathbf{a}(\sigma) = 1-16 \sigma -20 \sigma^2$ has been defined previously in~\eqref{defaborigin_critical} and $\mathbf{c}_\pm$ have been introduced in Remark~\eqref{rmk:first_expression_eigenvalues}:
\begin{equation}\label{def:mathbfc_theorem}
\mathbf{c}_\pm (\sigma)= \frac{U_2^{0,0}}{2}\mathbf{a} (\sigma) \pm  \frac{5 U_{2}^{1,0}}{3} \mathbf{b}(\sigma), \qquad \mathbf{b}(\sigma)=(2\sigma +3) \sqrt{3-4\sigma -4 \sigma^2}. 
\end{equation}
\item The fundamental matrix $\widehat \Psi (s,\delta)$ of~\eqref{final_variational_delta} satisfies that 
$$
\widehat{\Psi} (s, \delta) = \Psi \left ( \frac{3 \rho_0}{4L^7} s, \frac{\delta \rho_0}{\rho_1 L^4} \right ), \qquad 
\widehat{\Psi}(0,\delta) = \mathrm{Id}, \qquad \mathrm{det} \widehat{\Psi}(s,\delta)=1.
$$
\end{enumerate}
\end{lemma}
\begin{remark}The change of variables, presented in the previous lemma, allows to keep our perturbative analysis uniform in $L\in [L_{\min}, L_{\max}]$.
\end{remark}  
\begin{proof} 
Performing the change $\hat \Gam = \Gam/L$ and $s= 4L^7 t /(3\rho_0)$,   system~\eqref{first_variational_epsilon} becomes
\begin{equation} \label{second_variational_delta}
\zeta' =    \frac{4 L^7}{3 \rho_0} \partial_{\xi,\eta} X_{\AV} (0,0, L \hat \Gam(s)) \zeta  + \delta \frac{4 L}{3 } \tilde{B}\left ( \frac{4 L^7}{3 \rho_0} s \right )  \zeta .
\end{equation} 
Using Table~\ref{tab:functionsPoincare} for the explicit expressions of $\mathcal{D}_{l,k}$, we obtain that 
\begin{align*}
\mathcal{D}_{l,k}(0,L\hat \Gam, 0;L) &= \frac{1}{L} \mathbf{d}_{l,k} (\hat \Gam) , & \qquad &  l=0,1,2,\,k=0,2 \\ 
\mathcal{D}_{l,1}(0,L\hat \Gam, 0;L) &= \frac{1}{L^2} \mathbf{d}_{l,1} (\hat \Gam), & 
\qquad & l=0,1,2,
\end{align*}
for smooth functions $\mathbf{d}_{l,k}$. 
Then, by formula~\eqref{tildebij} of $\tilde{b}_{ij}$ (we recall that the functions $\mathbf{b}_{ij}$ are linear), 
we conclude that the equation~\eqref{second_variational_delta} can be rewritten as 
\begin{equation} \label{third_variational_delta}
\zeta' =    \frac{4 L^7}{3 \rho_0} \partial_{\xi,\eta} X_{\AV} (0,0, L \hat \Gam(s)) \zeta  + \delta \mathcal{ B} (s)    \zeta 
\end{equation}
with   
$$
\mathcal{B}(s) = \left ( \begin{array}{cc} \sum_{j=1}^3  {b}_{11}^j(\hat \Gam(s)) \sin j \hat h(s) & \sum_{j=1}^3 {b}_{12}^j( \hat \Gam(s)) \cos j \hat h(s) \\
\sum_{j=1}^3  {b}_{21}^j(\hat \Gam(s)) \cos j \hat h(s) & - \sum_{j=1}^3  b_{11}^j(\hat \Gamma(s)) \sin j \hat h(s) \end{array} 
\right ).
$$
As a consequence, the fundamental matrix $\widehat \Psi (s,\delta)$ of~\eqref{third_variational_delta} satisfies that 
$$
\widehat{\Psi} (s, \delta) = \Psi \left ( \frac{4L^7}{3 \rho_0} s, \frac{\delta \rho_0}{\rho_1 L^4} \right ), \qquad 
\widehat{\Psi}(0,\delta) = \mathrm{Id}, \qquad \mathrm{det} \widehat{\Psi}(s,\delta)=1.
$$
This proves the second item in the lemma, provided that system~\eqref{third_variational_delta} coincides with~\eqref{final_variational_delta}.

Now we study $\partial_{\xi,\eta} X_{\AV}$. By~\eqref{origin_critical_variational},
$$
\partial_{\xi,\eta} X_{\AV} (0,0,L \hat \Gam) = \frac{\rho_0}{8 L^{10}} \left ( \begin{array}{cc} 0 & -X^{\eta} (L \hat \Gam; L) \\ X^{\xi} (L \hat \Gam;L) & 0 \end{array}\right ).
$$
Then, using the expressions~\eqref{Xxietasigma} for $X^{\xi,\eta} (L\hat \Gam ;L)$ and the formulas for $d_0,d_1$ in~\eqref{def:d0d1}, we conclude that 
\begin{equation}\label{expXAVhatGamma} 
 \frac{4 L^7}{3 \rho_0} \partial_{\xi,\eta} X_{\AV} (0,0, L  \hat \Gam) 
= \frac{1}{4} \left (\begin{array} {cc} 0 & -(\mathbf{a} ( L  \hat \Gam) + \delta \mathbf{c}_-( L  \hat \Gam ) )\\  \mathbf{a} ( L  \hat \Gam )+\delta \mathbf{c}_+( L  \hat \Gam ) & 0 \end{array}\right )
\end{equation}
with $\mathbf{a}(\sigma) = 1-16 \sigma -20 \sigma^2$ (and $\mathbf{b}(\sigma) =(2\sigma +3) \sqrt{3-4\sigma -4 \sigma^2}$)  defined in~\eqref{defaborigin_critical} and
$$
\mathbf{c}_\pm (\sigma)= \frac{U_2^{0,0}}{2}\mathbf{a} (\sigma) \pm  \frac{5 U_{2}^{1,0}}{3} \mathbf{b}(\sigma). 
$$
On the other hand, since Proposition~\ref{prop:firstaproximation_periodic} implies that $\hat \Gam(s) =  \hat \Gam_0 +  \delta \hat c_0(\hat \Gam_0)  + \delta b(s) + \mathcal{O}(\delta^2) $, 
with 
\[
\int_{0}^{\widehat{\mathcal{T}}} b(s) \, ds= \int_{0}^{\frac{2\pi}{|\hat{a}_0|}} b(s)\, ds + \mathcal{O}(\delta)= \mathcal{O}(\delta),
\]  
 we have that the variational equation for $\widehat{\Psi}$ in~\eqref{third_variational_delta} is of the form  
\[
\zeta' =  \left [\frac{4 L^7}{3 \rho_0} \partial_{\xi,\eta} X_{\AV} (0,0, L (\hat \Gam_0 + \delta \hat{c}_0(\hat \Gam_0)))      + \delta \widehat{B}(s) + \mathcal{O}(\delta^2) \right  ] \zeta, \qquad \int_{0}^{\widehat{\mathcal{T}}} \widehat{B}(s) \,dt = \mathcal{O}(\delta) 
\]
where $\partial_{\xi,\eta} X_\AV (0,0,L (\hat \Gam_0+ \delta \hat c_0(\hat \Gam_0)))$ is a constant matrix (see~\eqref{expXAVhatGamma}) studied in Section~\ref{sec:circular_case_Appendix}, and the $\mathcal{O}(\delta^k)$ are uniform in $\hat \Gam_0, L$. Taking $
\sigma_0 = \sigma_0(\hat \Gam_0;\delta)= \hat \Gam_0 + \delta \hat c_0(\hat \Gam_0)
$ and using again~\eqref{expXAVhatGamma} the result holds true.  

\end{proof}
We write the fundamental matrix $ \widehat{\Psi}$  as
$$
 \widehat{\Psi} ( s;\delta) =  \widehat{\Psi}_0(s;\sigma_0) + \delta \widehat{\Psi}_1(s;\sigma_0) + \delta^2\widehat{\Psi}_2(s;\sigma_0,\de),  
$$
where, $\widehat{\Psi}_0(s;\sigma_0), \widehat{\Psi}_1(s;\sigma_0)$ depend on $\delta$ through the parameter $\sigma_0$ and $\widehat{\Psi}_2$ is bounded. 

The matrix $\widehat{\Psi}_0(s;\sigma_0)$ is the fundamental matrix of the constant coefficients linear system 
$$
\zeta'=  \mathbf{X}(\sigma_0) \zeta, \qquad \text{with} \qquad \widehat{\Psi}_0(0;\sigma_0)=\mathrm{Id},
$$
with eigenvalues $e^{\pm \hat{\lambda}_0(\sigma_0)}$ satisfying, from expression~\eqref{expression_variational_periodic} of $\mathbf{X}(\sigma_0)$,  
\begin{equation}\label{expression_new_eigenvalues}
|\hat \lambda_0 (\sigma_0)|  = \frac{1}{4}\sqrt{ \big |\mathbf{a} (\sigma_0) + \delta \mathbf{c}_+(\sigma_0) \big |
\big |\mathbf{a} (\sigma_0) + \delta \mathbf{c}_-(\sigma_0) \big |}.
\end{equation}
The matrix  $\widehat{\Psi}_1(s;\sigma_0)$ satisfies 
$$
 \widehat{\Psi}_1 '(s;\sigma_0)=\mathbf{X}(\sigma_0) \widehat{\Psi}_1(s;\sigma_0) + \widehat B(s)  \widehat{\Psi}_0(s;\sigma_0),\qquad \widehat{\Psi}_1(0;\sigma_0)=0 
$$  
and therefore  
\begin{equation}\label{formula_hat_Psi1}
 \widehat{\Psi}_1(s;\sigma_0) =  \widehat{\Psi}_0(s;\sigma_0) \int_{0}^s 
( \widehat{\Psi}_0(u;\sigma_0 )^{-1} \widehat{B}(u)   \widehat{\Psi}_0(u;\sigma_0) du.
\end{equation}

Whenever there is no danger of confusion, we will omit the dependence on $\hat \Gam_0,\sigma_0,\delta$. In particular, we write  $\widehat{\mathcal{T}} = \widehat{\mathcal{T}}(\hat \Gam_0,\delta)$, $\widehat{\Psi}_0=\widehat{\Psi}_0(\widehat{\mathcal{T}}, \sigma_0)$, $ \widehat{\Psi}_1=\widehat{\Psi}_1(\widehat{\mathcal{T}}, \sigma_0)$ and $\widehat{\Psi}=\widehat{\Psi}(\widehat{\mathcal{T}},\delta)$. In addition, we will denote by $M$ a generic constant, independent on $\hat \Gam_0, \delta$ that, along the proof, can change its value. 

In order to prove Theorem~\ref{prop:eigenvalues_monodromy}, we need to compare the eigenvalues of $\widehat \Psi  $ with the ones of $\widehat{\Psi}_0  $. To do so, we first compare the trace of $\widehat{\Psi}$ with the one of $\widehat{\Psi}_0$.  

\begin{lemma} \label{lem:trace_fondamental_matrix}
Let $\hat \lambda_0 = \hat \lambda_0 (\sigma_0)$ be such that the eigenvalues of $\mathbf{X}(\sigma_0)$ are $\pm \hat \lambda_0$. Then, there exists $\delta_0>0$ such that, if $\delta \in [0,\delta_0]$,
$$
\mathrm{tr}( \widehat{\Psi}) = \mathrm{tr}(\widehat{\Psi}_0)+ \mathcal{O}(\mathbf{L}  \delta^2) + \delta^3 T(\delta)
$$
with $T(\delta)$ bounded on $[0,\delta_0]$ and
$$
\mathbf{L}:= \min_{k\in \mathbb{N}} \{  |\widehat{\mathcal{T}} |\hat \lambda_0| -  \pi k|\}  .
$$
\end{lemma}
\begin{remark}\label{rem:mathbfL} 
The quantity $\mathbf{L}$ is well defined and $0\leq \mathbf{L}\leq \frac{\pi}{2}$. Indeed, for a given $\widehat{\mathcal{T}}|\hat \lambda_0| $, the integer part  $k'= \left [ \frac{\widehat{\mathcal{T}} |\hat\lambda_0|}{\pi} \right]  \in \mathbb{N}$ satisfies  
$\widehat{\mathcal{T}} |\hat \lambda_0| -k' \pi\in [0,\pi)$. 
Therefore if 
$$\widehat{\mathcal{T}} |\hat \lambda_0| -k' \pi\in \left  [0,\frac{\pi}{2}\right ],\qquad \text{implies } \qquad \mathbf{L} = \widehat{\mathcal{T}} |\hat \lambda_0| -k' \pi \leq \frac{\pi}{2}.
$$
On the contrary, when $\widehat{\mathcal{T}} |\hat \lambda_0| -k' \pi\in \left  [ \frac{\pi}{2}, \pi\right  )$, $\mathbf{L} = | \widehat{\mathcal{T}} |\hat \lambda_0| -(k'+1) \pi |\leq \frac{\pi}{2}$.
\end{remark}
\begin{proof} 

Let $\kappa>1$ big enough (its value will be determined later). We distinguish two cases, namely $\mathbf{L}\geq \kappa \delta$ and $\mathbf{L}\leq \kappa \delta$. 

When $| \widehat{\mathcal{T}} \hat \lambda_0|\leq \kappa \delta$, from~\eqref{expression_new_eigenvalues} and Remark~\ref{rem:period}, we deduce that $|\mathbf{a}(\sigma_0)| \leq M \kappa \delta$. Then, since $\mathbf{a}(\sigma)= 1 - 16\sigma - 20 \sigma^2$,  its unique positive zero is $\frac{\sl}{2}$, with $\sl$ defined in~\eqref{def:pendentresonance} and therefore the condition $|\widehat{\mathcal{T}}\hat \lambda_0| \leq \kappa \delta$ implies that 
$$
\left |\sigma_0 - \frac{\sl}{2} \right |\leq M\kappa \delta
$$
for some constant $M$. Therefore, using that 
$$
\mathbf{X}(\sigma_0)= \mathbf{X}\left (\frac{\sl}{2}\right ) + \mathcal{O}(\delta) = \mathcal{O}(\delta)
$$
we obtain that 
$$
\widehat{\Psi}_0(\widehat{\mathcal{T}};\sigma_0)= \widehat{\Psi}_0 \left (\widehat{\mathcal{T}};\frac{\sl}{2} \right )+\mathcal{O}(\delta)=\mathrm{Id} + \mathcal{O}(\delta). 
$$

If $|\widehat{\mathcal{T}} |\hat \lambda_0| - k \pi| \leq \kappa \delta$ for some $k \in \mathbb{N}\backslash \{0\}$, by Theorem~\ref{thm:origin_average} the eigenvalues of $\mathbf{X}(\sigma_0)$ are of the form $\pm i |\hat \lambda_0|$, if $\delta$ is small enough. Therefore, if $P$ is such that $P \mathbf{X}(\sigma_0) P^{-1}$ is a diagonal matrix,  
\begin{align*}
\widehat{\Psi}_0 & = 
P \left ( \begin{array}{cc} e^{i |\hat \lambda_0| \widehat{\mathcal{T}}} & 0 \\ 0 & e^{-i |\hat \lambda_0| \widehat{\mathcal{T}}} \end{array}\right ) P^{-1} = 
P \left ( \begin{array}{cc} (-1)^k + \mathcal{O}(\delta) & 0 \\ 0 &  (-1)^k + \mathcal{O}(\delta) 
 \end{array}\right ) P^{-1}  
 \\ & = (-1)^k \mathrm{Id}+ \mathcal{O}(\delta).
\end{align*}
Therefore, if $k$ is such that $\mathbf{L} = |\widehat{\mathcal{T}} |\hat \lambda_0| - k \pi|$, 
$\widehat{\Psi}_0 = (-1)^k \mathrm{Id} + \mathcal{O}(\delta)$. Then, from~\eqref{formula_hat_Psi1} and \eqref{final_variational_delta}, 
\begin{align*}
\widehat{\Psi}_1   & = ((-1)^k \mathrm{Id} + \mathcal{O}(\delta)) \int_0^{\widehat{\mathcal{T}}} ( (-1)^k\mathrm{Id} + \mathcal{O}(\delta))^{-1} \widehat{B}(u) 
((-1)^k \mathrm{Id} + \mathcal{O}(\delta)) \, du \\ & = 
((-1)^k \mathrm{Id} + \mathcal{O}(\delta)) \left ( 
\int_0^{\widehat{\mathcal{T}}} (-1)^k \widehat{B}(u) \, du + \mathcal{O}(\delta ) \right ) = \mathcal{O}(\delta ).
\end{align*}
Therefore $\widehat {\Psi}  = \widehat {\Psi}_0  + \mathcal{O}(\delta^2)$. 
We write
$$
\widehat{\Psi} = \left (
\begin{array}{cc} a & b \\ c & d \end{array}
\right) = \left (
\begin{array}{cc} a_0+\delta^2 a_1 + \mathcal{O}(\delta^3) &  b_0+\delta^2 b_1 + \mathcal{O}(\delta^3) \\  c_0+\delta^2 c_1 + \mathcal{O}(\delta^3) &  d_0+\delta^2 d_1 + \mathcal{O}(\delta^3) \end{array}
\right)
$$
where the elements $a_j,b_j,c_j,d_j$ are the corresponding ones for $\widehat{\Psi}_j $, $j=0,1$. Imposing that  
$\mathrm{det} (\widehat{\Psi} )= \mathrm{det} (\widehat{\Psi}_0 ) =1$, we have  
$$
\delta^2 ( a_0 d_1 + a_1 d_0 - b_0 c_1 - b_1 c_0 )+ \mathcal{O}( \delta^3)=0.
$$
Then, since $a_0,d_0= 1+ \mathcal{O}(\delta)$ (in fact $1+ \mathcal{O}(\delta^2))$ and $b_0,c_0 = \mathcal{O}(\delta)$, we obtain that 
$d_1 + a_1 = \mathcal{O}(\delta )$. As a consequence
$$
\mathrm{tr } (\widehat{\Psi})= \mathrm{tr } (\widehat{\Psi}_0)  + \mathcal{O}(\delta^3) 
$$
and, because we are assuming that $\mathbf{L}\leq \kappa \delta$, the lemma is proven for this case. 

Now we deal with the case $|\widehat{\mathcal{T}} \hat \lambda_0| \geq \mathbf{L} \geq \kappa \delta$ with $\kappa$ a large enough constant. If $\delta$ is small enough, by Theorem~\ref{thm:origin_average}, the eigenvalues of $\mathbf{X}(\sigma_0)$  are of the form $\pm i |\hat \lambda_0|$
with $|\lambda_0|$ in~\eqref{expression_new_eigenvalues} and moreover, again by Remark~\ref{rem:period}, 
$$
\left |\sigma_0 - \frac{\sl}{2} \right |\geq M\kappa \delta
$$
for some constant $M$. As a consequence, recalling that  $\mathbf{a}\left (\frac{\sl}{2}\right ) =0$ and $\partial_{\sigma} \mathbf{a} (\sigma) <0$, for $\sigma>0$, we have that
\begin{equation*}
|\mathbf{a}(\sigma_0)| \geq \left | \mathbf{a} \left ( \frac{\sl}{2} \pm \kappa \delta \right ) \right | \geq M\kappa \delta.
\end{equation*}
In addition, it is clear from the definition of $\mathbf{c}_{\pm}$ in the second item of Lemma~\ref{lem:scaled_eigenvalues} that, $|\mathbf{c}_{\pm} (\sigma_0)|\leq M$. Then, 
$$
\frac{\mathbf{a}(\sigma_0) + \delta \mathbf{c}_\pm (\sigma_0)}{\mathbf{a}(\sigma_0) + \delta \mathbf{c}_\mp (\sigma_0)} 
 \geq 
\frac{1 - \delta \frac{|\mathbf{c}_\pm (\sigma_0)|}{|\mathbf{a}(\sigma_0)}|}
{1 + \delta \frac{|\mathbf{c}_\mp (\sigma_0)|}{|\mathbf{a}(\sigma_0)|}} \geq 1 - \frac{M}{\kappa}, \qquad 
\frac{\mathbf{a}(\sigma_0) + \delta \mathbf{c}_\pm (\sigma_0)}{\mathbf{a}(\sigma_0) + \delta \mathbf{c}_\mp (\sigma_0)} \leq 1 + \frac{M'}{\kappa}
$$
for some constants $M,M'$ and $\kappa$ large enough. 
This trivial fact implies that, $P(\sigma_0;\delta)$, the constant linear change of variables  
that transforms $\mathbf{X}(\sigma_0)$ in its diagonal form satisfies that $\|P(\sigma_0;\delta)\|, \|P^{-1}(\sigma_0;\delta)\| $ are bounded uniformly in $\delta$. 
After this (constant) linear change of variables, 
$\widetilde{\Psi}= P \widehat{\Psi} P^{-1} = \widetilde{\Psi}_0 + \delta \widetilde{\Psi}_1 + \mathcal{O}(\delta^2)$ with $\widetilde{\Psi}_0$ 
$$
\widetilde{\Psi}_0 = \left ( \begin{array}{cc} 
e^{i |\hat \lambda_0| \widehat{\mathcal{T}}} & 0 
\\
0 & e^{-i |\hat \lambda_0| \widehat{\mathcal{T}}}
\end{array} \right ) = 
(-1)^k \left ( \begin{array}{cc} 
e^{i \mathbf{L}} & 0 
\\
0 & e^{-i\mathbf{L}}
\end{array} \right )
$$
and $\widetilde{\Psi}_1=\widetilde{\Psi}_1(\widetilde{\mathcal{T}};\sigma_0)$ with
$$
\widetilde{\Psi}_1(s;\sigma_0) =  \widetilde{\Psi}_0(s;\sigma_0) \int_{0}^s 
(\widetilde{\Psi}_0(u;\sigma_0) )^{-1} P^{-1} (\sigma_0,\delta) \widehat{B} (u) P(\sigma_0,\delta)  \widetilde{\Psi}_0(u;\sigma_0) du.
$$
Using the bound for $\widehat{B}$ in~\eqref{final_variational_delta} and the fact that $P, P^{-1}$ are uniformly bounded with respect to $\delta$, we obtain that
$$
\int_{0}^{\widehat{\mathcal{T}}}  P^{-1}(\sigma_0;\delta)\widehat{B}(u) P(\sigma_0;\delta)\, du = P^{-1}(\sigma_0;\delta) \left [\int_{0}^{\widehat{\mathcal{T}}} \widehat{B}(u)\, du \right ] P(\sigma_0;\delta)= \mathcal{O}(\delta).
$$
Then,  one easily checks that 
\begin{align*}
\widetilde{\Psi}  & = \widetilde{\Psi}_0   + \delta \widetilde{\Psi}_1   + \mathcal{O}(\delta^2)  \\& = (-1)^k \left ( \begin{array}{cc} e^{ i \mathbf{L} } & 0 \\ 0 & e^{-i \mathbf{L}} \end{array} \right )  \left [ \mathrm{Id} + \delta \left (\begin{array}{cc} 0 & {b}_{12} (\hat \Gam_0;\delta) \\ b_{21}(\hat \Gam_0;\delta ) & 0 \end{array} \right ) \right ]  + \mathcal{O}(\delta^2)
\end{align*}
with, writing $P^{-1} \widehat{B}(u) P = \big ( \hat {b}_{ij} (u))_{i,j}$, 
$$
b_{12}(\hat \Gam_0;\delta)= \int_{0}^{\widehat{\mathcal{T}}} \hat{b}_{12} (u) e^{-2 i\mathbf{L} u}\, du,\qquad  
b_{21}(\hat \Gam_0;\delta)=  \int_{0}^{\widehat{\mathcal{T}}} \hat{b}_{12} (u) e^{2 i\mathbf{L} u}\, du.
$$
Therefore, using \eqref{final_variational_delta},
$$
\widetilde{\Psi}  =
\widetilde{\Psi}_0  \left [ \mathrm{Id} +  \mathbf{L}\delta  \left (\begin{array}{cc} 0 & \beta_1(\hat \Gam_0) \\ \beta_2(\hat \Gam_0) & 0\end{array}\right ) + \mathcal{O}(\delta^2 ) \right ]
$$
for some constants $\beta_1,\beta_2$. We emphasize that, when $\mathbf{L} \geq M$, namely is not small when $\delta$ is small, this expression still holds true. We write now
$$
\widetilde{\Psi} = (-1)^k \left ( \begin{array}{cc} 
e^{i \mathbf{L} } + a_2 \delta^2 & \beta_1 \mathbf{L} \delta + b_2 \delta^2 \\ 
\beta_2 \mathbf{L} \delta + c_2 \delta^2 & e^{-i \mathbf{L}}
 + d_2 \delta^2
\end{array}\right )
$$
Therefore $\mathrm{det } (\widetilde{\Psi}) =1$ implies that
$$
1= 1 + \delta^2 (e^{i\mathbf{L}}d_2 +e^{-i\mathbf{L}}a_2 ) - \beta_1 \beta_2 \mathbf{L}^2 \delta^2 + \mathcal{O}(\delta^3 \mathbf{L})+ \mathcal{O}(\delta^4).
$$
Then, since $e^{ \pm i\mathbf{L}} = 1   + \mathcal{O}(\mathbf{L} )$
$$
a_2 + d_2 = \mathcal{O}(\mathbf{L} ) + \mathcal{O}(\delta \mathbf{L}) + \mathcal{O}(\delta^2) = \mathcal{O}(\mathbf{L} ),
$$
where we have used that $\mathbf{L}\geq \kappa \delta$. 
As a consequence,  
$$
\mathrm{tr} (\widehat{\Psi}   )= 
\mathrm{tr} (\widetilde{\Psi} ) = \mathrm{tr} (\widetilde{\Psi}_0   )  + \mathcal{O}(\mathbf{L} \delta^2) = 
 \mathrm{tr} (\widehat{\Psi}_0 )  + \mathcal{O}(\mathbf{L}  \delta^2)
$$
and the lemma is proven. 
\end{proof}

We denote by $\mu=e^{  \hat \lambda \widehat{\mathcal{T}}}$ one eigenvalue of $\widehat{\Psi}$. Denoting $\tau= \mathrm{tr } (\widehat{\Psi})$ and using that  $\mathrm{det} (\widehat{\Psi})=1$, we can assume that
$$
\mu = \frac{1}{2} \tau +\frac{1}{2} \sqrt{\tau^2  -4  }.
$$
Let $\mu_0 = e^{\hat \lambda_0 \widehat{\mathcal{T}}}$ be such that 
$$
\mu_0 = \frac{1}{2} \tau_0 +\frac{1}{2} \sqrt{\tau_0^2  -4  }.
$$
with $\tau_0= \mathrm{tr }(\widehat{\Psi}_0)$. 
Then 
$$
 2| \mu - \mu_0| \leq  
  |\tau -\tau_0| +  \left |\sqrt{\tau^2 - 4 } - \sqrt{\tau_0^2 -4} \right |.
$$
By Lemma~\ref{lem:trace_fondamental_matrix}, $|\tau -\tau_0| \leq M \delta^2$, therefore, to prove Theorem~\ref{prop:eigenvalues_monodromy} we have to analyze 
$$
\Delta \mu :=  \left | \sqrt{\tau^2 - 4 } - \sqrt{\tau_0^2 -4} \right |  .
$$
\begin{lemma}\label{lem:mathbfL}
Take $C_1,C_2>0$ two constants. There exist $\delta_0>0$ and a  constant $C_*>0$ such that, for $\delta \in [0,\delta_0]$, we have that 
\begin{itemize}
\item When $\mathbf{L}\geq C_1\delta$, then 
$ 
\Delta \mu \leq C_*\delta^2.
$ 
\item If $C_2 \delta^{3/2} \leq \mathbf{L} \leq  C_1\delta $, then 
$
\Delta \mu \leq C_* \delta^3 \mathbf{L}^{-1} \leq C_*\delta^{3/2} C_2^{-1}.
$
\item Otherwise, that is if  $\mathbf{L} \leq C_2\delta^{3/2}$, then $\Delta \mu \leq C_* \delta^{3/2}$. 
\end{itemize}
As a consequence $\Delta \mu \leq M\delta^{3/2}$ for some positive constant $M$.  
\end{lemma}
\begin{proof}
  We recall that, by Remark~\ref{rem:mathbfL}, $\mathbf{L}\leq \frac{\pi}{2}$. 
Assume first that $\mathbf{L} \geq \kappa \delta$ for some $\kappa$ big enough. Then by Theorem~\ref{thm:origin_average}, the eigenvalues of $\mathbf{X}(\hat \Gam_0)$ in~\eqref{expression_variational_periodic} are of the form $\pm i |\hat \lambda_0| $ with $|\hat \lambda_0|$ in~\eqref{expression_new_eigenvalues}. Then, 
$\tau_0 = e^{i |\hat \lambda_0 |\widehat{\mathcal{T}}} + e^{- i |\hat \lambda_0| \widehat{\mathcal{T}}} $ satisfies
$$
|\tau_0^2- 4 | =  4- 4 \cos^2 \big (\widehat{\mathcal{T}} \hat \lambda_0  \big )   = 4   \sin^2 
\big (\widehat{\mathcal{T}} \hat \lambda_0    \big ) = 4 \sin^2 \mathbf{L} .
$$
In view of Lemma~\ref{lem:trace_fondamental_matrix}, $\tau= \tau_0 + \mathbf{L} \delta^2 \tau_1$ with $|\tau_1|\leq M$. Then 
\begin{align*}
    \Delta \mu &= \big |\sqrt{\tau_0^2 - 4 } \big | 
    \left | 1 - \sqrt{1 + \frac{2 \delta^2 \mathbf{L} \tau_0 \tau_1 + \delta^4 \mathbf{L}^2 \tau_1^2} {\tau_0^2 -4 }} \right |  \leq  \frac{2 \delta^2 \mathbf{L} |\tau_0| |\tau_1|}{\sqrt{\tau_0^2 -4}} + M \delta^2 |\tau_1|  \\& \leq  M   \delta^2 |\tau_1|.
\end{align*}
In the last bound we have used that $\sqrt{\tau_0^2- 4} = 2 \sin \mathbf{L} \geq M \mathbf{L}$, provided $0 \leq \mathbf{L} \leq \frac{\pi}{2}$.

Now assume that $\mathbf{L} \leq \kappa \delta$.  
By Lemma~\ref{lem:trace_fondamental_matrix}, $\tau = \tau_0 + \delta^3 \tau_1$ with $|\tau_1|\leq M$.  
We notice that, using that either $\tau_0 = \pm (e^{\mathbf{L}} + e^{-\mathbf{L}})$ or $\tau_0 = \pm (e^{i\mathbf{L}} + e^{-i\mathbf{L}})$. Then 
$$
|\tau_0^2 - 4 |= 2 \mathbf{L}^2 + \mathcal{O}(\mathbf{L}^4).
$$
When $\mathbf{L}^2  \geq  M\delta^3 $,  for some suitable constant 
\begin{align*}
\Delta \mu & =  |\sqrt{\tau_0^2 -4 }|
 \left | 1- \sqrt{1 + \frac{2 \delta^3 \tau_0 \tau_1 + \delta^6 \tau_1^2 }{\tau_0^2 -4 }}   \right |  \leq M \delta^3 \frac{|\tau_1|}{ \sqrt{| \tau_0^2 -4|}} \\ 
 &\leq M \delta^3 \mathbf{L}^{-1} |\tau_1|.
\end{align*}
On the other hand, when $\mathbf{L}^2 \leq M \delta^3 $, clearly,
$$
\Delta \mu \leq 2 \sqrt{|\tau_0^2 -4| } + \delta^{3/2} \sqrt{2 |\tau_0| |\tau_1|} + \delta^3 |\tau_1|  \leq M \delta^{3/2} |\tau_1|^{1/2}. 
$$
\end{proof}
 
 Theorem~\ref{prop:eigenvalues_monodromy} is a straightforward consequence of Lemma~\ref{lem:mathbfL} just taking into account that, using definitions~\eqref{formula_period_theorem}, \eqref{def:muCP0} and~\eqref{def:GammaAV} of $\mathcal{T}$, $\mathcal{T}^{(0)}$ and $\Gam_\CP^{(0)}$ we have that, by Taylor's theorem,
$$
\widehat{\mathcal{T}}^{(0)} (\hat \Gam_\CP^{(0)} (\hat \Gam_0;\delta);\delta) = 
\frac{4L^7}{3\rho_0} {\mathcal{T}}^{(0)} (\Gam_\CP^{(0)} (\Gam_0;L);L) =
\widehat{\mathcal{T}} ( \hat \Gam_0;\delta) + \mathcal{O}(\delta^2).
$$
Therefore, 
$$
e^{i |\hat \lambda_0| \widehat{\mathcal{T}}^{(0)}} = e^{i |\hat \lambda_0| \widehat{\mathcal{T}}}+  \mathcal{O}(\delta^2).
$$
 
\subsubsection{Proof of Lemma~\ref{lem:Gamkpm}}\label{subsec:lastlema}

To prove this result, by expressions~\eqref{expression_new_eigenvalues} and~\eqref{def:muCP0} of $|\hat \lambda_0|$ and $\mu_\CP^{(0)}$ respectively, we need to solve the equation
$$
\frac{2\pi}{\big| \hat a_0 (\sigma;\delta) \big |} \frac{1}{4}\sqrt{ \big |\mathbf{a} (\sigma) + \delta \mathbf{c}_+(\sigma) \big |
\big |\mathbf{a} (\sigma) + \delta \mathbf{c}_-(\sigma) \big |} = k\pi , 
$$
with $\hat a_0$ defined in~\eqref{defhata0c0}, 
for $k=1,2$.  Using definition~\eqref{def:mathbfc_theorem} of $\mathbf{c}_\pm$, we rewrite it as  
$$
H_k(\sigma;\delta)= \mathbf{a}^2 (\sigma) + \delta \mathbf{a}^2 (\sigma) U_2^{0,0}
+ \delta^2 \mathbf{c}_+(\sigma) \mathbf{c}_-(\sigma)  - 4 k^2 \hat a_0^2(\sigma;\delta).
$$
Let $\hat \Gam_k^{(0)}$ be the values such that this equation holds true for $\delta=0$ (see Table~\ref{table_sigma} for the exact values of $\Gam_k^{(0)}$). 
One can easily check that $\partial_\sigma H_k (\hat \Gam_k^{(0)};\delta)\neq 0$ 
 and therefore, by the implicit function theorem, there exists $\hat \Gam_k(\delta)$ such that $H_k(\hat \Gam_k(\delta);\delta)=0$ and the lemma holds true.  In addition a simply computation proves that
$$
\partial_\delta \hat \Gam_k(0)= - \frac{\partial_\delta H_k (\hat \Gam_k^{(0)};0)}{\partial_\sigma H_k(\hat \Gam_k^{(0)};0)}=0
$$
and therefore $\hat \Gam_k(\delta)= \hat \Gam^{(0)}_k + \mathcal{O}(\delta^2)$. 
 
\subsubsection{Proof of Theorem~\ref{thm:eigenvalues_monodromy_more}} \label{subsec:lasttheorem}  

The proof of Theorem~\ref{thm:eigenvalues_monodromy_more} is a nontrivial consequence of  Lemma~\ref{lem:mathbfL}. Fix $\delta_0>0$ small enough (to be determined later) and consider $\delta \in [0,\delta_0]$. 
By Theorem~\ref{thm:origin_average} and Lemma~\ref{lem:Gamkpm}, there exist $\sigma_{\pm}$ and $\sigma_k$, $k=1,2$ such that  
$$
\hat \lambda_0(\sigma_\pm)=0,\qquad 
\widehat{\mathcal{T}}^{(0)}(\sigma_k;\delta) |\hat \lambda_0(\sigma_k) | = \frac{2\pi |\hat \lambda_0(\sigma_k) |}{|\hat a_0(\sigma_k;\delta)| } = k\pi
,\; k=1,2 
$$
with $\hat a_0$ defined in~\eqref{defhata0c0}. Notice also that 
$$
\sigma_\pm = \frac{\sl}{2} \pm  \delta C_0 + \mathcal{O}(\delta^2),\qquad \sigma_k = \hat \Gam_k^{(0)} + C_k \delta + \mathcal{O}(\delta^2).
$$
Therefore, since $1\leq \nu\leq 2$, we can assume that the parameter $\sigma$ satisfies, taking $\delta$ small enough, 
$$
|\sigma - \sigma_\pm | \geq C \delta^{\nu}, \qquad |\sigma - \sigma_k |\geq C \delta^{\frac{1+\nu}{2}}.
$$
 
We now observe that, using the definition of $ \mathbf{c}_\pm$ in~\eqref{def:mathbfc_theorem} 
\begin{equation}\label{expr_lambda0_proof}
16 |\hat \lambda_0(\sigma)|^2 = \mathbf{a}^2 (\sigma) (1+ \delta U_2^{0,0} ) + \mathcal{O}(\delta^2)
\end{equation}
with $\mathbf{a}= 1 - 16 \sigma - 20 \sigma^2$. In addition 
\begin{equation}\label{defT0hat_proof}
    \widehat{\mathcal{T}}^{(0)}(\sigma;\delta) := \frac{2\pi}{|\hat a_0(\sigma;\delta)|} = \widehat{\mathcal{T}}(\sigma;\delta) + \mathcal{O}(\delta^2).
\end{equation}

We distinguish the cases $|\mathbf{a}(\sigma)|\geq \delta |\log \delta|$ and $|\mathbf{a}(\sigma)|< \delta |\log \delta|$. We recall that $\mathbf{a} \left (\frac{\sl}{2}\right ) =0$. Then, the cases considered are equivalent to $\left |\sigma - \frac{\sl}{2}\right | \geq \kappa \delta |\log \delta|$ and $\left |\sigma - \frac{\sl}{2}\right | <\kappa \delta |\log \delta|$ for some constant $\kappa$.

We start with $\left |\sigma - \frac{\sl}{2}\right |\geq \kappa \delta |\log \delta|$. From~\eqref{expr_lambda0_proof} we have that  
\begin{equation}\label{expr_lambda0_proof_bis}
4 |\hat \lambda_0(\sigma)|  = |\mathbf{a} (\sigma)|  \big (1 + \mathcal{O}(|\log \delta|^{-1})  \big ). 
\end{equation}
\begin{itemize}
    \item If $\sigma \in \left (0, \frac{\sl}{2} - \kappa \delta |\log \delta| \right )$, by Proposition~\ref{prop:origin_modulus_eigenvalue}, the function $| \hat \lambda_0|$ is decreasing. Therefore, taking into account~\eqref{defT0hat_proof} and definition~\ref{defhata0c0} of $\hat a_0$, if $\delta$ is small enough, 
    $$
    \frac{1}{2} \left |\hat \lambda_0 \left (\frac{\sl}{2} - \delta |\log \delta| \right ) \right |\widehat{\mathcal{T}}\left (\frac{\sl}{2} - \delta |\log \delta| ;\delta \right )\leq  |\hat \lambda_0 (\sigma) |\widehat{\mathcal{T}}(\sigma;\delta)  \leq \frac{3}{2} |\hat \lambda_0(0)| \widehat{\mathcal{T}}(0;\delta) 
    $$
    Then $|\hat \lambda_0 (\sigma) |\widehat{\mathcal{T}}(\sigma;\delta)| \leq \frac{3\pi}{4} (1+ M\delta)$ and $|\hat \lambda_0 (\sigma) |\widehat{\mathcal{T}}(\sigma;\delta) \geq M \delta |\log \delta|$. Those bounds imply that, for all $k$ 
    $$
    \big | |\hat \lambda_0 (\sigma) |\widehat{\mathcal{T}}(\sigma;\delta) - k \pi \big | \geq M \delta |\log \delta| 
    $$
    and, in conclusion, 
    $$
    \mathbf{L} \geq M \delta |\log \delta|.
    $$
    By Lemma~\ref{lem:mathbfL}, $|\mu - \mu_0| \leq M \delta^2$ and the result follows for this case.  
    \item If $\sigma > \frac{\sl}{2} + \delta |\log \delta|$, we notice that, using~\eqref{expr_lambda0_proof_bis} and~\eqref{defT0hat_proof},
    \begin{align*}
\widehat{\mathcal{T}}(\sigma;\delta) |\hat \lambda_0(\sigma)| -k\pi = & \widehat{\mathcal{T}}(\sigma;\delta) |\hat \lambda_0 (\sigma)| -\widehat{\mathcal{T}}(\sigma_k;\delta) |\hat \lambda_0(\sigma_k)| \\ = & \left [\partial_\sigma \left (
\frac{1}{4}\frac{|\mathbf{a}(s)|}{ |\hat a_0 (s;0) | } \right ) + \mathcal{O}(|\log \delta|^{-1}) \right ] |\sigma - \sigma_k| \\ 
\widehat{\mathcal{T}}(\sigma;\delta) |\hat \lambda_0(\sigma)| = & \left [\partial_\sigma \left (
\frac{1}{4}\frac{|\mathbf{a}(t)|}{ |\hat a_0 (t;0) | } \right ) + \mathcal{O}(|\log \delta|^{-1}) \right ] |\sigma - \sigma_+| ,
\end{align*}
with $s\in \overline{\sigma,\sigma_k}$, $t\in \overline{\sigma,\sigma_\pm}$ and $\hat a_0$ defined in~\eqref{defhata0c0}. 
Since $\sigma, \sigma_k,\sigma_+ >\frac{\sl}{2}$ (see Lemma~\ref{lem:Gamkpm}), 
$$
\partial_\sigma 
\left (
\frac{1}{4}\frac{|\mathbf{a}(s)|}{ |\hat a_0 (s;0) | } \right ) = \frac{1}{4|\hat a_0(s;0)|^2 } (40s^2 + 40 s +18) \geq M.
$$
Therefore $\mathbf{L}\geq M \min \left \{\delta |\log \delta | , \delta^{\frac{1+\nu}{2}} \right \} = M\delta^{\frac{1+\nu}{2}} \geq M \delta^{\frac{3}{2}}$. Therefore, Lemma~\ref{lem:mathbfL} assures that
$$
|\mu - \mu_0| \leq M \delta^{3} \mathbf{L}^{-1} \leq M \delta^{\frac{5-\nu}{2}}.  
$$
\end{itemize}

Now we deal with the case $\left |\sigma - \frac{\sl}{2} \right | \leq \kappa \delta |\log \delta|$, which in particular implies that $|\hat \lambda_0(\sigma)| \leq M \delta |\log \delta|$. Then, by continuity, in this case $\mathbf{L} = \widehat{\mathcal{T}} |\hat \lambda_0(\sigma)|$, if $\delta$ is small enough. We first observe that, $\sigma_\pm= \Gam_\pm L^{-1}$ with $\Gam_\pm$ defined in Theorem~\ref{thm:origin_average} and $\sigma_\pm = \frac{\sl}{2} \pm C_0 \delta + \mathcal{O}(\delta^2)$ for some constant $C_0$ (see~\eqref{def:Cm}). Therefore $|\sigma - \sigma_\pm |\leq \kappa \delta |\log \delta|$ for a (different) constant $\kappa$. 

Assume now that  $C\delta^\nu \leq |\sigma- \sigma_-| \leq |\sigma-\sigma_+| \leq \kappa \delta |\log \delta|$. 
We have that (see Theorem~\ref{thm:origin_average})
$$
\mathbf{a}(\sigma_-) + \delta \mathbf{c}_-(\sigma_-) =  0 .
$$
Taking $\Delta \sigma = \sigma  - \sigma_- = \mathcal{O}(\delta |\log \delta|)$, we have that  
\begin{align*}
|\hat \lambda_0 (\sigma)|^2 = & \left | \partial_\sigma \mathbf{a}(\sigma_-) \Delta \sigma +  \Delta \sigma \mathcal{O} (\delta)   \right | \left | \partial_\sigma \mathbf{a}(\sigma_-) \Delta \sigma+ \delta \big ( \mathbf{c}_+ (\sigma_-) - \mathbf{c}_-(\sigma_-)\big ) + \Delta \sigma \mathcal{O} (\delta) \right |\\ 
=& \left |\left [ \partial_\sigma \mathbf{a} (\sigma_-) \right ]^2 (\Delta \sigma )^2 + \partial_\sigma \mathbf{a}(\sigma_-) \big ( \mathbf{c}_+ (\sigma_-) - \mathbf{c}_-(\sigma_-)\big ) \Delta \sigma \delta +  \Delta \sigma \mathcal{O}(\delta^2|\log \delta|^2) \right |.
\end{align*}
We note that $|\partial_\sigma \mathbf{a} (\sigma)| = 16 + 40 \sigma > 16$ and, therefore, we can write $|\hat \lambda_0 (\sigma)|^2$ as 
\[
|\hat \lambda_0 (\sigma)|^2=  \left [ \partial_\sigma \mathbf{a} (\sigma_-) \right ]^2  \Lambda
\]
where  $\Lambda$,   using  formula~\eqref{def:mathbfc_theorem} for $\mathbf{c}_\pm$, formula~\eqref{def:Cm} of $C_0$ and that $\sigma_- = \frac{\sl}{2} + \mathcal{O}(\delta)$, is given by
$$
\Lambda =   \left |(\Delta \sigma)^2 - 2 \delta \Delta \sigma  C_0 + \Delta \sigma \mathcal{O}(\delta^2|\log \delta|^2) 
\right |=  |\Delta \sigma| \left | \Delta \sigma  - 2 \delta C_0 + \mathcal{O}(\delta^2 |\log \delta|^2) \right |.
$$
Since $|\sigma - \sigma_-|\leq |\sigma - \sigma_+ |$ and 
$\sigma_\pm = \frac{\sl}{2} \pm C_0 \delta + \mathcal{O}(\delta^2|\log \delta|^2)$, $\Delta \sigma < 2 \delta C_0$ and then, using that $|\Delta \sigma|\geq C \delta^\nu$ with $C<C_0$, for $\delta $ small enough 
$$
\Lambda = |\Delta \sigma| (2\delta C_0 - \Delta \sigma + + \mathcal{O}(\delta^2|\log \delta|^2) )\geq C\delta ^\nu (2\delta C_0 - C \delta^\nu + \mathcal{O}(\delta^2|\log \delta|^2) )\geq  M \delta^{1+\nu} 
$$
for some constant $M>0$. Therefore, since $1+\nu \leq 3$, $\mathbf{L} = \widehat{\mathcal{T}} |\hat \lambda_0| \geq M \delta^{\frac{1+\nu}{2}} \geq M \delta^{3/2}$ and by Lemma~\ref{lem:mathbfL}, 
\begin{equation}\label{boundmu_proof}
|\mu - \mu_0| \leq C_* \delta^{\frac{5- \nu}{2}}. 
\end{equation}
An analogous analysis can be done if $C\delta^\nu \leq |\sigma - \sigma_+|\leq |\sigma - \sigma_-| \leq \kappa \delta |\log \delta | $ and, as a consequence, \eqref{boundmu_proof} holds true for
$C\delta^\nu \leq |\sigma-\sigma_\pm|\leq \kappa \delta |\log \delta| $.

\appendix 
\section{Expression of the Hamiltonian}\label{app:tables}

This appendix is devoted to obtain the expressions for the coplanar and averaged Hamiltonians introduced in Section~\ref{section:hierarchy}.
		
\paragraph{Slow-fast Delaunay coordinates $(y,x)$.}
\label{appendix:HamiltonianDelaunay}
	
Analogously to the definition of the coplanar Hamiltonian $\HH_{\CP,1}$ in Poincar\'e coordinates (see~\eqref{def:Hcoplanar:perturb}), we define
\[
	\rH_{\CP, 1}(y,\Gam,x,h) = \rH_1(y,\Gam,x,h,\OmegaM;0).
\]
Recall that, by Remark~\ref{rem:auton}, $\rH_{1}$ is independent of $\OmegaM$ when $\iM=0$.
Then, by~\eqref{def:hamiltonianDelaunayH1} and~\eqref{def:H1tilde}, one has that
\begin{align*}
	\rH_{\CP,1}(y,\Gam, x, h)
	&=	\frac{\rho_1}{L^2}					
	\sum_{m=0}^2
	\sum_{p=0}^2 
	\hat{c}_m U_2^{m,0}
	D_{m,p}(y,\Gam)
	\cos\Big(\psi_{m, p, 0}(x,h)	\Big),
\end{align*}
where $D_{m,p} = \Dt_{m,p} \circ \UDelaunay$,
 $\psi_{m,p,0}=\psit_{m,p,0} \circ \UDelaunay$, with $\Dt_{m,p}$ and $\psit_{m.p,0}$ defined in~\eqref{def:Dtilde} and~\eqref{def:psitilde} 
 respectively, and
\begin{align*}
	\hat{c}_0 = c_{0,s} = \frac12, \qquad
	\hat{c}_1 = c_{1,s} = \frac13, \qquad
	\hat{c}_2 = c_{2,s} = -\frac1{12}.
\end{align*}
Notice that these last definitions are possible since, for $s,m\in\{0,1,2\}$, the constants $c_{m,s}$ as defined in~\eqref{def:auxiliaryFunctionsOriginal} do not depend on $s$.
In addition, $U_2^{m, 0}(\epsilon)$ is the Giacaglia function given in Table~\ref{tab:U}.
\begin{table}[!t]
	\begin{center}
		\begin{tabular}{lcc}
			\hline	
			$m$ 		&	$U_2^{m,0}(\epsilon)$ 									&	$\simeq$\\
			\hline 	
			0 		 		& 	$1-6 C^{2}+6 C^{4}$ 							&	0.762646\\
			1 			&  $-3 C S^{-1}\left(2 C^{4}-3 C^{2}+1\right)$	 	&	0.547442\\
			2 			&  $6 C^{2} S^{-2}\left(C^{2}-1\right)^{2}$	 		&	0.237353 \\
			\hline 
			\phantom{a} & & 
		\end{tabular}
		\caption{The Giacaglia function $U_2^{m, 0}(\epsilon)$ for the moon perturbation (see~\cite{G1974}) where $C=\cos\frac{\epsilon}2$ and $S=\sin\frac{\epsilon}2$ and its value for $\epsilon = 23.44$\degre.}
		\label{tab:U}
	\end{center}
\end{table}

Applying the slow-fast change of coordinates, one obtains that
\begin{align*}
	\psi_{m,p, 0}(x,h) &= (1 -p)x -(1 -p -m) h.
\end{align*}
See Table~\ref{tab:functionsDelaunay} for the explicit expressions of $\psi_{m, p, 0}(x,h)$ and $D_{m,p}(y,\Gam)$.

	\begin{table}[h]
		\begin{tabular}{cclc}
			\hline	
			$m$ 	&	$p$		& \multicolumn{1}{c}{$D_{m,p}(y,\Gam)$}	& $\psi_{m,p,0}$	\\
			\hline 	
			0 		& 	0 		&	
			$-\frac{15}{64}(Ly)^{-2}(y - \Gam)(3y + \Gam)\left(L - 2 y\right)\left(L + 2 y\right)$
			& $x-h$
			\\
			0 		& 	1 		&
			$\phantom{-}\frac{1}{32}(Ly)^{-2}( y^2 - 6y\Gam- 3\Gam^2)\left(5L^2	-	12 y^2\right)$
			& $0$
			\\
			0 		& 	2 		& 	
			$-\frac{15}{64}(Ly)^{-2}(y - \Gam)(3y + \Gam)\left(L - 2 y\right)\left(L + 2 y\right)$
			& $-x-h$
			\\
			1 		& 	0 		& 
			$\phantom{-}\frac{15}{32}(Ly)^{-2}\sqrt{(y - \Gam)(3y + \Gam)}\left(3y + \Gam\right)\left(L - 2 y\right)\left(L + 2 y\right)$	
			& $x$	 
			\\
			1 		& 	1 		& 
			$-\frac{3}{16}(Ly)^{-2}\sqrt{(y - \Gam)(3y + \Gam)}(\Gam + y)\left(5L^2	-	12y^2\right)$
			& $h$		 
			\\
			1 		& 	2 		& 	
			$-\frac{15}{32}(Ly)^{-2}\sqrt{(y - \Gam)(3y + \Gam)}\left(y-\Gam\right)\left(L - 2 y\right)\left(L + 2 y\right)$  
			& $-x+2h$
			\\
			2 		& 	0 		& 
			$\phantom{-}\frac{15}{32}(Ly)^{-2}\left(3y +\Gam\right)^{2}\left(L - 2 y\right)\left(L + 2 y\right)$
			& $x+h$
			\\
			2 		& 	1 		& 	
			$\phantom{-}\frac{3}{16}(Ly)^{-2}(y - \Gam)(3y + \Gam)\left(5L^2	-	12y^2\right)$
			& $2h$
			\\
			2 		& 	2 		& 	
			$\phantom{-}\frac{15}{32}(Ly)^{-2}\left(y-\Gam\right)^{2}\left(L - 2 y\right)\left(L + 2 y\right)$
			& $-x+3h$
			\\
			\hline
		\end{tabular}
		\caption{Computation of the functions $(D_{m,p})_{m,p\in\{0,1,2\}}$ and  $(\psi_{m,p, 0})_{m,p\in\{0,1,2\}}$ for the prograde case.}
		\label{tab:functionsDelaunay} 
	\end{table}
	
	Now, we consider the $h$-average of $\rH_{\CP,1}$,
	\begin{equation*}
		\rH_{\AV,1}(y,\Gam, x)	=	\frac{1}{2\pi} \int_0^{2\pi}\rH_{\CP, 1}	(y,\Gam, x, h)\rd h
	\end{equation*}	
	The definition of $\psi_{m,p,0}(x,h)$ implies that the only terms independent of $h$ are for the couples $(m,p)=(0,1)$ and $(m,p) = (1,0)$ (see Table~\ref{tab:functionsDelaunay}).
	As a consequence,  we obtain the expression of the averaged Hamiltonian presented in \eqref{eq:expansionsH1averagedDelaunay}.
	
	\paragraph{Poincar\'e coordinates $(\eta,\xi)$.}
	Taking into account that the Poincar\'e change of coordinates satisfies that
	\begin{align*}
		y = \frac{2L-\xi^2-\eta^2}4,
		\qquad
		(L-2y)\cos x = \frac{\xi^2-\eta^2}2,
		\qquad
		(L-2y)\sin x= \xi \eta,
	\end{align*}
	one has that
	\begin{equation}\label{eq:expressionPoincareHCP1}
		\begin{aligned}
			\HH_{\CP, 1}
			=& \,	\frac{\rho_1}{L^2}					
			\sum_{m=0}^2 \hat{c}_m U_2^{m,0} \cdot \\
			&\Big[
			\DDD_{m,0}(\eta,\Gam,\xi) \left(
			\frac{\xi^2-\eta^2}2 \cos\big((1-m)h\big)
			+
			\xi\eta \sin\big((1-m)h\big) \right)
			\\
			&+ 
			\DDD_{m,1}(\eta,\Gam,\xi)
			\left(8L^2+12 L (\xi^2+\eta^2)-3(\xi^2+\eta^2)^2\right)
			\cos(mh)
			\\
			&+ \DDD_{m,2}(\eta,\Gam,\xi) \left(
			\frac{\xi^2-\eta^2}2 \cos\big((1+m)h\big)
			+
			\xi\eta \sin\big((1+m)h\big) \right)
			\Big],
		\end{aligned}
	\end{equation}
	where the functions $(\DDD_{m,p})_{m,p\in\{0,1,2\}}$ are given in Table~\ref{tab:functionsPoincare}.
	%
	%
	
	\begin{table}[h]
		\begin{equation*}
			\def\arraystretch{1.3}
			\begin{array}{ccl}
				\hline
				m & p & \multicolumn{1}{c}{\DDD_{m,p}(\eta,\Gam,\xi)} \\
				\hline 
				0 		& 	0 		&	
				-\frac{15}{128} {L^{-2} (2L-M)^{-2}} 
				{(2L-M-4\Gam)(6L-3M+4\Gam)(4L-M)}
				\\
				\hline
				0 		& 	1 		&
				\phantom{-}\frac1{128}
				{L^{-2} (2L-M)^{-2}}((2L-M)^2-24(2L-M)\Gam -48\Gam^2)
				\\
				\hline
				0 		& 	2 		& 	
				-\frac{15}{128}
				{L^{-2} (2L-M)^{-2}}
				{(2L-M-4\Gam)(6L-3M+4\Gam)(4L-M) }
				\\
				\hline
				1 & 0 &
				\phantom{-}\frac{15}{64}
				{L^{-2} (2L-M)^{-2}}
				{\sqrt{(2L-M-4\Gam)}(6L-3M+4\Gam)^{3/2}(4L-M)}
				\\
				\hline
				1 & 1 &
				-\frac{3}{64}
				{L^{-2} (2L-M)^{-2}}
				{\sqrt{(2L-M-4\Gam)(6L-3M+4\Gam)}(2L-M+4\Gam)} 
				\\
				\hline
				1 & 2 &
				-\frac{15}{64}
				{L^{-2} (2L-M)^{-2}}
				{(2L-M-4\Gam)^{3/2}\sqrt{6L-3M+4\Gam}(4L-M)}
				\\
				\hline
				2 & 0 &
				\phantom{-}\frac{15}{64}
				{L^{-2} (2L-M)^{-2}}
				{(6L-3M+4\Gam)^{2}(4L-M)}
				\\
				\hline
				2 & 1 &
				\phantom{-}\frac{3}{64}
				{L^{-2} (2L-M)^{-2}}
				{(2L-M-4\Gam)(6L-3M+4\Gam)}
				\\
				\hline
				2 & 2 &
				\phantom{-}\frac{15}{64}
				{L^{-2} (2L-M)^{-2}}
				{(2L-M-4\Gam)^{2}(4L-M) }
				\\
				\hline
			\end{array}
		\end{equation*}
		\caption{Computation of the functions $(\DDD_{m,p})_{m,p\in\{0,1,2\}}$ with $M:=\xi^2+\eta^2$.}
		\label{tab:functionsPoincare}
	\end{table}

Let us now consider the $h$-averaged Hamiltonian $\HH_{\AV,1}$ as defined in \eqref{def:HamAV}.
Analogously to $\rH_{\AV,1}$, one sees that the only terms independent of $h$ are for the couples $(m,p)=(0,1)$ and $(1,0)$. 
As a consequence, one obtains the expression in~\eqref{def:Ham:AV1}.
\bibliographystyle{abbrv}
\bibliography{references_av} 

\begin{thebibliography}{10}

\bibitem{AlessiBGGP23}
E.~M. Alessi, I.~Baldom\'a, M.~Giralt, and M.~Guardia.
\newblock On the {A}rnold diffusion mechanism in {M}edium {E}arth {O}rbit.
\newblock Preprint, available at \url{https://arxiv.org/abs/2403.01283}, 2024.

\bibitem{Alessi2014}
E.~M. Alessi, A.~Rossi, G.~B. Valsecchi, L.~Anselmo, C.~Pardini, C.~Colombo,
  H.~G. Lewis, J.~Daquin, F.~Deleflie, M.~Vasile, F.~Zuiani, and K.~Merz.
\newblock Effectiveness of {GNSS} disposal strategies.
\newblock {\em Acta Astronautica}, 99:292--302, 2014.

\bibitem{sB01}
S.~Breiter.
\newblock Lunisolar resonances revisited.
\newblock {\em Celest. Mec. Dyn. Astron.}, 81:81--91, 2001.

\bibitem{Chao2005}
C.~Chao.
\newblock {\em Applied Orbit Perturbation and Maintenance}.
\newblock American Institute of Aeronautics and Astronautics/Aerospace Press,
  Chao, C., `Applied Orbit Perturbation and Maintenance', American Institute of
  Aeronautics and Astronautics/Aerospace Press, Reston, Virginia/El Segundo,
  California, 2005.

\bibitem{gC62}
G.~Cook.
\newblock Luni-solar perturbations of the orbit of an earth satellite.
\newblock {\em J. Roy. Astron. Soc.}, 6:271--291, 1962.

\bibitem{Daquin2022}
J.~Daquin, E.~Legnaro, I.~Gkolias, and C.~Efthymiopoulos.
\newblock A deep dive into the $2g+h$ resonance: separatrices, manifolds and
  phase space structure of navigation satellites.
\newblock {\em Celestial Mechanics and Dynamical Astronomy}, 134:6, 2022.

\bibitem{Daquin2016}
J.~Daquin, A.~J. Rosengren, E.~M. Alessi, F.~Deleflie, G.~B. Valsecchi, and
  A.~Rossi.
\newblock The dynamical structure of the {MEO} region: long-term evolution
  stability, chaos, and transport.
\newblock {\em Celestial Mechanics and Dynamical Astronomy}, 124:335--366,
  2016.

\bibitem{G1974}
G.~E.~O. Giacaglia.
\newblock Lunar perturbations of artificial satellites of the earth.
\newblock {\em Celestial Mechanics and Dynamical Astronomy}, 9:239--267, 1974.

\bibitem{Gkolias2019}
I.~Gkolias, J.~Daquin, D.~K. Skoulidou, K.~Tsiganis, and C.~Efthymiopoulos.
\newblock Chaotic transport of navigation satellites.
\newblock {\em Chaos}, 29:101106, 2019.

\bibitem{Hughes2000}
S.~Hughes.
\newblock Earth satellite orbits with resonant lunisolar perturbations. i.
  resonances dependent only on inclination.
\newblock {\em Proceedings of the Royal Society of London. Series A,
  Mathematical and Physical Sciences}, 372:243--264, 1980.

\bibitem{Hughes1980}
S.~Hughes.
\newblock Earth satellite orbits with resonant lunisolar perturbations i.
  resonances dependent only on inclination.
\newblock {\em Proceedings of the Royal Society of London. A. Mathematical and
  Physical Sciences}, 372(1749):243--264, 1980.

\bibitem{Jenkin2005}
A.~B. Jenkin and R.~A. Gick.
\newblock Dilution of disposal orbit collision for the {M}edium {E}arth {O}rbit
  constellation.
\newblock In {\em Proceedings of the 4th {E}uropean {C}onference on {S}pace
  {D}ebris}, {E}{S}{A} {S}{P}-587, pages 309--314. ESA/ESOC, 2005.

\bibitem{Kaula66}
W.~Kaula.
\newblock {\em Theory of {Satellite} {Geodesy}: {Applications} of {Satellites}
  to {Geodesy}}.
\newblock Blaisdell Publi. CO., 1966.

\bibitem{Kaula1962}
W.~M. Kaula.
\newblock Development of the lunar and solar disturbing functions for a close
  satellite.
\newblock {\em The Astronomical Journal}, 67:300--303, 1962.

\bibitem{Legnaro2023}
E.~Legnaro and C.~Efthymiopoulos.
\newblock A detailed dynamical model for inclination-only dependent lunisolar
  resonances. {E}ffect on the ``eccentricity growth'' mechanism.
\newblock {\em Advances in Space Research}, 72:2460--2480, 2023.

\bibitem{Legnaro2023b}
E.~Legnaro, C.~Efthymiopoulos, and M.~Harsoula.
\newblock Semi-analytical estimates for the chaotic diffusion in the {Second}
  {Fundamental} {Model} of {Resonance}. {Application} to {Earth's} navigation
  satellites.
\newblock {\em Physica D}, 456:133946, 2023.

\bibitem{Musen1961}
P.~Musen.
\newblock On the long-period lunar and solar effects on the motion of an
  artificial satellite, 2.
\newblock {\em J. Geophys. Res.}, 66:2797--2805, 1961.

\bibitem{Pellegrino2021}
M.~Pellegrino, D.~Scheeres, and B.~J. Streetman.
\newblock The feasibility of targeting chaotic regions in the {GNSS} regime.
\newblock {\em The Journal of the Astronautical Sciences}, 68:553--584, 2021.

\bibitem{eP91}
E.~Perozzi, A.~E. Roy, B.~A. Steves, and G.~B. Valsecchi.
\newblock Significant high number commensurabilities in the main lunar problem.
  {I}: The {S}aros as a near periodicity of the {M}oon's orbit.
\newblock {\em Celestial Mechanics and Dynamical Astronomy}, 52:241--261, 1991.

\bibitem{Radtke2015}
J.~Radtke, R.~Dom\'{i}nguez-Gonz\'{a}lez, S.~K. Flegel, N.~S\'{a}nchez-Ortiz,
  and K.~Merz.
\newblock Impact of eccentricity build-up and graveyard disposal strategies on
  {MEO} navigation constellations.
\newblock {\em Advances in Space Research}, 56:2626--2644, 2015.

\bibitem{Rosengren2015}
A.~J. Rosengren, E.~M. Alessi, A.~Rossi, and G.~B. Valsecchi.
\newblock Chaos in navigation satellite orbits caused by the perturbed motion
  of the {M}oon.
\newblock {\em Monthly Notices of the Royal Astronomical Society},
  449:3522--3526, 2015.

\bibitem{Rossi2008}
A.~Rossi.
\newblock Resonant dynamics of medium earth orbits: space debris issues.
\newblock {\em Celestial Mechanics and Dynamical Astronomy}, 100:267--286,
  2008.

\bibitem{aR73}
A.~E. Roy.
\newblock The use of the saros in lunar dynamical studies.
\newblock {\em The Moon}, 7:6--13, 1973.

\bibitem{Szebehely}
V.~G. Szebehely.
\newblock {\em {Theory of orbits : the restricted problem of three bodies}}.
\newblock Academic Press, New York [etc.], 1967.

\end{thebibliography}

\end{document}